\documentclass{llncs}

\usepackage{tabto}
\usepackage{amssymb}
\usepackage{smartref} 
\usepackage{fullpage}
\usepackage{listings}
\usepackage{paralist}
\usepackage{pgf}

\usepackage{tabularx}
\newcolumntype{Y}{>{\centering\arraybackslash} X}

\newcommand{\comnospace}{\mbox{$\triangleright$}}
\newcommand{\com}{\mbox{\comnospace\ }}

\newcommand{\func}[1]{\mbox{\sc #1}}

\newcommand{\cas}{\mbox{CAS}}

\newcommand{\true}{\textsc{True}}
\newcommand{\false}{\textsc{False}}

\newcommand{\op}{\mbox{\sct -record}}
\newcommand{\rec}{\mbox{Data-record}}

\newcommand{\listrec}{\mbox{Node}}

\newcommand{\llresults}{\info Fields}
\newcommand{\info}{\textit{info}}

\newcommand{\help}{\func{Help}}
\newcommand{\validate}{\vlt}
\newcommand{\llt}{\func{LLX}}

\newcommand{\sct}{\func{SCX}}
\newcommand{\vlt}{\func{VLX}}
\newcommand{\del}{\func{Delete}}
\newcommand{\ins}{\func{Insert}}

\newcommand{\search}{\func{Search}}

\newcommand{\freezing}{\mbox{InProgress}}
\newcommand{\retry}{\mbox{Aborted}}
\newcommand{\done}{\mbox{Committed}}

\newcommand{\freezingdone}{allFrozen}

\newcommand{\fail}{\textsc{Fail}}
\newcommand{\finalized}{\textsc{Finalized}}

\newcommand{\fcas}{{freezing~\cas}}
\newcommand{\astep}{{abort step}}
\newcommand{\fstep}{{frozen step}}
\newcommand{\fcstep}{{frozen check step}}
\newcommand{\cstep}{{commit step}}
\newcommand{\upcas}{{update~\cas}}
\newcommand{\markstep}{{mark step}}

\newcommand{\nil}{\textsc{Nil}}

\newcommand{\presctfld}{(3)}
\newcommand{\presctinfo}{(2)}

\newcommand{\presctlinked}{1}
\newcommand{\presctabainit}{2}
\newcommand{\presctaba}{3}

\def\pwidth{2.5cm}
\def\indentamt{0mm}

\newcommand{\sidecom}[1]{\tabto{12.3cm} \com \mbox{#1}}
\newcommand{\dline}[2]{#1\\ \mbox{\hspace{\indentamt}#2}}
\newcommand{\tline}[3]{\dline{#1}{#2}\\ \mbox{\hspace{\indentamt}#3}}

\lstdefinestyle{nonumbers}{numbers=none}
\newcommand{\preplisting}{\lstset{gobble=1, numbers=left, numberstyle=\tiny, numberblanklines=false, numbersep=-8.5pt, firstnumber=last, escapeinside={//}{\^^M}, breaklines=true, keywordstyle=\bfseries, morekeywords={type,subtype,break,if,else,end,loop,while,do,done,exit, when,then,return,read,and,or,not,,for,each,boolean,procedure,invoke,next,iteration,until}}}
\newcommand{\prepnewlisting}{\lstset{gobble=1, numbers=left, numberstyle=\tiny, numberblanklines=false, numbersep=-8.5pt, firstnumber=1, escapeinside={//}{\^^M}, breaklines=true, keywordstyle=\bfseries, morekeywords={type,subtype,break,if,else,end,loop,while,do,done,exit, when,then,return,read,and,or,not,for,each,boolean,procedure,invoke,next,iteration,until}}}

\newtheorem{thm}{Theorem}
\newtheorem{obs}[thm]{Observation}
\newtheorem{lem}[thm]{Lemma}
\newtheorem{cor}[thm]{Corollary}
\newtheorem{con}[thm]{Constraint}
\newtheorem{defn}[thm]{Definition}
\newtheorem{prop}[thm]{Proposition}
\newtheorem{inv}[thm]{Invariant}

\newcommand{\after}[1]{}

\begin{document}

\title{Pragmatic Primitives for Non-blocking Data Structures}
\author{Trevor Brown\inst{1} \and Faith Ellen\inst{1} \and Eric Ruppert\inst{2}}
\institute{
    University of Toronto, Department of Computer Science \\
    27 King's College Cir, Toronto, ON M5S, Canada \\
    \email{\{tabrown,faith\}@cs.toronto.edu}
    \and
    York University, Dept. of Elec. Eng. and Comp. Sci. \\
    4700 Keele St, Toronto, ON M3J 1P3, Canada \\
    \email{ruppert@cs.yorku.ca}
}

\maketitle

\begin{abstract}
We define a new set of primitive operations that greatly simplify
the implementation of  
non-blocking data structures in asynchronous shared-memory systems.
The new operations operate on a set of \rec s, each of which contains
multiple fields.  The operations are generalizations of the well-known
load-link (LL) and store-conditional (SC) operations called \llt\ and \sct. 
The \llt\ operation takes a snapshot of one \rec.  An \sct\ operation
by a process $p$ succeeds 
only if no \rec\ in a specified set 
has been changed since $p$ last performed an \llt\ on it.
If successful, the \sct\ atomically
updates one specific field of a
\rec\
in the set and prevents any future changes to some
specified subset of those \rec s.
We provide a provably correct implementation of these new primitives from single-word
compare-and-swap.
As a simple example, we 
show how to
implement a non-blocking multiset data structure in a 
straightforward way using \llt\ and \sct.
\end{abstract}

\section{Introduction}

Building a library of
concurrent data structures 
is an essential way to simplify the difficult task of developing concurrent software. 
There are many lock-based data structures,
but locks 
are not fault-tolerant and are susceptible to
problems such as deadlock \cite{Fra07}.
It is often preferable to use hardware synchronization primitives like compare-and-swap 
(\textsc{CAS}) instead of locks.  However, the difficulty of this task has inhibited the development of 
{\it non-blocking} data structures.
These are data structures
which guarantee that some operation will eventually complete even if some processes crash.

Our goal is to facilitate the implementation of
high-per\-form\-ance, provably correct, non-blocking data structures  on any system that supports a hardware \textsc{CAS} instruction.
We introduce three new operations,
\textit{load-link-extended} (\llt), \textit{validate-extended} (\vlt) and \textit{store-conditional-extended} (\sct), which are natural generalizations of the well known \textit{load-link} (\textsc{LL}), \textit{validate} (\textsc{VL}) and \textit{store-conditional} (\textsc{SC}) operations.
We  provide a practical implementation of our new operations
from \textsc{CAS}.
Complete proofs of correctness
appear in Appendix~\ref{sec-proof}.
We also show how
these operations
make the implementation of non-blocking 
data structures 
and their proofs of correctness
substantially less difficult, as compared to using \textsc{LL},
\textsc{VL}, \textsc{SC}, and \textsc{CAS} directly.

\llt, \sct\ and \vlt\ operate on {\it \rec s}.
Any number of types of \rec s can be defined, each type containing a fixed number of {\it mutable} fields (which can be updated), and a fixed number of {\it immutable} fields (which cannot).
Each \rec\ can represent a natural unit of a data structure, such as a node of 
a tree or a table entry.
A successful \llt\ operation returns a snapshot of the mutable fields of one \rec. (The immutable fields can be read directly, since they never change.)
An \sct\ operation by a process $p$ is used to atomically
store a value in
one mutable field of one \rec\ 
{\it and}
{\it finalize} a set of \rec s, meaning that those \rec s cannot undergo any further changes. 
The \sct\ succeeds only if each \rec\ in a specified set has not
changed since $p$ last performed an $\llt$ on it.
A successful \vlt\ on a set of \rec s simply assures the caller that each of these
\rec s has not changed since the caller last performed an \llt\ on it. 
A more formal specification of the
behaviour of these operations is given in Section \ref{sec-operations}.

Early on, researchers recognized that operations 
accessing multiple locations atomically make
the design of non-blocking data structures much easier \cite{Barnes,IR94,ST97}.
Our new primitives 
do this
in three ways.  First, they operate
on \rec s, rather than individual words,
to allow the data structure designer to 
think at a higher level of abstraction.  
Second, and more importantly, a \vlt\ or \sct\ can depend upon
multiple \llt s.
Finally, the effect of an \sct\ can apply to multiple \rec s,
modifying one and finalizing others.

The precise specification of our operations was  chosen to balance ease of use
and efficient implementability.
They are more restricted than multi-word CAS \cite{IR94}, multi-word RMW \cite{AMTT97}, or 
transactional memory \cite{ST97}.
On the other hand, the ability to finalize \rec s makes \sct\ more 
general than $k$-compare-single-swap \cite{LMS09},
which can only change one word.
We found that atomically changing one pointer and finalizing a collection of \rec s provides 
just enough power
to implement numerous pointer-based data structures in which operations replace a small
portion of the data structure.
To demonstrate the usefulness of our new operations,
in Section \ref{sec-multiset}, we give
an implementation of
a simple, linearizable, non-blocking multiset based on an ordered,
singly-linked list.

Our implementation of
\llt, \vlt, and \sct\ 
is designed for an asynchronous system where processes may crash.
We assume shared memory locations can be accessed by single-word CAS, read and write instructions.
We assume a safe garbage collector
(as in the Java environment) that will not reallocate a memory location if any process can reach it by following pointers.
This allows records to be reused. 

Our implementation has some desirable performance properties.
A \vlt\ on $k$ \rec s only requires reading $k$ words of memory.
If \sct s being performed concurrently depend on \llt s
of disjoint sets of \rec s, they all succeed.
If an \sct\ encounters no contention with any other \sct\ and finalizes  $f$ \rec s, then
a total of $k+1$ CAS steps and $f+2$ writes are used for
the \sct\ and the $k$ \llt s on which it depends.
We also prove progress properties that suffice for building
non-blocking data structures using \llt\ and \sct.

\section{Related work} \label{sec-rel}

Transactional memory \cite{HM93,ST97} is a
general approach to simplifying the design of concurrent algorithms by
providing atomic access to multiple objects.
It allows
a block of code
designated as a transaction to be executed 
atomically, with respect to other transactions.
Our \llt/\vlt/\sct\ primitives may be viewed as
implementing a restricted kind of transaction, in which 
each transaction can perform any number of reads
followed by a single write and then finalize any number of words.
It is possible to implement general transactional memory in a non-blocking manner (e.g., \cite{Fra07,ST97}).
However, at present, 
implementations of transactional memory in software incur significant overhead,
so there is still a need for more specialized techniques 
for designing
shared data structures
that combine ease of use and efficiency.

Most shared-memory systems  provide CAS operations in hardware.
However, LL and SC operations have often been seen
as more convenient primitives for building algorithms.
Anderson and Moir gave the first
wait-free implementation of small LL/SC objects from CAS 
using $O(1)$ steps per operation~\cite{AM95}.
See \cite{JP05:opodis} for a survey of 
other implementations that use less space
or handle larger LL/SC objects.    

Many non-blocking implementations of primitives that access multiple objects
use the {\it cooperative technique}, first described by Turek, Shasha and Prakash \cite{TSP92}
and Barnes \cite{Barnes}.
Instead of using locks that give a process exclusive access
to a part of the data structure, this approach gives exclusive access to {\it operations}.  If the
process performing an operation that holds a lock is slow, other processes can {\it help}
complete the operation and release the lock.

The cooperative technique was also used recently
for a wait-free universal construction \cite{CER10} and
to obtain non-block\-ing binary
search trees \cite{EFRB10:podc} and Patricia tries \cite{Shafiei}.  
The approach used here is similar.

Israeli and Rappoport \cite{IR94} used a version of the cooperative technique to implement multi-word CAS from single-word CAS (and sketched how this could be used to implement multi-word SC operations).
However, their approach applies single-word CAS to very large words.
The most efficient implementation of $k$-word CAS \cite{Sun11}
first uses single-word
CAS to replace each of the $k$ words with a pointer to a record
containing information about the operation, and then uses 
single-word CAS to replace each of
these pointers with the desired new value and update the status field
of the record.
In the absence of contention, this
takes $2k+1$ CAS steps.
In contrast, in our implementation, an \sct\ that depends on \llt s
of $k$ \rec s performs $k+1$ single-word CAS steps when there is
no contention, no matter how many words each record contains.
So, our weaker primitives can be significantly more
efficient 
than multi-word CAS or multi-word RMW \cite{AMTT97,AH11},
which is even more general.

If $k$ \rec s are removed from a data structure by a multi-word CAS,
then the multi-word CAS must depend on every mutable field
of these records to prevent another process from concurrently
updating any of them.
It is possible to use $k$-word CAS to apply to $k$ \rec s instead
of $k$ words with indirection:
Every \rec\ is represented by a single word containing a pointer to
the contents of the record.
To change any fields of the \rec, a process swings the pointer to a
new copy of its contents containing the updated values.
However, the extra level of indirection affects all reads,
slowing them down considerably.

Luchangco, Moir and Shavit \cite{LMS09} defined the $k$-compare-single-swap (KCSS) primitive,
which atomically tests
wheth\-er $k$ specified memory locations contain specified values and, if all tests succeed,
writes a value to one of the locations.
They provided an {\it obstruction-free} implementation of KCSS,  meaning
that  a process
performing a KCSS is guaranteed to terminate if it runs alone.
They implemented KCSS using an obstruction-free implementation of LL/SC
from CAS.
Specifically, to try to update location $v$ using KCSS, a process performs
LL($v$), followed by two collects of the other $k-1$ memory locations.
If $v$ has its specified value,
both collects return their specified values, and the contents of these
memory locations do not change between the two collects, the process
performs SC to change the value of $v$.
Unbounded version numbers are used both in their implementation
of LL/SC and to avoid the ABA problem between the two collects.

Our \llt\ and \sct\ primitives can be viewed as multi-\rec-LL and
single-\rec-SC primitives, with the additional power to finalize \rec s.
We shall see that this extra ability
is  extremely useful for implementing 
pointer-based data structures.
In addition, our implementation of \llt\ and \sct\ allows us to develop shared
data structures that satisfy the non-blocking progress condition, which is stronger
than obstruction-freedom.

\section{The primitives} \label{sec-operations}

Our primitives operate on a collection of \rec s of various user-defined types.
Each type of \rec\ has a fixed number of mutable fields (each fitting into a single word), and a fixed number of immutable fields (each of which can be large).  Each field is given a value when the \rec\ is created.  Fields can contain pointers that refer to other \rec s.  
\rec s are accessed using \llt, \sct\ and \validate,
and reads of individual mutable or immutable fields of a \rec.
Reads of mutable fields are permitted because a snapshot of a \rec's fields is sometimes excessive, 
and it is sometimes sufficient (and more efficient) to use reads instead of \llt s.

An implementation 
of LL and SC from \cas\ has to ensure that,
between when a process performs LL and when it next performs SC
on the same word, the value of the word has not changed.
Because the value of the word could change and then change back
to a previous value, it is not sufficient to check that the word
has the same value when the LL and the SC are performed.
This is known as the ABA problem.
It also arises for implementations of \llt\ and \sct\ from \cas.
A general technique to overcome this  problem is described
in Section~\ref{sec-impl-aba}.
However, if the data structure designer can guarantee that the 
ABA problem will not arise
(because each \sct\ never attempts to store a value into
a field that previously contained that value),
our implementation can be used in a more efficient manner.

Before giving the precise specifications of the behaviour of \llt\ and \sct,
we describe how to use them,
with the implementation of a multiset as a running example.
The multiset abstract data type supports three operations: \func{Get}$(key)$,
which returns the number of occurrences of $key$ in the multiset, \func{Insert}$(key, count)$, which inserts $count$ occurrences of $key$ into the multiset, and \func{Delete}$(key, count)$, which deletes $count$ occurrences of $key$ from the multiset and returns \true, provided there are at least $count$ occurrences of $key$ in the multiset. Otherwise, it simply returns \false. 

Suppose we would like to implement a multiset using a sorted, singly-linked list.
We represent each node in the list by a \rec \ with an immutable field $key$, which contains a key in the multiset, and mutable fields: $count$, which records the number of times $key$ appears in the multiset, and $next$, which points to the next node in the list.
The first and last elements of the list are sentinel nodes
with count 0 and
with special keys $-\infty$ and $\infty$, respectively, which never occur in the multiset.

Figure~\ref{fig-example-multiset} shows how updates to the list are handled.  
Insertion behaves differently depending
on whether the key is already present.
Likewise, deletion behaves differently depending
on whether it removes all copies of the key.
For example,
consider the operation \func{Delete}$(d, 2)$ depicted in Figure~\ref{fig-example-multiset}(c).
This operation removes node $r$ by changing $p.next$ 
to point to a new copy of $rnext$.
A new copy is used to avoid the ABA problem, since $p.next$
may have pointed to $rnext$ in the past.
To perform the \func{Delete}$(d,2)$,
a process first invokes \llt s on $p$, $r$, and $rnext$. 
Second, it creates a copy $rnext'$ of $rnext$.
Finally, it performs an \sct\ that depends on these three \llt s.  This \sct\ attempts to change $p.next$ to point to $rnext'$.
This \sct \ will succeed only if none of $p$, $r$ or $rnext$ have changed since the aforementioned \llt s.
Once $r$ and $rnext$ are removed from the list, we want subsequent invocations of \llt \ and \sct \ to be able to detect this, so that we can avoid, for example, erroneously inserting a key into a deleted part of the list.
Thus, we specify in our invocation of \sct \ that $r$ and $rnext$ should be \textit{finalized} if the \sct \ succeeds.
Once a \rec\ is finalized, it can never be changed again.

\llt\ takes (a pointer to) a \rec\ $r$ as its argument.  Ordinarily, it returns either a snapshot of $r$'s mutable fields or \finalized. 
If an \llt$(r)$ is concurrent with an \sct\ involving $r$, it is also allowed to fail
and return \fail.
\sct\ takes four arguments:
a sequence $V$ of (pointers to) \rec s upon which the \sct\ depends, 
a subsequence $R$ of $V$ containing (pointers to) the \rec s to be finalized,
a mutable field $fld$ of 
a \rec\ in $V$ to be modified, and a value $new$ to store in this field.
\vlt\ takes a sequence $V$ of (pointers to) \rec s as its only argument.
Each \sct\ and \vlt\ and returns a Boolean value.
 
For example, in Figure~\ref{fig-example-multiset}(c), the \func{Delete}($d,2$) operation invokes \sct($V, R, fld, new$), where $V = \langle p, r, rnext \rangle$, $R = \langle r, rnext\rangle$, $fld$ is the next pointer of $p$, and $new$  points to the 
node $rnext'$.

A terminating \llt\  is called {\it successful} if it returns a snapshot or \finalized, and {\it unsuccessful} if it returns \fail.
A terminating \sct\ or \vlt\ is called {\it successful} if it returns \true, and {\it unsuccessful} if it returns \false.
Our operations are wait-free, but an operation may not terminate if the process performing
it fails, in which case the operation is neither successful nor unsuccessful.
We say an invocation $I$ of \llt$(r)$ by a process $p$
 is \textit{linked to} an invocation $I'$ of \sct$(V, R, fld, new)$ or \vlt$(V)$ by process $p$ if  $r$ is in $V$,
$I$ 
returns a snapshot, and between $I$ and $I'$, process $p$ performs
no invocation of \llt$(r)$ or \sct$(V', R', fld', new')$
and no unsuccessful invocation of \vlt$(V')$, for any $V'$ that
contains $r$. 
Before invoking \vlt$(V)$ or \sct$(V, R, fld, new)$, a process must
{\it set up} the operation by performing an \llt$(r)$ linked to the invocation
for each $r$ in $V$.

\subsection{Correctness Properties}
An implementation of \llt, \sct\ and \vlt\ is {\it correct} if,
for every execution, there is a linearization of all successful \llt s,
all successful \sct s, a subset of the non-terminating \sct s,
all successful \vlt s, and all reads, such that the following conditions
are satisfied.
\begin{compactenum}[{\bf C\arabic{enumi}}:]
	\item 
		Each read of a field $f$ of a \rec\ $r$ returns
		the last value stored in $f$ by an \sct\ linearized before the read
		(or $f$'s initial value, if no such \sct\ has modified $f$).
	\item 
Each linearized \llt($r$) that does not return \finalized\
		returns	the last value stored in each mutable field $f$ of $r$
		by an \sct\ linearized before the \llt\ (or $f$'s initial value, if no such \sct\ has modified~$f$).
	\item 
		Each linearized \llt$(r)$ returns \finalized\ if and only if it is linearized
		after an 
		\sct($V, R, fld,$ $new$) with $r$ in $R$.
	\item 
For each linearized invocation $I$ of \sct($V, R,$ $fld, new$) or \vlt$(V)$,
and for each $r$ in $V$, 
no \sct($V'$, $R'$, $fld'$, $new'$) with $r$ in $V'$ is linearized between the
\llt$(r)$ linked to $I$ and $I$.
\end{compactenum}

The first three properties assert that successful reads and \llt s return
correct answers.
The last property says that an invocation of \sct\ or \vlt\
does not succeed when it should not.
However, an \sct\ can fail if it is concurrent with another \sct\ 
that accesses some \rec\ in common.
LL/SC also exhibits analogous failures in real systems.
Our progress properties limit
the situations in which
this
can occur.

\subsection{Progress Properties} 
\label{progress-spec}

In our implementation, \llt, \sct \ and \vlt\ are technically wait-free, but this is only because they may fail.
The first progress property guarantees that 
\llt s on finalized \rec s succeed.
\begin{compactenum}[{\bf P\arabic{enumi}}:]
	\item 
        Each terminating \llt$(r)$ returns \finalized\ if it begins after the end of a successful \sct$(V,$ $R,$ $fld,$ $new)$ with $r$ in $R$ or after another \llt$(r)$ has returned \finalized. 
\end{compactenum}
The next progress property guarantees non-blocking progress of invocations of our primitives.
\begin{compactenum}[{\bf P\arabic{enumi}}:]
\setcounter{enumi}{1}
\item If operations are performed infinitely often, then operations succeed infinitely often.
\end{compactenum}
However,
this progress property
leaves open the possibility that only \llt s succeed.
So, we
want
an additional progress property:
\begin{compactenum}[{\bf P\arabic{enumi}}:]
\setcounter{enumi}{2}
\item If \sct\ and \vlt\ operations are performed infinitely often, then \sct\ or \vlt\ operations succeed infinitely often.
\end{compactenum}
Finally, the following progress property ensures that {\it update} operations that are built using
\sct\ can be made non-blocking.
\begin{compactenum}[{\bf P\arabic{enumi}}:]
\setcounter{enumi}{3}
\item If \sct\ operations are performed infinitely often, then \sct\ operations succeed infinitely often.
\end{compactenum}

When the progress properties defined here are used to prove that an application built
from the primitives is
non-blocking, there is an important, but subtle point:
an \sct\ can be invoked only after it has been properly set up by a sequence of \llt s.
However, if processes repeatedly perform \llt\ on \rec s that have been finalized,
they may never be able to invoke an \sct.
One way to prevent this from happening is to have
each process keep track of the
\rec s it knows are finalized.
However, in many natural applications, for example, the multiset
implementation in Section \ref{sec-multiset}, explicit bookkeeping can be avoided.
In addition, to ensure that changes to a data structure can continue to occur,
there must always be at least one 
non-finalized
\rec. 
For example, in our multiset, $head$ is never finalized and,
if a node is reachable from $head$ by following $next$ pointers,
then it is not finalized.

Our implementation of \llt, \sct\ and \vlt\ in Section \ref{sec-impl}
actually satisfies stronger progress properties
than the ones described above.  For example, a \vlt($V$) or
\sct($V,R,fld$, $new$)
is guaranteed to succeed if there is no concurrent \sct( $V',R',fld',new')$
such that $V$ and $V'$ have one or more elements in common.
However, for the purposes of the specification of the primitives, 
we decided to give progress guarantees that are sufficient to prove that algorithms that
use the primitives are non-blocking, but weak enough that it may be possible to design
other, even more efficient implementations of the primitives.  For example, our
specification would allow some spurious failures of the type that occur in common
implementations of ordinary LL/SC operations (as long as there is some 
guarantee that not all operations can fail spuriously).

\section{Implementation of Primitives} \label{sec-impl}

\newcommand{\wcnarrow}[2]{\parbox{\namewidth}{#1} \com \mbox{#2}}
\begin{figure}[tb]
\def\namewidth{18mm}
\preplisting
\small
\begin{lstlisting}[mathescape=true,style=nonumbers]
 type $\rec$
   //\com User-defined fields 
   //\wcnarrow{$m_1, \ldots, m_y$}{mutable fields}
   //\wcnarrow{$i_1, \ldots, i_z$}{immutable fields}
   //\com Fields used by  \llt /\sct\ algorithm
   //\wcnarrow{$\info$}{pointer to an \op}
   //\wcnarrow{$marked$}{Boolean}
\end{lstlisting}
\def\namewidth{17mm}
\preplisting
\begin{lstlisting}[mathescape=true,style=nonumbers]
 type $\op$
   //\wcnarrow{$V$}{sequence of \rec s}
   //\wcnarrow{$R$}{subsequence of $V$ to be finalized}
   //\wcnarrow{$fld$}{pointer to a field of a \rec\ in $V$}
   //\wcnarrow{$new$}{value to be written into the field $fld$}
   //\wcnarrow{$old$}{value previously read from the field $fld$} 
   //\wcnarrow{$state$}{one of \{\freezing, \done, \retry\}}
   //\wcnarrow{$\freezingdone$}{Boolean}
   //\wcnarrow{$\llresults$}{sequence of pointers, one read from the}
   //\wcnarrow{\mbox{ }}{\info\ field of each element of $V$}
 \end{lstlisting}
	\caption{Type definitions for shared objects used to implement \llt, \sct, and \vlt.
		}
	\label{code1}
\end{figure}

The shared data structure used
to implement \llt, \sct\ and \vlt\ 
consists of a set of \rec s and a set of \op s. (See Figure~\ref{code1}.)  
Each \rec\ contains  user-defined mutable and immutable fields.  It also contains
a 
$marked$ bit, which is used to finalize the \rec,
and an $\info$ field.
The marked bit is initially \false\ and only ever changes from
\false\ to \true.
The $\info$ field points to an \op\ that describes the last \sct\ that accessed the \rec.
Initially, it points to a {\it dummy} \op.
When an \sct\ accesses a \rec, it changes the $\info$ field of the \rec\ to point to its \op.
While this \sct\ is  active, the $\info$ field acts as a kind of lock on the \rec,
granting exclusive access to this \sct, rather than to a process.
(To avoid confusion, we call this {\it freezing}, rather than locking, a \rec.)
We ensure that an \sct\ $S$ does not change a \rec\ for its own purposes while it is 
frozen for another \sct\ $S'$.
Instead, $S$
uses the information in the \op\ of $S'$ to help  $S'$ complete 
(successfully or unsuccessfully),
so that the \rec\ can be unfrozen. 
This cooperative approach is used to ensure progress.


An \op\ contains enough information to allow any process to complete an \sct\ operation 
that is in progress. 
$V, R, fld$ and $new$ store the arguments of the \sct\ operation that created the \op.
Recall that $R$ is a subsequence of $V$ and $fld$ points to a mutable field $f$ of some \rec\ $r'$ in $V$.
The value that was read from $f$ by the \llt$(r')$ linked to the \sct\ is 
stored in $old$.
The \op\ has one of three states, \freezing, \done\ or \retry,
which is stored in its $state$ field. This field is
initially \freezing.
The \op\ of each \sct\ that terminates is eventually set
to \done\ or \retry, depending on whether or not it successfully
makes its desired update.
The dummy \op\ always has $state$ = \retry.
The $\freezingdone$ bit, which is initially \false, gets 
set to \true\ after all \rec s in $V$ have been frozen for the \sct.
The values of $state$ and $\freezingdone$ change in accordance with the diagram in
Figure~\ref{fig-state-transitions}.
The steps in the 
algorithm
that cause these changes are also 
indicated.
The $\llresults$ field stores, for each $r$ in $V$, the value of $r$'s \info\ field that was
read by the \llt$(r)$ linked to the \sct.


\begin{figure}[t]
	\begin{minipage}{\textwidth}
        \centering
    	\includegraphics[scale=0.065]{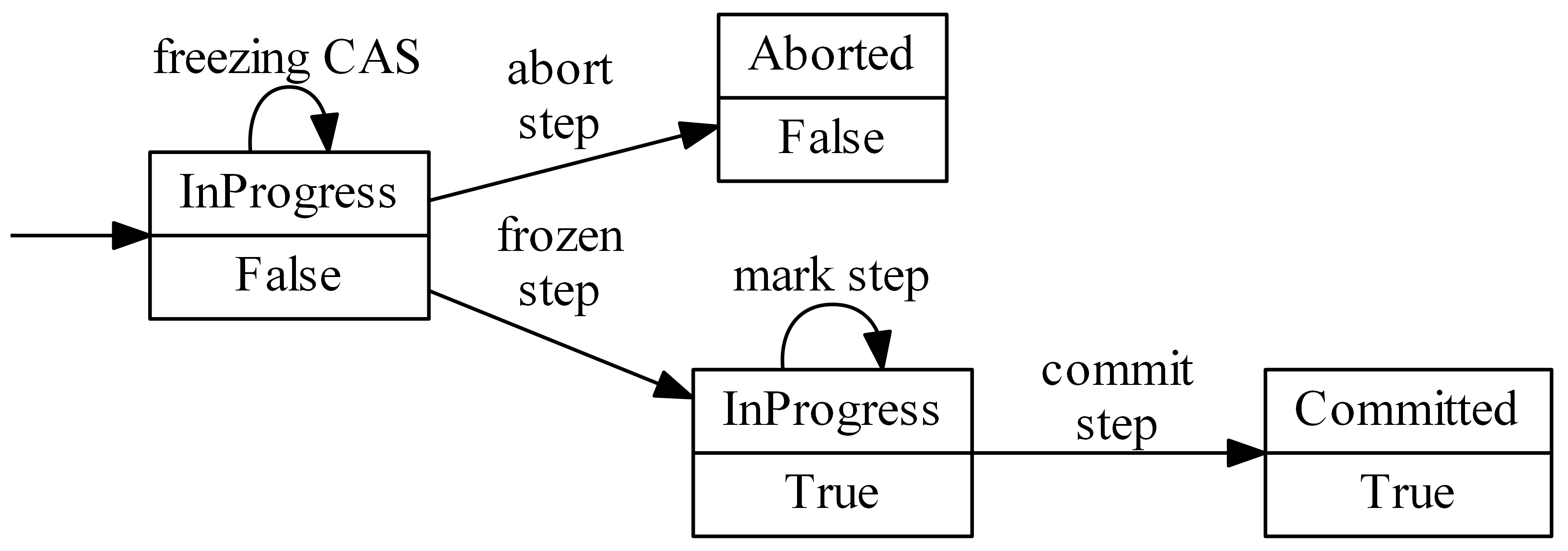} \\
    \end{minipage}
\caption{Possible 
[$state$, $\freezingdone$] field transitions
of an \op.}
\label{fig-state-transitions}
\end{figure}

We say that a \rec\ $r$ is \textit{marked} when $r.marked = \true$.
A \rec\ $r$ is \textit{frozen} for an \op\ $U$ if $r.\info$ points to $U$ and either $U.state$ is \freezing, or $U.state$ is \done\ and $r$ is marked. 
While a \rec\ $r$ is frozen for an \op\ $U$, a mutable field $f$ of $r$ can be changed 
only if $f$ is the field pointed to by $U.fld$ (and it can only be changed by a process 
helping the \sct\ that created $U$).
Once a \rec\ $r$ is marked and $r.\info.state$ becomes \done, $r$
 will never be modified again in any way.
Figure~\ref{fig-state-transitions2} shows how 
a \rec\ can change between frozen and unfrozen.
The three bold boxes represent frozen \rec s.
The other two boxes represent \rec s that are not frozen.
A \rec\ $r$ can only become frozen when $r.\info$ is changed (to point to a new \op\ whose state is \freezing). This is represented by the grey edges.
The black edges represent changes to $r.\info.state$ or $r.marked$.
A frozen \rec\ $r$ can only become unfrozen when $r.\info.state$ is changed.

\begin{figure}[h]
	\begin{minipage}{\textwidth}
        \centering
    	\includegraphics[scale=0.065]{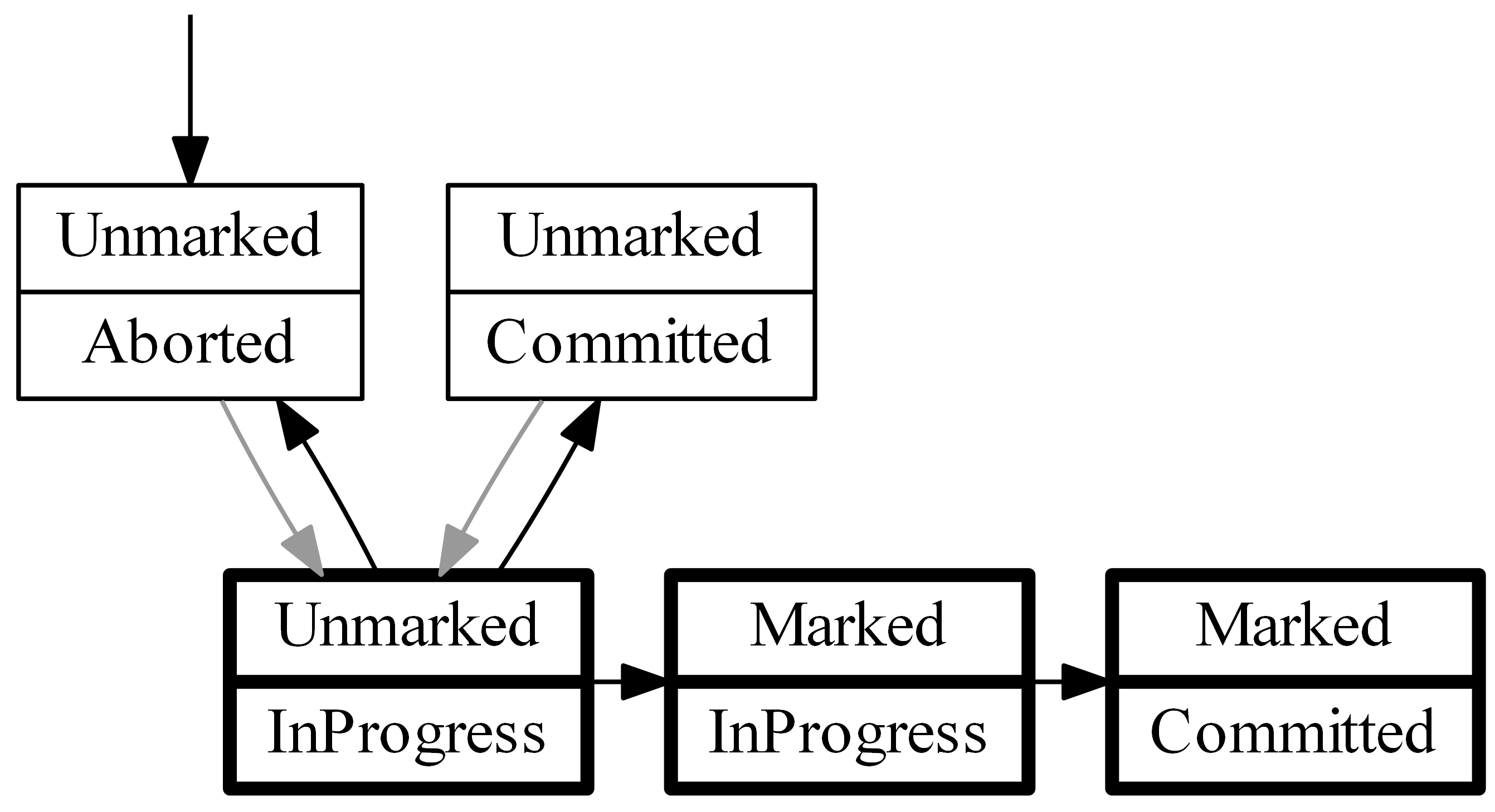} \\
    \end{minipage}
\caption{
Possible transitions for the
$marked$ field of a \rec\ and the $state$ of the \op\ pointed to by
the {\it info} field of the \rec .}
\label{fig-state-transitions2}
\end{figure}

\subsection{Constraints} 
\label{constraints}
\label{sec-impl-aba}

For the sake of efficiency, we have designed our implementation of \llt, \vlt\ and \sct\ to work
only if the primitives are used in a way that satisfies certain constraints, described in this section.  We also describe general (but somewhat inefficient) ways to ensure these constraints
are satisfied.
However, there are often quite natural ways to ensure
the constraints are satisfied
without resorting to the extra work required by the general solutions.

Since our implementation of \llt, \sct\ and \vlt\ uses helping to guarantee progress, each \cas\ of an \sct\ might be repeatedly performed by several helpers, possibly after the \sct\ itself has terminated.
To avoid difficulties, we must show there is no ABA problem in the fields affected by these CAS steps.

The \info\ field of a \rec\ $r$ is modified by \cas\ steps that attempt to freeze
$r$ for an \sct.
All such steps performed by processes helping one invocation of 
\sct\ try to \cas\ the \info\ field of $r$
from the same old value to the same new value, and that new value is 
a pointer to a newly created \op.  Because the
\op\ is allocated a location that has 
never been used before, the ABA problem will not arise in the \info\ field.
(This approach is compatible with safe garbage collection schemes that 
only reuse an old address once no process can reach it by following pointers.)
 
A similar approach could be used to avoid the ABA problem
in a mutable field
of a \rec:
the new value
could be placed inside a wrapper object that is allocated a new
location in memory.  (This is referred to as Solution 3 of the ABA problem in \cite{DPS10}.) 
However, the extra level of indirection slows down accesses to fields.

To avoid the ABA problem, it suffices to prove the following constraint is satisfied.
\begin{compactitem}
\item
{\bf Constraint}: For every invocation $S$
of  \sct$(V, R$, $fld, new)$,
$new$ is not the initial value of $fld$ and no invocation of
\sct$(V', R', fld, new)$  was linearized before the $\llt(r)$ linked to
$S$ was linearized,  where $r$ is the \rec\ that contains $fld$.
\end{compactitem}
The multiset in Section~\ref{sec-multiset} provides an example
of a simple, more efficient way  to
ensure that this constraint is always satisfied.

To ensure property P4,
we put a constraint on the way \sct\ is used.
Our implementation of \sct($V, R, fld$, $new$)
does something akin to acquiring locks on each \rec\ in $V$.
Livelock could occur if different invocations of \sct \ do not process \rec s in the same order.
To prevent this, we could define a way of ordering all \rec s (for example,
by their locations in memory)
and 
each sequence
passed to an invocation of  \sct\ could  be sorted using this ordering. 
However,
this could be expensive.
Moreover, to prove our progress properties,
we do not require that {\it all} \sct s order their
sequences $V$ consistently.
It suffices that,
if all the \rec s stop changing, then 
the sequences passed to later invocations of \sct\
are all consistent with some total order.
This property is often easy to satisfy
in a natural way.
More precisely, use of our implementation of \sct\ requires adherence to the following constraint.
\begin{compactitem}
\item
{\bf Constraint}:
Consider each execution that contains
a configuration $C$ after which the value of no
field of any \rec\ changes.
There must be a total order
on all \rec s created during this execution such that,
if \rec\ $r_1$ appears before
\rec\ $r_2$ in the sequence $V$ passed to an invocation 
of \sct\ whose linked \llt s begin after $C$,
then $r_1 < r_2$.
\end{compactitem}
For example, if one was using \llt\ and \sct\ to implement
an {\it unsorted} singly-linked list,
this constraint would be satisfied if the 
nodes in each sequence $V$ occur 
in the order they are encountered by following next pointers
from the beginning of the list, {\it even if} some operations
could reorder the nodes in the list.
While the list is changing, such a sequence may have repeated elements
and might not be consistent with any total order.

\subsection{Detailed Algorithm Description\\and Sketch of Proofs}


\begin{figure*}[tbp]
\def\pwidth{4cm}
\prepnewlisting
\hrule
\small
\vspace{-2mm}
\begin{lstlisting}[mathescape=true]
    //\llt$(r)$ by process $p$
    //\com Precondition: $r \neq \nil$.
      $marked_1 := r.marked$ // \label{ll-read-marked1} \sidecom{order of lines~\ref{ll-read-marked1}--\ref{ll-read-marked2} matters}
      $r\info := r.\info$ // \label{ll-read} 
      $state := r\info.state$ // \label{ll-read-state}
      $marked_2 := r.marked$ // \label{ll-read-marked2}
      if $state = \retry$ or $(state = \done$ and not $marked_2)$ then //  \label{ll-check-frozen} \sidecom{if $r$ was not frozen at line~\ref{ll-read-state}}
        read $r.m_1,...,r.m_y$ //and record the values in local variables $m_1,...,m_y$%
        \label{ll-collect}
        if $r.\info = r\info$ then//\label{ll-reread}\sidecom{if $r.\info$ points to the same} 
          //store $\langle r, r\info, \langle m_1, ..., m_y \rangle \rangle$ in $p$'s local table %
\sidecom{\op\ as on line~\ref{ll-read}}\label{ll-store}
          return $\langle m_1, ..., m_y \rangle$ // \label{ll-return}  \vspace{2mm}
      if ($r\info.state = \done$ or ($r\info.state = \freezing$ and $\help(r\info)))$ and $marked_1$ then//%
      \label{ll-check-finalized}
        return $\finalized$ // \label{ll-return-finalized}
      else
        if $r.\info.state = \freezing$ then $\help(r.\info)$ // \label{ll-help-fail} 
        return $\fail$ // \label{ll-return-fail}\vspace{2mm} \hrule %
\vspace{2mm}    
    //\sct$(V, R, fld, new)$ by process $p$
    //\tline{\com Preconditions: (\presctlinked) for each $r$ in $V$, $p$ has performed an invocation $I_r$ of \llt$(r)$ linked to this \sct}%
            {\hspace{19.5mm}(\presctabainit) $new$ is not the initial value of $fld$}%
            {\hspace{19.5mm}(\presctaba) for each $r$ in $V$, no $\sct(V', R', fld, new)$ was linearized before $I_r$ was linearized}%
%            {\hspace{19.5mm}(\presctdistinct) if, for each $r \in V$, no $\sct$ has been linearized since $I_r$, then $V$'s elements are distinct}

      //\dline{Let $\llresults$ be a pointer to a newly created table in shared memory containing,}%
              {for each $r$ in $V$, a copy of $r$'s \info\ value in $p$'s local table of \llt\ results}%
              \label{sct-create-llresults}
      //Let $old$ be the value for $fld$ stored in $p$'s local table of \llt\ results\label{sct-create-old}
      return $\help(\mbox{pointer to new \op} (V, R, fld, new, old, \freezing,  \false, \llresults ))$ // \label{sct-create-op}\label{sct-call-help} \vspace{2mm} \hrule %
\vspace{2mm}    
    //\help$(scxPtr)$ 
      //\com \mbox{Freeze all \rec s in $scxPtr.V$ to protect their mutable fields from being changed by other \sct s}
      for each $r$ in $scxPtr.V \mbox{enumerated in order}$ do
        //Let $r$\info\ be the pointer indexed by $r$ in $scxPtr.\llresults$ \label{help-rinfo}
        if not $\cas(r.\info,r\info,scxPtr)$ then //\sidecom{\textbf{\fcas}}\label{help-fcas}
          if $r.\info \neq scxPtr$ then // \label{help-check-frozen} 
            //\com \mbox{Could not freeze $r$ because it is frozen for another \sct}
            if $scxPtr.\freezingdone = \true$ then//\sidecom{\textbf{\fcstep}}\label{help-fcstep}
              //\com the \sct\ has already completed successfully 
              return $\true$ // \label{help-return-true-loop} 
            else
              //\com Atomically unfreeze all nodes frozen for this \sct 
              $scxPtr.state := \retry$ //\sidecom{\textbf{\astep}}\label{help-astep}
              return $\false$ // \label{help-return-false} \vspace{2mm}
      //\com Finished freezing \rec s (Assert: $state \in \{\freezing, \done\}$) 
      $scxPtr.\freezingdone := \true$//\sidecom{\textbf{\fstep}}\label{help-fstep}
      for each $r$ in $scxPtr.R$ do $r.marked := \true$ //\sidecom{\textbf{\markstep}}\label{help-markstep}
      //$\cas(scxPtr.fld, scxPtr.old, scxPtr.new)$ \sidecom{\textbf{\upcas}}\label{help-upcas} \vspace{2mm}
      //\com Finalize all $r$ in $R$, and unfreeze all $r$ in $V$ that are not in $R$ 
      $scxPtr.state := \done$//\sidecom{\textbf{\cstep}}\label{help-cstep}
      return $\true$ // \label{help-return-true} \vspace{2mm} \hrule %
\vspace{2mm} 
    //\validate$(V)$ by process $p$ 
    //\mbox{\com Precondition: for each \rec\ $r$ in $V$, $p$ has performed an \llt$(r)$ linked to this \vlt}
      for each $r$ in $V$ do
        //Let $r\info$ be the \info\ field for $r$ stored in $p$'s local table of \llt\ results\label{vlt-info} 
        if $r\info \neq r.\info$ then return $\false$ //\tabto{8cm}\mbox{\com $r$ changed since \llt$(r)$ read $\info$}%
        \label{vlt-reread} 
      return $\true$ //\tabto{4.5cm}\mbox{\com At some point during the loop, all $r$ in $V$ were unchanged} %
      \vspace{2mm} \hrule %
\end{lstlisting}
    \vspace{-3mm}
	\caption{Pseudocode for \llt, \sct\ and \validate.}
	\label{code-main}
\end{figure*}

Pseudocode for our implementation of \llt, \vlt\ and \sct\ appears in Figure~\ref{code-main}. 
If $x$ contains a pointer to a record, then $x.y := v$
assigns the value $v$ to field $y$ of this record,
\&$x.y$ denotes the address of this field and all other
occurrences of $x.y$ denote the value stored in this field.
\begin{thm}
The algorithms in Figure \ref{code-main} satisfy properties C1 to C4 and P1 to P4 in every execution where the constraints of Section \ref{constraints} are satisfied.
\end{thm}
The detailed proof of correctness 
in Appendix \ref{sec-proof}
is quite involved, so we only sketch the main ideas here.


An \llt$(r)$ returns a snapshot, \fail, or \finalized.
At a high level, it works as follows.
If the \llt\ determines that $r$ is not frozen and $r$'s $\info$ field does not change
while the \llt\ reads the mutable fields of $r$, 
the \llt\ returns the values read as a snapshot.
Otherwise, the \llt\ helps the \sct\ that
last froze $r$, if it is frozen,
and returns \fail\ or \finalized.
If the \llt\ returns \fail, it is not linearized.
We now discuss in more detail how \llt\ operates and is linearized in the other two cases.

First, suppose the \llt($r$) returns a snapshot at line~\ref{ll-return}.
Then, the test at line~\ref{ll-check-frozen} evaluates to \true.  So,
either $state=\retry$, which means $r$ is not frozen at line~\ref{ll-read-state},
or $state=\done$ and $marked_2=\false$. This also means $r$ is not frozen at 
line~\ref{ll-read-state}, since $r.marked$ cannot change from \true\ to \false.
The \llt\ reads $r$'s mutable fields (line~\ref{ll-collect}) and rereads
$r.\info$ at line~\ref{ll-reread}, finding it
the same as on 
line~\ref{ll-read}.
In Section~\ref{sec-impl-aba}, we
explained 
why this implies that
$r.\info$ did not change between lines~\ref{ll-read} and \ref{ll-reread}.
Since $r$ is not frozen at line~\ref{ll-read-state}, we know from Figure~\ref{fig-state-transitions2} that $r$ is unfrozen at all times between line~\ref{ll-read-state} and~\ref{ll-reread}.
We prove that mutable fields can change only while $r$ is frozen,
so the values read by line \ref{ll-collect}
constitute a snapshot of $r$'s mutable fields.  Thus, we can linearize the 
\llt\ at line~\ref{ll-reread}.

Now, suppose the \llt($r$) returns \finalized.  Then, 
the test on line~\ref{ll-check-finalized} evaluated to \true.
In particular, 
$r$ was already marked when
line~\ref{ll-read-marked1} was performed.
If {\it rinfo.state} = \freezing\ when line~\ref{ll-check-finalized}
was performed, $\help(r\info)$ was called and returned \true.
Below, we argue that {\it rinfo.state} was changed to \done\
before the return occurred.
By Figure~\ref{fig-state-transitions2}(a), the $state$ of an \op\ never changes after 
it is set to \done. 
So, after line~\ref{ll-check-finalized}, 
{\it rinfo.state} = \done\ and, thus, $r$ has been finalized.
Hence, the \llt\ can be linearized at line \ref{ll-return-finalized}.

When a process performs an \sct, it first creates a new \op\ and then invokes \help\ (line~\ref{sct-create-op}).
The \help\ routine performs the real work of the \sct. It is also used by a process to help other processes
complete their \sct s (successfully or unsuccessfully).
The values in an \op's $old$ and $\llresults$
come from a table in the local memory of the process that invokes the \sct,
which stores the results of the last \llt\ it performed on each \rec.  (In
practice, the memory required for this table could be greatly reduced
when a process knows which of these values
are needed for future \sct s.)

Consider an invocation of \help($U$) by process $p$ to carry out the work of the invocation 
$S$ of \sct($V,$ $R,$ $fld,$ $new$) that is described by the \op\ $U$.
First, $p$
attempts to freeze each $r$ in $V$ by performing a {\it \fcas} to store a pointer to $U$ in $r.\info$ (line~\ref{help-fcas}).
Process $p$ uses the value read from $r.\info$ by the $\llt(r)$ linked to $S$ 
as the old value for this CAS
and, hence, it will succeed only if $r$ has not been frozen for any other \sct\ since then.
If $p$'s \fcas\ fails, it checks whether some other helper has successfully frozen
the \rec\ with a pointer to $U$ (line~\ref{help-check-frozen}).

If every $r$ in $V$ is successfully frozen, $p$ performs a {\it \fstep} to set $U.\freezingdone$ to 
\true\ (line~\ref{help-fstep}).
After this \fstep, the \sct\ is guaranteed not to fail, meaning that no process
will  perform an \astep\ while
helping this 
\sct.
Then, for each $r$ in $R$,
$p$ performs a {\it \markstep} to set $r.marked$ to \true\ 
(line~\ref{help-markstep}) and,
from Figure~\ref{fig-state-transitions2},
$r$ remains frozen from then on.
Next, $p$ performs an {\it \upcas}, storing $new$ in the field pointed to by $fld$ (line~\ref{help-upcas}), if successful.
We prove that, among all the \upcas\ steps on $fld$ performed by the helpers of $U$,
only the first can succeed.
Finally, $p$ unfreezes all $r$ in $V$ that are not in $R$ by performing a {\it \cstep} that changes $U.state$ to \done\ (line~\ref{help-cstep}).

Now suppose that, when $p$ performs line~\ref{help-check-frozen}, it
finds that some \rec\ $r$ in $V$ is already frozen for another invocation $S'$ of \sct.
If $U.\freezingdone$ is \false\ at line~\ref{help-fcstep},
then we can prove that
no helper of $S$ will ever reach line~\ref{help-fstep}, so $p$ can
abort $S$.  To do so,
it unfreezes each $r$ in $V$ that it has frozen by performing an {\it \astep}, which changes $U.state$ to \retry\ (line~\ref{help-astep}), and then returns \false\ (line~\ref{help-return-false}) to indicate that $S$ has been aborted.
If $U.\freezingdone$ is \true\ at line~\ref{help-fcstep}, 
it means that each element of $V$, including $r$, was
successfully frozen by some helper of $S$ and then, later, a process
froze $r$ for $S'$.
Since $S$ cannot be aborted after $U.\freezingdone$ was set to \true,
its state must have changed from \freezing\ to \done\ before $r$
was frozen for another \op.
Therefore, $S$ was successfully completed and
$p$ can return \true\ at line~\ref{help-return-true-loop}.

We linearize an invocation of \sct\ at the first \upcas\ performed
by one of its helpers.
We prove that this \upcas\ always succeeds.
Thus, all \sct s that return \true\ are linearized,
as well as possibly some non-terminating \sct s.
The first \upcas\ of  \sct($V,R,fld,new$) modifies the value of $fld$, so
a read($fld$) that occurs immediately after the \upcas\ will return
the value of $new$.
Hence, the linearization point of an \sct\ must occur at its first \upcas.
There is one subtle issue about this linearization point: 
If an \llt($r$)  is linearized 
between the \upcas\ and \cstep\ of an \sct\ that finalizes $r$,
it might not return \finalized, violating condition C3.
However, this cannot happen, because,
before the \llt\ is linearized on line \ref{ll-return-finalized},
the \llt\ either
sees that the \cstep\ has been performed or helps the \sct\ perform its \cstep .

An invocation $I$ of \vlt$(V)$ is executed by a process $p$ after $p$ has performed an invocation of \llt$(r)$ linked to $I$, for each $r$ in $V$.
\vlt$(V)$ simply checks, for each $r$ in $V$, that the \info\ field of $r$ is the same as when it was read by $p$'s last
\llt$(r)$ and, if so, \vlt$(V)$ returns \true.
In this case, we prove that each \rec\ in $V$ does not change between the linked \llt\ and the time its \info\ field is reread.  Thus, the \vlt\ can be linearized at the first time it executes  line \ref{vlt-reread}.
Otherwise, the \vlt\ returns \false\ to indicate that the \llt\ results may not constitute a snapshot.

We remark that our use of the cooperative method avoids costly recursive helping.
If, while $p$ is helping $S$, it cannot freeze all of $S$'s \rec s because one of them is already frozen for a third \sct, then $p$ will simply  perform an \astep,  which unfreezes all \rec s that $S$ has frozen.


\begin{figure*}[tb]
	\centering
    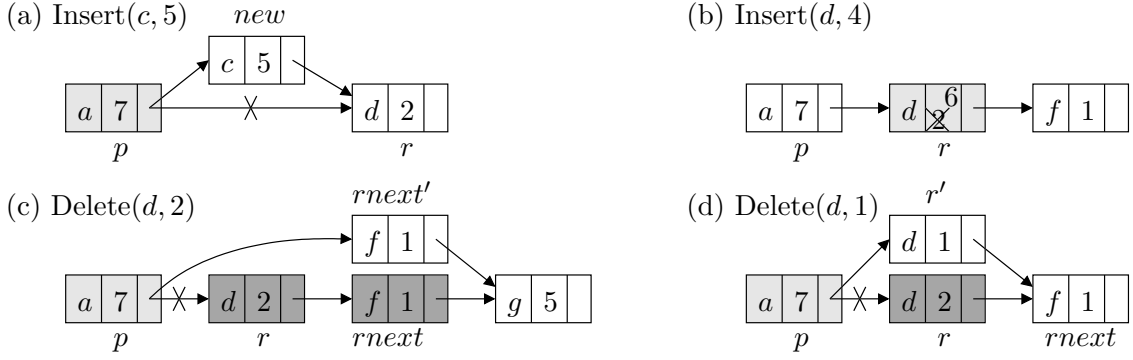
	\caption{Using \sct\ to update a multiset.  \llt s of all shaded nodes are linked to the \sct. Darkly shaded nodes are finalized by the \sct.
Where a field has changed, the old value is crossed out.}
	\label{fig-example-multiset}
\end{figure*}


We briefly sketch why the progress properties 
described in Section~\ref{progress-spec} are satisfied.
It follows easily from the code that
an invocation of \llt($r$) returns \finalized\ 
if it begins after the end of an \sct\ that finalized $r$ or another \llt\ sees that $r$ is finalized.
To prove the progress properties P2, P3 and P4, we consider two cases.

First, consider an execution where only a finite number of \sct s are invoked.
Then, only finitely many \op s are created.
Each process calls \help($U$) if it sees that $U.state = \freezing$,
which it can do at most 
once for each \op\ $U$.
Since every \cas\ is performed inside the \help\ routine, there is some point after which
no process performs a \cas,
calls \help, or sees a \op\ whose $state$
is \freezing.
A \vlt\ can fail only when an \info\ field is modified by a concurrent operation and an \llt\ can only fail for the same reason or when
it sees a \op\ whose $state$ is \freezing.
Therefore, all \llt s and \vlt s  that begin after this point
will succeed,
establishing P2 and P3.
Moreover,
P4 is vacuously satisfied in this case.

Now, consider an execution where infinitely many \sct s are invoked.
To derive a contradiction, suppose only finitely many \sct s succeed.
Then, there is a time after which no more \sct s succeed.
The constraint on the sequences passed to invocations of \sct s ensures that
all \sct s whose linked \llt s begin after this time will
attempt to freeze their sequences of \rec s in a consistent order.
Thus, one of these \sct s will succeed in freezing all of the \rec s
that were passed to it and will
successfully complete.
This is a contradiction.
Thus, infinitely many of the \sct s do succeed, establishing properties P2, P3 and P4.

\subsection{Additional Properties}

Our implementation of \sct\ satisfies some additional properties, which are helpful
for designing certain kinds of non-blocking data structures so that query operations
can run efficiently.
Consider a pointer-based data structure with a fixed set of \rec s called 
{\it entry points}.
An operation on the data structure starts at an entry point and follows pointers to
visit other \rec s.  (For example, in our multiset example, the head of the linked list
is the sole entry point for the data structure.)
We say that a \rec\ is {\it in the data structure} if it can be reached by following
pointers from an entry point, and a \rec\ $r$ is {\it removed from the data structure}
by an \sct\ if $r$ is in the data structure immediately prior to the
linearization point of the \sct\ and is not in the data structure immediately afterwards.

If the data structure is designed so that a \rec\ is finalized when (and only when) 
it is removed from the data structure, then we have the following additional properties.
\begin{prop}
\label{searches-work}
Suppose each linearized \sct$(V,R,fld$, $new)$ removes precisely the \rec s in $R$ from the data structure.
\begin{compactitem}
\item
If \llt$(r)$ returns a value different from \fail\ or \finalized, $r$ is in the data structure just before the \llt\ is linearized.
\item
If an \sct$(V,R,fld,new)$ is linearized and $new$ is (a pointer to) a \rec, then 
this \rec\ 
is in the data structure just after the \sct\ is linearized.
\item
If an operation reaches a \rec\ $r$ by following pointers read from other \rec s, starting from an entry point, then $r$ was in the data structure at some earlier time during the operation.
\end{compactitem}
\end{prop}

The first two properties are straightforward to prove.
The last property is proved by induction on the \rec s reached.
For the base case, entry points are always reachable.  For the induction step,
consider the time when an operation reads a pointer to $r$ from another 
\rec\ $r'$ that the operation
reached earlier.  By the induction hypothesis, there was an earlier time $t$ during the operation
when $r'$ was in the
data structure.  If $r'$ already contained a pointer to $r$ at $t$, then $r$ was also
in the data structure at that time.  Otherwise, an \sct\ wrote a pointer to $r$ in $r'$ after $t$,
and just after that update occurred, $r'$ and $r$ were in the data structure 
(by the second part of the proposition).

The last property is a particularly useful one for linearizing query operations.  It means that 
operations that search through a data structure can use simple reads of pointers 
instead of the more expensive \llt\ operations.  
Even though the \rec\ that such a search operation reaches
may have been removed from the data structure by the time it is reached, the lemma guarantees
that there {\it was} a time during the search when the \rec\ was in the data structure.
For example, we use this property to linearize searches in our multiset algorithm in Section \ref{sec-multiset}.

\section{An example: multiset} \label{sec-multiset}

We now give a detailed description of the implementation of a multiset
using \llt\ and \sct.
We assume that keys stored in the multiset are drawn from a totally ordered set and $-\infty < k < \infty$ for every key $k$ in the multiset.
As described in Section \ref{sec-operations}, we use a singly-linked
list of nodes, sorted by key.
To avoid special cases, it always has a sentinel node, $head$,
with key $-\infty$ at its beginning and a sentinel node, $tail$,
with key $\infty$ at its end.
The definition of \listrec, the \rec\ used to represent a node,
and the pseudocode are presented in
Figure~\ref{code-list}.

\func{Search}$(key)$ traverses the list, starting from 
$head$,
by reading $next$ pointers until reaching the first node $r$
whose key is at least $key$.
This node and the preceding node $p$ 
are returned.
\func{Get}$(key)$ performs \func{Search}($key$), outputs $r$'s count
if $r$'s key matches $key$, and outputs 0, otherwise.

\begin{figure*}[tb]
\hrule
\begin{minipage}{0.365\textwidth}
\prepnewlisting
\def\namewidth{12mm}
\vspace{-1mm}
\scriptsize
\begin{lstlisting}[mathescape=true,style=nonumbers]
 type// \listrec
    //\com Fields from sequential data structure
    //\wcnarrow{$key$}{key (immutable)}
    //\wcnarrow{$count$}{occurrences of $key$ (mutable)}
    //\wcnarrow{$next$}{next pointer (mutable)}
    //\com Fields defined by \llt /\sct\ algorithm
    //\wcnarrow{$\info$}{a pointer to an \op}
    //\wcnarrow{$marked$}{a Boolean value} \vspace{-0.5mm}
          
 //\textbf{shared} \listrec\ $tail := \mbox{new }\listrec ( \infty, 0, \nil )$
 //\textbf{shared} \listrec\ $head := \mbox{new }\listrec ( -\infty, 0, tail )$
 \end{lstlisting}%
 \vspace{-1mm}
\prepnewlisting
\hrule
\begin{lstlisting}[mathescape=true]
    //\func{Get}$(key)$
      $\langle r, - \rangle := \search(key)$
      if $key = r.key$ then
        return $r.count$ 
      else return $0$// \vspace{2mm} \hrule %
\vspace{2mm}   
    //\search$(key)$
      //\com Postcondition: $p$ and $r$ point to \listrec s with $p.key < key \le r.key$.
      $p := head$ // \label{multiset-search-p}
      $r := p.next$
      while $key > r.key$ do // \label{multiset-search-loop}
        $p := r$ // \label{multiset-search-advance-p}
        $r := r.next$ // \label{multiset-search-advance-r}
      return $\langle r, p \rangle$ // %\vspace{2mm} \hrule
\end{lstlisting}
\vspace{-2mm}
\end{minipage}
\begin{minipage}{0.635\textwidth}
\preplisting
\scriptsize
\begin{lstlisting}[mathescape=true]
    //\ins$(key, count)$ \tabto{4cm} \com Precondition: $count > 0$
      while $\true$ do
        $\langle r, p \rangle := \search(key)$ // \label{multiset-insert-search}
        if $key = r.key$ then // \label{multiset-insert-check-key}
          $localr := \llt(r)$//\label{multiset-insert-llt-r}
          if $localr \notin \{\fail, \finalized\}$ then 
            //\mbox{\textbf{if} $\sct$($\langle r \rangle, \langle \rangle, \&r.count,$ $localr.count + count$) \textbf{then return}}\label{multiset-insert-sct1}
        else
          $localp := \llt(p)$//\label{multiset-insert-llt-p}
          if $localp \notin \{\fail, \finalized\}$ and $r = localp.next$ then
            //\mbox{\textbf{if} $\sct$($\langle p \rangle, \langle \rangle, \&p.next,$ $\mbox{new}\listrec(key, count, r)$) \textbf{then return}} \label{multiset-insert-sct2} \vspace{2mm} \hrule \vspace{2mm}
    //\del$(key, count)$ \tabto{4cm} \com Precondition: $count > 0$
      while $\true$ do
        $\langle r, p \rangle := \search(key)$ // \label{multiset-delete-search}
        $localp := \llt(p)$//\label{multiset-delete-llt-p}
        $localr := \llt(r)$//\label{multiset-delete-llt-r}
        if $localp, localr \notin \{\fail, \finalized\}$ and $r = localp.next$ then
          if $key \neq r.key$ or $localr.count < count$ then return $\false$//\label{multiset-delete-return-false}
          else if $localr.count > count$ then // \label{multiset-delete-check-count-less}
            if $\sct(\langle p \rangle, \langle r \rangle, \&p.next, \mbox{new }$ $\listrec( r.key, localr.count - count,$ $localr.next))$ then return $\true$//\label{multiset-delete-sct1}
          else //\com assert: $localr.count = count$  \label{multiset-delete-check-count-else}
            if $\llt(localr.next) \notin \{\fail, \finalized\}$ then//\label{multiset-delete-llt-rnext}
              if $\sct(\langle p,r,localr.next \rangle,$ $\langle r,localr.next\rangle,$ $\&p.next, \mbox{new copy of } localr.next)$//\mbox{\textbf{ then return} $\true$}\label{multiset-delete-sct2}
\end{lstlisting}
\end{minipage}
\hrule
\vspace{-2mm}
	\caption{Pseudocode for a multiset, implemented with a singly linked list. 
	}
	\label{code-list}
\end{figure*}


An invocation $I$ of \func{Insert}($key,count)$
starts by calling \func{Search}$(key)$.  Using the nodes $p$ and $r$
that are returned, it updates the data structure.
It decides whether $key$ is already in the multiset
(by checking whether $r.key = key$) and, if so,
it invokes $\llt(r)$ followed by an \sct\ linked to $r$ to
increase $r.count$ by $count$, as depicted in
Figure~\ref{fig-example-multiset}(b).
Otherwise, $I$ performs the update depicted in
Figure~\ref{fig-example-multiset}(a):
It invokes $\llt(p)$, checks that $p$ still points to $r$,
creates a node, $new$, and invokes an \sct\ linked to $p$
to insert $new$ between $p$ and $r$.
If $p$ no longer points to $r$, the \llt\ 
returns \fail\ or \finalized, or the \sct\ returns \false,
then $I$ restarts.

An invocation $I$ of \func{Delete}($key,count)$ also begins by calling
\func{Search}$(key)$.  It invokes \llt\ on the nodes $p$ and $r$
and then checks that $p$ still points to $r$. If
$r$ does not contain at least $count$ copies of $key$, then $I$ returns \false.
If $r$ contains exactly $count$ copies,
then $I$ performs the update depicted in Figure~\ref{fig-example-multiset}(c)
to remove node $r$ from the list.
To do so, it invokes \llt\ on the node, $rnext$, that $r.next$ points to,
makes a copy $rnext'$ of $rnext$,
and invokes an \sct\ linked to $p, r$ and $rnext$
to change $p.next$ to point to $rnext'$.
This \sct\ also finalizes the nodes $r$ and $rnext$,
which are thereby removed from the data structure.
The node $rnext$ is replaced by a copy to
avoid the ABA problem in $p.next$.
If $r$ contains more than $count$ copies,
then $I$ replaces $r$ by a new copy $r'$ with an appropriately reduced count
using an \sct\ linked to $p$ and $r$,
as shown in Figure~\ref{fig-example-multiset}(d).
This \sct\ finalizes $r$.
If
an
\llt\ returns \fail\ or \finalized, or the \sct\ returns \false\,
then $I$ restarts.


A detailed proof of correctness appears in Appendix~\ref{app-multiset}. 
It begins by showing that this multiset implementation
satisfies some basic properties.
\begin{inv} The following are true at all times.
\begin{compactitem}
\item $head$ always points to a node.
\item If a node has key $\infty$, then its $next$ pointer is $\nil$.
\item If a node's key is not $\infty$, then its $next$ pointer
points to some node with a strictly larger key.
\end{compactitem}
\end{inv}
It follows that the data structure is always a sorted list.

We prove the following lemma by considering
the \sct s performed by update operations shown in Figure~\ref{fig-example-multiset}.
\begin{lem}
\label{R-sets-correct}
The \rec s removed from the data structure by
a linearized invocation of \sct$(V$, $R$, $fld$, $new)$
are exactly the \rec s in $R$.
\end{lem}
This lemma allows us to apply Proposition \ref{searches-work}
to prove that there is a time during each \search\ when the nodes $r$ 
and $p$ that it returns
are both in the list and $p.next = r$.

Each \func{Get} and each
\func{Delete} that returns \false\ is linearized at the linearization
point of the \func{Search} it performs. 
Every other \func{Insert} or \func{Delete}
is linearized at its successful \sct. 
Linearizability of all operations then follows from the next invariant.

\begin{lem}
At every time $t$, the multiset of keys
in the data structure is equal to the multiset of keys that would result
from the atomic execution of the sequence of operations linearized up to time $t$.
\end{lem}

To prove the algorithm is non-blocking, suppose there is some infinite execution
in which only finitely many operations terminate.
Then, eventually, no more \ins\ or \del\ operations 
perform a successful \sct, so there is a time after which the pointers that form the linked
list stop changing.
This implies that all calls to the \search\ subroutine must terminate.
Since a \func{Get} operation merely calls \search, all \func{Get} operations must also
terminate.
Thus, there is some collection of \ins\ and \del\ operations that take steps forever
without terminating.
We show that each such operation sets up and performs an \sct\ infinitely often.
For any  \ins\ or \del\ operation, consider any iteration of the loop
that begins after the last successful \sct\ changes the list.
By Lemma \ref{R-sets-correct} and Proposition \ref{searches-work}, the nodes $p$ and $r$
reached by the \search\ in that iteration
were in the data structure at some time during the \search\
and, hence, throughout the \search.
So when the \ins\ or \del\ performs \llt s on $p$ or $r$,
they cannot return \finalized.
Moreover, they must succeed infinitely often by property P2, and this allows the \ins\
or \del\ to perform an \sct\ infinitely often.  By property P4, \sct s will succeed infinitely often, a contradiction.

Thus, we have the following theorem.
\begin{thm}
The algorithms in Figure \ref{code-list} implement a non-blocking, linearizable multiset.
\end{thm}

\section{Conclusion} \label{sec-conclusion}

The \llt, \sct\ and \vlt\ primitives we introduce in this paper
can also be used to produce practical, non-blocking
implementations of a wide variety of tree-based data structures.
In \cite{paper2}, we describe a general method for obtaining such
implementations and use it to design a provably correct,
non-blocking implementation
of a chromatic tree, which is a relaxed variant of a red-black tree.
Furthermore, we provide an experimental performance analysis,
comparing our Java implementation of the
chromatic search tree to leading concurrent implementations of
dictionaries.
This demonstrates that our primitives enable efficient non-blocking
implementations of more complicated data structures to be built
(and added to standard libraries), together with manageable
proofs of their correctness.

Our implementation of  \llt, \sct\ and \vlt\ relies on the existence
of efficient garbage collection, which is provided in
managed languages such as Java and C\#.
However, in other languages, such as C++,
memory management is an issue. This can be addressed, for example, by the
new, efficient memory reclamation method of Aghazadeh, Golab and Woelfel \cite{AGW13}.


\paragraph{Acknowledgements}
Funding was provided by the Natural Sciences and Engineering Research Council of Canada.

\bibliographystyle{abbrv}
\bibliography{bibliography}

\newpage
\appendix

\section{Complete Proof} \label{sec-proof}

\subsection{Basic properties}

We begin with some elementary properties that are needed to prove basic lemmas about freezing.
In particular, we show that the $\info$ field of a \rec\ cannot experience an ABA problem.

\begin{defn} \label{defn-llt-linked-to-sct}
Let $I'$ be an invocation of \sct$(V, R, fld, new)$ or \vlt$(V)$ by a process $p$, and $r$ be a \rec\ in $V$.
We say an invocation $I$ of \llt$(r)$ is \textbf{linked to} $I'$ if and only if:
\begin{enumerate}
\item $I$ returns a value different from \fail\ or \finalized, and
\label{prop-returns-value-different-from-fail-or-finalized}
\item no invocation of \llt$(r)$, \sct$(V', R', fld', new')$, or \vlt$(V')$, where $V'$ contains $r$, is performed by $p$ between $I$ and $I'$.
\label{prop-no-sct-or-vlt-between-linked-llt-and-sct-or-vlt}
\end{enumerate}
\end{defn}

\begin{obs} \label{obs-op-invariants}
An \op\ $U$ created by an invocation $S$ of \sct\ satisfies the following invariants.
\begin{enumerate}
	\item{$U.fld$ points to a mutable field $f$ of a \rec\ $r'$ in $U.V$.}
			\label{inv-fld}
	\item{The value stored in $U.old$ was read at line~\ref{ll-collect}
			from $f$ by the \llt$(r')$ linked to $S$.}
			\label{inv-old}
	\item{For each $r$ in $U.V$, the pointer indexed by $r$ in $U.\llresults$
			was read from $r.\info$ at line~\ref{ll-read}
			by the \llt$(r)$ linked to $S$.}
			\label{inv-llresults}
	\item{For each $r$ in $U.V$,
			the \llt$(r)$ linked to $S$
			must enter the if-block
			at line~\ref{ll-check-frozen},
			and see $r.\info = r\info$
			at line~\ref{ll-reread}.}
			\label{inv-unfrozen}
\end{enumerate}
\end{obs}
\begin{proof}
None of the fields of an \op\ except $state$ change after they are initialized at line~\ref{sct-create-op}.  The contents of the table pointed to by $U.\llresults$ do not change, either.  Therefore, it suffices to show that these invariants hold when $u$ is created.
The proof of these invariants follows immediately from the precondition of \sct, the pseudocode of \llt, and the definition of an \llt\ linked to an \sct.
\end{proof}

The following two definitions ease discussion of the important steps that access shared memory.

\begin{defn} \label{defn-helping}
A process is said to be \textbf{helping} an invocation of \sct\ that created an \op\ $U$
whenever it is executing \help$(ptr)$, where $ptr$ points to $U$.
(For brevity, we sometimes say a process is ``helping $U$'' instead of ``helping the \sct\ that created $U$.'')
\end{defn}

Note that, since \help\ does not call itself directly or indirectly, a process cannot be helping two different invocations of \sct\ at the same time.

\begin{defn} \label{defn-belongs}
We say that a \fcas, \upcas, \fstep, \markstep, \astep, \cstep\ or \fcstep\ $S$ \textbf{belongs} to an \op\ $U$ when $S$ is performed by a process 
helping $U$.
We say that a \fstep, \markstep, \astep, \cstep\ or \fcstep\ is \textbf{successful} if it changes the the field it modifies to a different value.
A \fcas\ or \upcas\ is successful if the \cas\ succeeds.
Any step is \textbf{unsuccessful} if it is not successful.
\end{defn}


\begin{lem} \label{lem-no-steps-belong-to-dummy-op}
No \fcas, \upcas, \fstep, \markstep, \astep, \cstep\ or \fcstep\ belongs to the dummy \op.
\end{lem}
\begin{proof}
According to Definition~\ref{defn-belongs},
we must simply show that no process ever helps the dummy \op\ $D$.
To derive a contradiction, assume there is some invocation of \help$(ptr)$ where $ptr$ points to $D$, and let $H$ be the first such invocation.
\help\ is only invoked at lines~\ref{ll-check-finalized}, \ref{ll-help-fail} and \ref{sct-call-help}.

If $H$ occurs at line~\ref{ll-check-finalized}
then $D.state = \freezing$ at some point before $H$ (by line~\ref{ll-check-finalized}).
Since $D.state$ is initially \retry, and \freezing\ is never written into any $state$ field, this is impossible.

Now, suppose $H$ occurs at line~\ref{ll-help-fail}.
If $r.\info$ points to $D$ both times it is read at line~\ref{ll-help-fail}, then we obtain the same contradiction as the previous case.
Otherwise, a successful \fcas\ changes $r.\info$ to point to $D$ in between the two reads of $r.\info$.
By line~\ref{help-fcas}, this \fcas\ must occur in an invocation of \help$(ptr)$.
However, since this \fcas\ must precede the \textit{first} invocation, $H$, of \help$(ptr)$, this case is impossible.

$H$ cannot occur at line~\ref{sct-call-help}, since that line calls \help\ on a newly created \op, not $D$.
\end{proof}

\begin{lem} \label{lem-no-aba-info}
Every update to the $\info$ field of a \rec\ $r$ changes $r.\info$ to a value that has never previously appeared there. Hence, there is no ABA problem on $\info$ fields.
\end{lem}
\begin{proof}
We first note that $r.\info$ can only be changed by a \fcas\ at line~\ref{help-fcas}.
When a \fcas\ attempts to change an $\info$ field $f$ from $x$ to $y$, $y.\llresults$ contains $x$ (by 
line~\ref{help-rinfo}).  Then, since $y.\llresults$ does not change after the \op\ pointed to by $y$ is 
created, the \op\ pointed to by $x$ was created before the \op\ pointed to by $y$.
So, letting $a_1, a_2, ...$ be the sequence of \op s ever pointed to by $r.\info$, we know 
that $a_1, a_2, ...$ were created (at line~\ref{sct-create-op}) in that order.  Since we have assumed memory allocations always receive new addresses, $a_1, a_2, ...$ are distinct.
\end{proof}

\begin{defn}
A \fcas\ \textbf{on} a \rec\ $r$ is one that operates on $r.\info$.
A \markstep\ \textbf{on} a \rec\ $r$ is one that writes to $r.marked$.
\end{defn}

\begin{lem} \label{lem-only-first-fcas-can-succeed}
For each \rec\ $r$ in the $V$ sequence of an \op\ $U$, only the first \fcas\ belonging to $U$ on $r$ can succeed.
\end{lem}
\begin{proof}
Let $ptr$ be a pointer to $U$, and {\it fcas} be the first \fcas\ belonging to $U$ on $r$.  Let $r\info$ be the old value used by {\it fcas}.  By Definition~\ref{defn-belongs}, the new value used by {\it fcas} is $ptr$. 
Since {\it fcas} belongs to $U$, Lemma~\ref{lem-no-steps-belong-to-dummy-op} implies that $U$ is not the dummy \op\ initially pointed to by each \info\ field.
Hence, $U$ was created by an invocation of \sct, so Observation~\ref{obs-op-invariants}.\ref{inv-llresults} implies that $r.\info$ contained $r\info$ during the $\llt(r)$ linked to $S$.
Since the $\llt(r)$ linked to $S$ terminates before the start of $S$, and $S$ creates $U$, the $\llt(r)$ linked to $S$ must terminate before any invocation of \help$(ptr)$ begins.
From the code of \help, {\it fcas} occurs in an invocation of \help$(ptr)$.
Thus, $r.\info$ contains $r\info$ at some point before {\it fcas}.
If {\it fcas} is successful, then $r.\info$ contains $r\info$ just before {\it fcas}, and $ptr$ just after.
Otherwise, $r.\info$ contains $r\info$ at some point before {\it fcas}, but contains some other value just before {\it fcas}.
In either case, Lemma~\ref{lem-no-aba-info} implies that $r.\info$ can never again contain $r\info$ after {\it fcas}.
Finally, since each \fcas\ belonging to $U$ on $r$ uses $r\info$ as its old value (by line~\ref{help-rinfo} and the fact that table $U.\llresults$ does not change after it is first created), there can be no successful \fcas\ belonging to $U$ on $r$ after {\it fcas}.
\end{proof}

\subsection{Changes to the \info\ field of a \rec\ and the \textit{state} field of an \op}

We prove that freezing of nodes proceeds an orderly way.
The first lemma shows that a process cannot freeze a node that is frozen 
by a different operation that is still in progress.

\begin{lem} \label{lem-no-info-change-while-freezing}
The $\info$ field of a \rec\ $r$ cannot be changed while $r.\info$ points to an \op\ with $state$ \freezing.
\end{lem}
\begin{proof}
Suppose an $\info$ field of a \rec\ $r$ is changed while it points to an \op\ $U$ with $U.state =$ \freezing.
This change can only be performed by a successful \fcas\ {\it fcas} whose old value is a pointer to $U$ and whose new value is a pointer to $W$.
Let $S$ be the invocation of \sct\ that created $W$.
From line~\ref{help-rinfo}, we can see that the old value for {\it fcas} (a pointer to $U$) is stored in the table $W.\llresults$ and,
by Observation~\ref{obs-op-invariants}.\ref{inv-llresults},
this value was read from $r.\info$ (at line~\ref{ll-read}) by the \llt$(r)$ linked to $S$.
Hence, the \llt$(r)$ linked to $S$ reads $U.state$ at line~\ref{ll-read-state}.
By Observation~\ref{obs-op-invariants}.\ref{inv-unfrozen}, the \llt$(r)$ linked to $S$ passes the test at line~\ref{ll-check-frozen} and enters the if-block.
This implies that, when $U.state$ was read at line~\ref{ll-read-state}, either it was \done\ and $r$ was unmarked, or it was \retry.
Thus, $U.state$ must be \retry\ or \done\ prior to {\it fcas}, and the claim follows from the fact that \freezing\ is never written to $U.state$.
\end{proof}

It follows easily from Lemma \ref{lem-no-info-change-while-freezing} 
that if a node is frozen for an operation, it remains so until
the operation is committed or aborted.

\begin{lem} \label{lem-if-succ-fcas-then-point-u-until-bcas-or-uass}
If there is a successful \fcas\ {\it fcas} belonging to an \op\ $U$ on a \rec\ $r$, and some time $t$ after the first \fcas\ belonging to $U$ on $r$ and before the first \astep\ or \cstep\ belonging to $U$, then $r.\info$ points to $U$ at $t$.
\end{lem}
\begin{proof}
Since {\it fcas} belongs to $U$, by Lemma~\ref{lem-no-steps-belong-to-dummy-op}, $U$ cannot be the dummy \op, so $U$ is created at line~\ref{sct-create-op}, where $U.state$ is initially set to \freezing.
Let $t'$ be when the first \astep\ or \cstep\ belonging to $U$ on $r$ occurs.
By Lemma~\ref{lem-only-first-fcas-can-succeed}, {\it fcas} must be the first \fcas\ belonging to $U$ on $r$.
Thus, $t$ is after {\it fcas} occurs, and before $t'$.
Immediately following {\it fcas}, $r.\info$ points to $U$.
From the code, $U.state$ can only be changed by an \astep\ or \cstep\ belonging to $U$.
Therefore, $U.state =$ \freezing\ at all times after {\it fcas} and before $t'$.
By Lemma~\ref{lem-no-info-change-while-freezing}, $r.\info$ cannot change after {\it fcas}, and before $t'$.
Hence, $r.\info$ points to $U$ at $t$.
\end{proof}


A \fstep\ occurs only after all \rec s are successfully frozen.

\begin{lem} \label{lem-if-fass-then-all-succ-fcas}
If a \fstep\ belongs to an \op\ $U$ then, for each $r$ in $U.V$, there is a successful \fcas\ belonging to $U$ on $r$ that occurs before the first \fstep\ belonging to $U$.
\end{lem}
\begin{proof}
Suppose a \fstep\ belongs to $U$.  Let \textit{fstep} be the first such \fstep\ and let $H$ be the invocation of \help\ that performs \textit{fstep}.
Since \textit{fstep} occurs at line~\ref{help-fstep}, for each \rec\ $r$ in $U.V$, $H$ must perform a successful \fcas\ belonging to $U$ on $r$ or see $otherPtr = scxPtr$ in the preceding loop.
If $H$ performs a successful \fcas\ belonging to $U$ on $r$, then we are done.
Otherwise, $r.\info = scxPtr$ at some point before \textit{fstep}.
Since \textit{fstep} belongs to $U$, $scxPtr$ points to $U$.
From Lemma~\ref{lem-no-steps-belong-to-dummy-op}, $U$ is not the dummy \op\ to which $r.\info$ initially points.
Hence, some process must have changed $r.\info$ to point to $U$, which can only be done by a successful \fcas.
\end{proof}

Finally, we show that \rec s are frozen in the correct order.  (This will be useful later on to show that no livelock can occur.)


\begin{lem} \label{lem-fcas-on-ri-only-after-succ-fcas-on-previous}
Let $U$ be an \op, and $\langle r_1, r_2, ..., r_l\rangle$ be the sequence of \rec s in $U.V$.
For $i \geq 2$, a \fcas\ belonging to $U$ on \rec\ $r_i$
can occur only after a successful \fcas\ belonging to $U$ on $r_{i-1}$.
\end{lem}
\begin{proof}
Let {\it fcas} be a \fcas\ belonging to $U$ on $r_i$, for some $i \geq 2$.
Let $H$ be the invocation of \help\ which performs {\it fcas}.
The loop in $H$ iterates over the sequence
$r_1, r_2, ..., r_l$, so 
$H$ performs {\it fcas} in iteration $i$ of the loop.
Since $H$ reaches iteration $i$, $H$ must perform iteration $i-1$.
Thus, by the code of \help, $H$ must perform a \fcas\ {\it fcas$'$} belonging to $U$ on $r_{i-1}$ at line~\ref{help-fcas} before {\it fcas}.
If {\it fcas$'$} succeeds, then the claim is proved.
Otherwise, $H$
will check whether $r_{i-1}.\info$ is equal to $scxPtr$ at line~\ref{help-check-frozen}.
Since \help\ does not return in iteration $i-1$, $r_{i-1}.\info = scxPtr$.
This can only be true if $U$ is the dummy \op, or there has already been a successful \fcas \ belonging to $U$ on $r_{i-1}$.
Since {\it fcas} belongs to $U$, Lemma~\ref{lem-no-steps-belong-to-dummy-op} implies that $U$ cannot not be the dummy \op.
\end{proof}

\subsection{Proving $state$ and $info$ fields change as described in Fig.~\ref{fig-state-transitions-simple} and Fig.~\ref{fig-state-transitions}}

\begin{figure}[tb]
\centering
\mbox{
	\hspace{-1.5cm}
	\includegraphics[scale=0.5]{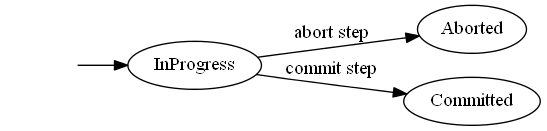}
}
\caption{Possible transitions for the $state$ field of an \op\ (initially \freezing).}
\label{fig-state-transitions-simple}
\end{figure}

In this section, we first prove that an \op's $state$ transitions respect Figure~\ref{fig-state-transitions-simple}, then we expand upon the information in Figure~\ref{fig-state-transitions-simple} by showing that \fstep s, \astep s, \cstep s and successful \fcas s proceed as illustrated in Figure~\ref{fig-state-transitions}.

We now prove an \op\ $U$'s $state$ transitions respect Figure~\ref{fig-state-transitions-simple} by noting that a $U.state$ is never changed to \freezing, and proving that it does not change from \retry\ to \done\ or vice versa.
Since $U.state$ can only be changed by a \cstep\ or \astep\ belonging to $U$, and each \cstep\ is preceded by a \fstep\ (by the code of \help), it suffices to show that there cannot be both a \fstep\ and an \astep\ belonging to an \op.


\begin{lem} \label{lem-if-all-succ-fcas-then-no-fcstep-until-fass}
Let $U$ be an \op.
Suppose that there is a successful \fcas\ belonging to $U$ on each \rec\ in $U.V$.
Then, a \fcstep\ belonging to $U$ cannot occur until after a \fstep\ belonging to $U$ has occurred.
\end{lem}
\begin{proof}
To derive a contradiction, suppose that the first \fcstep\ {\it fcstep} belonging to $U$ occurs before any \fstep\ belonging to $U$.
Let $H$ be the invocation of \help\ in which {\it fcstep} occurs.
%
Before {\it fcstep}, $H$ performs an unsuccessful \fcas\ at line~\ref{help-fcas} on one \rec\ in $U.V$.
The hypothesis of the lemma says that there is a successful \fcas, {\it fcas}, belonging to $U$ on this same \rec.
By Lemma~\ref{lem-only-first-fcas-can-succeed}, {\it fcas} must occur before $H$'s unsuccessful \fcas.
Thus, {\it fcas} occurs before {\it fcstep}, which occurs before any \fstep\ belonging to $U$.
From the code of \help, a \fstep\ belonging to $U$ precedes the first \cstep\ belonging to $U$, which implies that {\it fcas} and {\it fcstep} occur before the first \cstep\ belonging to $U$.
Further, the code of \help\ implies that no \astep\ can occur before {\it fcstep}.
Thus, {\it fcas} and {\it fcstep} occur strictly before the first \cstep\ or \astep\ belonging to $U$.
After {\it fcas}, and before {\it fcstep}, $H$ sees $r.\info \neq scxPtr$ at line~\ref{help-check-frozen}.
Since $scxPtr$ points to $U$, this implies that $r.\info$ points to some \op\ different from $U$.
However, Lemma~\ref{lem-if-succ-fcas-then-point-u-until-bcas-or-uass} implies that $r.\info$ must point to $U$ when $H$ performs line~\ref{help-check-frozen}, which is a contradiction.
\end{proof}

\begin{cor} \label{cor-first-fass-before-first-fcstep}
If a \fstep\ belongs to an \op\ $U$, then the first such \fstep\ must occur before any \fcstep\ belonging to $U$.
\end{cor}
\begin{proof}
Suppose a \fstep\ belongs to $U$.  By Lemma~\ref{lem-if-fass-then-all-succ-fcas}, we know there is a successful \fcas\ belonging to $U$ on $r$ for each $r$ in $U.V$.  Thus, by Lemma~\ref{lem-if-all-succ-fcas-then-no-fcstep-until-fass},
a \fcstep\ belonging to $U$ cannot occur until a \fstep\ belonging to $U$ has occurred.
\end{proof}

\begin{lem} \label{lem-fass-then-no-bcas}
There cannot be both a \fstep\ and an \astep\ belonging to the same \op.
\end{lem}
\begin{proof}
Suppose a \fstep\ belongs to $U$.  By
Corollary~\ref{cor-first-fass-before-first-fcstep},
the first such \fstep\ precedes the first \fcstep,
which, 
by the pseudocode of \help,
precedes the first \astep. 
This \fstep\ sets
$U.\freezingdone$
to \true\ and
$U.\freezingdone$
is never changed from \true\ to \false.
Therefore, any process that performs a \fcstep\ belonging to $U$
will immediately return \true, without performing an \astep.
Thus, there can be no \astep\ belonging to $U$.
\end{proof}

\begin{cor} \label{cor-no-change-from-done-retry}
An \op\ $U$'s $state$ cannot change from \done\ to \retry\ or from \retry\ to \done.
\end{cor}
\begin{proof}
Suppose $U.state =$ \done.
Then a \cstep\ belonging to $U$ must have occurred.
Since each \cstep\ is preceded by a \fstep,
Lemma~\ref{lem-fass-then-no-bcas} implies that no \astep\ belongs to $U$.
Thus, $U.state$ can never be set to \retry.

Now suppose that $U.state =$ \retry. Then either $U$ is the dummy \op\ or an \astep\ belongs to $U$.  If $U$ is the dummy \op\ then, by Lemma~\ref{lem-no-steps-belong-to-dummy-op}, no \cstep\ belongs to $U$,
so $U.state$ never changes to \done.
Otherwise, $U.state$ is initially \freezing, so there must have been an \astep\ belonging to $U$.  Hence, by Lemma~\ref{lem-fass-then-no-bcas}, no \fstep\ can belong to $U$.
If there is no \fstep\ belonging to $U$,
then there can be no \cstep\ belonging to $U$ (by the pseudocode of \help).
Therefore $U.state$ can never be set to \done.
\end{proof}

\begin{cor} \label{cor-state-transitions-respect-figure}
The changes to the $state$ field of an \op\ respect Figure~\ref{fig-state-transitions-simple}.
\end{cor}
\begin{proof}
Immediate from Corollary~\ref{cor-no-change-from-done-retry} and the fact that an \op's state field cannot change to \freezing\ from any other state.
\end{proof}

The next five lemmas prove that any successful \fcas s belonging to an \op\ $U$ must occur prior to the first \fstep\ or \astep\ belonging to $U$.  This result allows us to fill in the gaps between Figure~\ref{fig-state-transitions-simple} and Figure~\ref{fig-state-transitions}.

\after{This next lemma seems to be not quite true if V seq can have duplicates.  I think the way to fix it is to define things slightly differently.  Instead of talking about $fcas on r_i$ we should talk about the $fcas of U for position i$ (meaning an fcas that occurs during the ith iteration of the loop) or some such.  This will probably have an effect on how we blame for failed SCXs when proving progress.  Eric}

\begin{lem} \label{lem-abort-fcas-flow}
Let $U$ be an \op, and let $\langle r_1, r_2, ..., r_l\rangle$
be the sequence of \rec s in $U.V$. 
Suppose an \astep\ belongs to $U$ and let $astep$ be the first such \astep.
Then, there is a
$k \in \{1, \ldots, l\}$ such that
\begin{compactenum}[1.]
	\item{a \fcas\ belonging to $U$ on $r_k$ occurs prior to $astep$,}
	\label{abort-fcas-flow-claim-fcas-on-rk}
	\item{there is no successful \fcas\ belonging to $U$ on $r_k$,}
\label{abort-fcas-flow-claim-no-succ-fcas-on-rk}
	\item{for each $i \in \{1, ..., {k-1}\}$,
a successful \fcas\ belonging to $U$ on $r_i$ occurs prior to $astep$,
and no successful \fcas\ belonging to $U$ on $r_i$ occurs after $astep$, and}
	\label{abort-fcas-flow-claim-succ-fcas-on-r1-etc}
	\item{$r_k.\info$ changes after the \llt$(r_k)$ linked to $S$ reads $r_k.\info$ at line~\ref{ll-reread} and before the first \fcas\ belonging to $U$ on $r_k$.
	} \label{abort-fcas-flow-claim-rk-changes}
\end{compactenum}
\end{lem}
\begin{proof}
Let $H$ be the invocation of \help\ that performs $astep$ and $k$ be the iteration of the loop in \help\ during which $H$ performs $astep$.
The loop in $H$ iterates over the sequence $r_1, r_2, ..., r_l$ of \rec s. 

Claim~\ref{abort-fcas-flow-claim-fcas-on-rk} follows from the definition of $k$:  before $H$ performs $astep$, it performs a \fcas\ belonging to $H$ on $r_k$ at line~\ref{help-fcas}.

To derive a contradiction, suppose Claim 2 is false, i.e., there is a \textit{successful} \fcas\ belonging to $U$ on $r_k$.  By Claim 1, {\it fcas} is before $astep$.
From Corollary~\ref{cor-no-change-from-done-retry} and the fact that $astep$ occurs, we know that no \cstep\ belongs to $U$.
By Lemma~\ref{lem-if-succ-fcas-then-point-u-until-bcas-or-uass}, $r_k.\info$ points to $U$ at all times between the first \fcas\ belonging to $U$ on $r_k$ and $astep$.
However, this contradicts the fact that $r_k.\info$ does not point to $U$ when $H$ performs line \ref{help-check-frozen} just before performing $astep$.

We now prove Claim~\ref{abort-fcas-flow-claim-succ-fcas-on-r1-etc}.
By Claim \ref{abort-fcas-flow-claim-fcas-on-rk}, prior to $astep$,
$H$ performs a \fcas\ belonging to $U$ on $r_k$.
By Lemma~\ref{lem-fcas-on-ri-only-after-succ-fcas-on-previous}, this can only occur after a successful \fcas\ belonging
to $U$ on $r_i$, for all $i <k$.
By Lemma~\ref{lem-only-first-fcas-can-succeed}, there is no successful
\fcas\ belonging to $U$ on $r_i$ after $astep$.

We now prove Claim~\ref{abort-fcas-flow-claim-rk-changes}.
By Claim~\ref{abort-fcas-flow-claim-fcas-on-rk} and Claim~\ref{abort-fcas-flow-claim-no-succ-fcas-on-rk}, an unsuccessful \fcas \ {\it fcas} belonging to $U$ on $r_k$ occurs prior to $astep$.
By line~\ref{help-rinfo} and Observation~\ref{obs-op-invariants}.\ref{inv-llresults}, the old value for {\it fcas} is read from $r_k.\info$ and stored in $r\info$ at line~\ref{ll-read} by the \llt$(r_k)$ linked to $S$.
By Observation~\ref{obs-op-invariants}.\ref{inv-unfrozen}, the \llt$(r_k)$ linked to $S$ again sees $r_k.\info = r\info$ at line~\ref{ll-reread}.
Thus, since {\it fcas} fails, $r_k.\info$ must change after the \llt$(r_k)$ linked to $S$ performs line \ref{ll-reread} and before {\it fcas} occurs.
\end{proof}

\begin{lem} \label{lem-no-succ-fcas-after-fass-or-bcas}
No \fcas\ belonging to an \op\ $U$ is successful after the first \fstep\ or \astep\ belonging to $U$.
\end{lem}
\begin{proof}
First, suppose a \fstep\ belongs to $U$, and let $fstep$ be the first such \fstep.
Then, by Lemma~\ref{lem-if-fass-then-all-succ-fcas}, there is
a successful \fcas\ belonging to $U$ on $r$ for each $r$ in $U.V$
that occurs before $fstep$.
By Lemma~\ref{lem-only-first-fcas-can-succeed}, only the first \fcas\ belonging to $U$ on $r$ can be successful.
Hence, no \fcas\ belonging to $U$ is successful after $fstep$.

Now, suppose an \astep\ belongs to $U$, and let $astep$ be the first such \astep.
Let $\langle r_1, r_2, ..., r_l\rangle$ be the sequence of \rec s in $U.V$. 
By Lemma~\ref{lem-abort-fcas-flow}, there is a $k \in \{1, ..., l\}$ such that no successful \fcas\ belonging to $U$ is performed on any $r_i \in \{r_1, ..., r_{k-1}\}$ after $astep$, and no successful \fcas\ belonging to $U$ on $r_k$ ever occurs.
By Lemma~\ref{lem-fcas-on-ri-only-after-succ-fcas-on-previous},
there is no \fcas\ belonging to $U$ on any $r_i \in \{r_{k+1}, ..., r_l\}$.
\end{proof}

\begin{cor} \label{cor-if-succ-fcas-then-point-u-until-bcas-or-uass}
If there is a successful \fcas\ {\it fcas} belonging to an \op\ $U$ on a \rec\ $r$, then {\it fcas} occurs before time $t$, when the first \astep\ or \cstep\ belonging to $U$ occurs. Moreover, $r.\info$ points to $U$ at all times after {\it fcas} occurs, and before $t$.
\end{cor}
\begin{proof}
By Lemma~\ref{lem-no-succ-fcas-after-fass-or-bcas}, {\it fcas} occurs before $t$.
The claim then follows from Lemma~\ref{lem-if-succ-fcas-then-point-u-until-bcas-or-uass}. 
\end{proof}

\begin{lem} \label{lem-state-transitions2}
Changes to the $state$ and $\freezingdone$ fields of an \op, as well as \fstep s, \astep s, \cstep s and successful \fcas s can only occur as depicted in Figure~\ref{fig-state-transitions}.
\end{lem}
\begin{proof}
Initially, the dummy \op\ has $state =$ \retry\ and $\freezingdone =$ \false\
and, by Lemma \ref{lem-no-steps-belong-to-dummy-op}, they never change.

Every other \op\ $U$ initially has $state =$ \freezing\ and
$\freezingdone =$ \false.
Only \astep s, \fstep s, and \cstep s
can change $state$ or $\freezingdone$.
From the code of \help, each transition shown in Figure \ref{fig-state-transitions},
results in the indicated values for $state$ and $\freezingdone$.
A \cstep\ on line \ref{help-cstep} must be preceded by a \fstep\ on line \ref{help-fstep}.
Therefore, from [\freezing, \false], the only outgoing transitions are
to [\retry, \false] and [\freezing, \true].
By Lemma \ref{lem-fass-then-no-bcas}, there cannot be both
a \fstep\ and an \astep\ belonging to $U$.
Hence, from [\retry, \false], there cannot be a \fstep\ or \cstep\
and there cannot be an \astep\ from  [\freezing, \true] or [\done, \true]. 

By Lemma \ref{lem-no-succ-fcas-after-fass-or-bcas}, successful \fcas s can only occur when $state =$ \freezing\ and $\freezingdone =$ \false.
From the code of \help, for each $r$ in $U.R$, the first \markstep\ belonging to $U$ on $r$ must occur after the first \fstep \ belonging to $U$ and before the first \cstep \ belonging to $U$.
Since each $r$ in $U.R$ initially has $r.marked = \false$, and $r.marked$ is only changed at line~\ref{help-markstep}, where it is set to \true, only the first \markstep\ belonging to $U$ on $r$ can be successful.
\end{proof}

\subsection{The period of time over which a \rec\ is frozen}


We now prove several lemmas which characterize the period of time over which a \rec\ is frozen for an \op.
We first use the fact that the $state$ of an \op\ cannot change from \retry\ to \done\ to extend Lemma~\ref{lem-no-info-change-while-freezing} to prove that the \info\ field of a \rec\ cannot be changed while the \rec\ is frozen for an \op.
In the following, the phrase ``after X and before the first time Y happens'' should be interpreted
to mean ``after X'' in the event that Y never happens.

%

\begin{figure}[tb]
	\centering
	\begin{tabular}{|l|c|c|}
		\hline
		& \hspace{1mm} $r.\info.state$ \hspace{1mm} & $r.marked$ \\
		\hline
		Frozen & \done & \true \\
		& \freezing & \{\true, \false\} \\
		\hline
		Unfrozen & \done & \false \\
		& \retry & \{\true, \false\} \\
		\hline
	\end{tabular}
	\caption{
	When a \rec\ $r$ is frozen,
			in terms of $r.\info.state$ and $r.marked$.}
	\label{fig-freezing-table}
\end{figure}

\begin{lem} \label{lem-records-frozen}
If a \fstep\ belongs to an \op\ $U$ then,
for each $r$ in $U.V$, a \fcas\ belonging to $U$ on $r$ precedes the first \fstep\ belonging to $U$, and $r$ is frozen for $U$ at all times after the first \fcas\ belonging to $U$ on $r$
and before the first \cstep\ belonging to $U$.
\end{lem}
\begin{proof}
Fix any $r$ in $U.V $.
If a \fstep\ belongs to $U$ then,
by Lemma~\ref{lem-if-fass-then-all-succ-fcas},
it is preceded by
a successful \fcas\ belonging to $U$ on $r$.
Further, by Lemma~\ref{lem-fass-then-no-bcas},
no \astep\ belongs to $U$.
Thus, by Corollary~\ref{cor-if-succ-fcas-then-point-u-until-bcas-or-uass},
$r.\info$ points to $U$ at all points between
time $t_0$, when the first \fcas\ belonging to $U$ on $r$ occurs,
and time $t_1$, when the first \cstep\ belonging to $U$ occurs
(after the first \fstep).
Since no \astep\ belongs to $U$,
$U.state =$ \freezing\ at all times before $t_1$.
Hence, by the definition of freezing (see Figure~\ref{fig-freezing-table}),
$r$ is frozen for $U$
at all times between $t_0$ and $t_1$.
\end{proof}

\begin{cor} \label{cor-records-frozen-from-fass-to-uass}
If a \fstep\ belongs to an \op\ $U$, then
each $r$ in $U.V$ is frozen for $U$
at all times between
the first \fstep\ belonging to $U$
and the first \cstep\ belonging to $U$.
\end{cor}
\begin{proof}
Suppose there is a \fstep\ belonging to $U$.
By Lemma~\ref{lem-records-frozen}, each $r$ in $U.V$ is frozen for $U$ at all times between the first \fcas\ belonging to $U$ on $r$
and the first \cstep\ belonging to $U$.
It then follows directly from the pseudocode of \help\ that the first \fstep\ belonging to $U$ must follow the first \fcas\ belonging to $U$ on $r$,
for each $r$ in $U.V$, and precede the first \cstep\ belonging to $U$.
\end{proof}

\begin{cor} \label{cor-marksteps}
A successful \markstep\ belonging to $U$ can occur only while $r$ is frozen for $U$.
\end{cor}
\begin{proof}
Immediate from Lemma~\ref{lem-state-transitions2} and Corollary~\ref{cor-records-frozen-from-fass-to-uass}.
\end{proof}

\begin{lem} \label{lem-become-frozen-only-by-info-change}
A \rec\ can only be changed from unfrozen to frozen by a change in its $\info$ field (which can only be the result of a \fcas).
\end{lem}
\begin{proof}
Let $r$ be a \rec\ whose \info\ field points to an \op\ $U$.
According to the definition of a frozen \rec\ (see Figure~\ref{fig-freezing-table}), if $r.\info$ does not change, then $r$ can only become frozen if $U.state$ changes from \done\ or \retry\ to \freezing, or from \retry\ to \done\ (provided $r$ is marked).  However, both cases are impossible by Corollary~\ref{cor-state-transitions-respect-figure}.
\end{proof}


\newcommand{\permafrozen}{permafrozen}

\begin{defn}
A \rec\ $r$ is called \textbf{\permafrozen\ for} \op\ $U$ if $r$ is marked, $r.\info$ points to $U$ and the $U.state$ is \done.  Notice that a \rec\ that is \permafrozen\ for $U$ is also frozen for $U$.
\end{defn}

\begin{lem} \label{lem-finalized-forever-frozen}
Once a \rec\ $r$ is \permafrozen\ for \op\ $U$, it remains \permafrozen\ for $U$ thereafter.
\end{lem}
\begin{proof}
By definition, when $r$ is \permafrozen\ for $U$, it is frozen for $U$, $U.state$ is \done\ and $r.marked = \true$.
Once $r.marked$ is set to \true, it can never be changed back to \false.
By Corollary~\ref{cor-state-transitions-respect-figure},
$U.state$ was never \retry,
$U.state$ will remain \done\ forever, and
$r$ will be frozen for $U$ as long as $r.\info$ points to $U$.
It remains only to prove that $r.\info$ cannot change while $r$ is \permafrozen\ for $U$.
Note that $r.\info$ can be changed only by a successful \fcas.

To obtain a contradiction,
suppose a \fcas\ {\it fcas} changes $r.\info$ from $U$ to $W$ while $r$ is \permafrozen\ for $U$.
By Lemma~\ref{lem-no-steps-belong-to-dummy-op},
$W$ is not the dummy \op.
Let $S$ be the invocation of \sct\ that created $W$.
From the code of \help, $r$ is in $W.V$. So, by the precondition of \sct, there is an invocation of \llt($r$) linked to $S$.
By Observation~\ref{obs-op-invariants}.\ref{inv-llresults} and line~\ref{help-rinfo}, the old value for {\it fcas} (a pointer to $U$) was read at line~\ref{ll-read} of the \llt$(r)$ linked to $S$.
Let $I$ be the invocation of $\llt(r)$ linked to $S$.
Since we have argued that $U.state$ is never \retry, $U.state \in \{\freezing, \done\}$ when $I$ reads $ state$ from $U.state$ at line~\ref{ll-read-state}.

If $state = \freezing$ then $I$ does not enter the if-block at line~\ref{ll-check-frozen}, and returns \fail \ or \finalized, which contradicts Definition~\ref{defn-llt-linked-to-sct}.\ref{prop-returns-value-different-from-fail-or-finalized}.

Now, consider the case where $state = \done$.
If we can argue that $r$ is marked when $I$ performs line~\ref{ll-read-marked2}, then we shall obtain the same contradiction as in the previous case.
Since $state = \done $, a \cstep \ belonging to $U$ occurs before $I$ performs line~\ref{ll-read-state}.
By Lemma~\ref{lem-state-transitions2}, any successful \markstep\ belonging to $U$ occurs prior to this \cstep.
Therefore, if $r$ is in $U.R$, then $r$ will be marked when $I$ performs line~\ref{ll-read-marked2}, and we obtain the same contradiction.
The only remaining possibility is that $r$ is not in $U.R$, and $r$ is marked by a successful \markstep\ $mstep$ belonging to some other \op\ $U'$ \textit{after} $I$ performs line~\ref{ll-read-marked2}, and before {\it fcas} occurs (which is while $r$ is \permafrozen\ for $U$).
Since $r.\info$ points to $U$ when $I$ performs line~\ref{ll-read}, and again when {\it fcas} occurs, Lemma~\ref{lem-no-aba-info} implies that $r.\info$ points to $U$ throughout this time.
However, this contradicts Corollary~\ref{cor-marksteps}, which states that $mstep$ can only occur while $r.\info$ points to $U'$.
\end{proof}

\begin{lem} \label{lem-frozen-forever-after-markstep}
Suppose a successful \markstep\ $mstep$ belonging to an \op \ $U$ on $r$ occurs.
Then, $r$ is frozen for $U$ when $mstep$ occurs, and forever thereafter.
\end{lem}
\begin{proof}
By Corollary~\ref{cor-marksteps}, $mstep$ must occur while $r$ is frozen for $U$.
From the code of \help, a \fstep \ belonging to $U$ must precede $mstep$, and $r$ must be in $V$ (since it is marked at line~\ref{help-markstep}).
Thus, Corollary~\ref{cor-records-frozen-from-fass-to-uass} implies that $r$ is frozen for $U$ at all times between $mstep$ and the first \cstep\ belonging to $U$.
Since $r.marked$ is never changed from \true \ to \false, $mstep$ must be the first \markstep\ that ever modifies $r.marked$.
From the code of \help, $mstep$ must precede the first \cstep \ belonging to $U$.
If any \cstep\ belonging to $U$ occurs after $mstep$, immediately after the first such \cstep, $r$ will be marked, and $r.\info.state$ will be \done, so $r$ will become \permafrozen\ for $U$.
By Lemma~\ref{lem-finalized-forever-frozen}, $r$ will remain frozen for $U$, thereafter.
\end{proof}

\begin{lem} \label{lem-r-not-frozen-in-good-llt}
Suppose $I$ is an invocation of $\llt(r)$ that returns a value different from \fail \ or \finalized.
Then, $r$ is not frozen at any time in $[t_0, t_1]$, where $t_0$ is when $I$ reads $r\info.state$ at line~\ref{ll-read-state}, and $t_1$ is when $I$ reads $r.\info$ at line~\ref{ll-reread}.
\end{lem}
\begin{proof}
We prove that $r$ is not frozen at any time between $t_0$ and $t_1$.
Since $I$ returns a value different from \fail\ or \finalized, it enters the if-block at line~\ref{ll-check-frozen}, and sees $r.\info = r\info$ at line~\ref{ll-reread}.
Therefore, it sees either $state =$ \done\ and $r.marked = \false$, or $state =$ \retry\ at line~\ref{ll-check-frozen}.
In each case, Corollary~\ref{cor-state-transitions-respect-figure} guarantees that $r\info.state$ will never change again after time $t_0$.
Thus, if $state = \retry$, then $r$ is not frozen at any time between $t_0$ and $t_1$.
Now, suppose $state = \done$.
We prove that $r.marked$ does not change between $t_0$ and $t_1$.
A pointer to an \op\ $W$ is read from $r.\info$ and stored in the local variable $r\info$ at line~\ref{ll-read}, before $t_0$.
At line~\ref{ll-reread}, $r.\info$ still contains a pointer to $W$.
By Lemma~\ref{lem-no-aba-info}, $r.\info$ must not change between line~\ref{ll-read} and line~\ref{ll-reread}.
Therefore, $r.\info$ points to $W$ at all times between $t_0$ and $t_1$.
By Corollary~\ref{cor-marksteps}, a successful \markstep\ can occur between $t_0$ and $t_1$ only if it belongs to $W$.
Since $state = \done$, a \cstep \ belonging to $W$ must have occurred before $t_0$.
By Lemma~\ref{lem-state-transitions2}, any successful \markstep\ belonging to $W$ must have occurred before $t_0$.
Therefore, $W.state =$ \done\ and $r.marked = \false$ throughout [$t_0, t_1$].
\end{proof}

\begin{cor} \label{cor-r-not-frozen-in-linked-llt}
Let $S$ be an invocation of \sct, and $r$ be any \rec \ in the $V$ sequence of $S$. 
Then, $r$ is not frozen at any time in $[t_0, t_1]$, where $t_0$ is when the \llt$(r)$ linked to $S$ reads $r\info.state$ at line~\ref{ll-read-state}, and $t_1$ is when the \llt$(r)$ linked to $S$ reads $r.\info$ at line~\ref{ll-reread}.
\end{cor}
\begin{proof}
Immediate from Definition~\ref{defn-llt-linked-to-sct}.\ref{prop-returns-value-different-from-fail-or-finalized} and Lemma~\ref{lem-r-not-frozen-in-good-llt}.
\end{proof}

\subsection{Properties of \upcas\ steps}

\begin{obs} \label{obs-immutable-fields-do-not-change}
An immutable field of a \rec\ cannot change from its initial value.
\end{obs}
\begin{proof}
This observation follows from the facts that
\rec s can only be changed by \sct\ and 
an invocation of \sct\ can only accept a pointer
to a mutable field as its $fld$ argument (to modify).
\end{proof}

\begin{obs} \label{obs-only-upcas-modifies-records}
Each mutable field of a \rec\ can be modified only by a successful \upcas.
\end{obs}

\begin{obs} \label{obs-all-upcas-write-new}
Each \upcas\ belonging to an \op\ $U$ is of the form \cas$(U.fld, U.old, U.new)$. 
Invariant: $U.fld$ and $U.new$ contain the arguments $fld$ and $new$, respectively, that were passed to the invocation of \sct$(V, R, fld, new)$ that created $U$.
\end{obs}
\begin{proof}
An \upcas\ occurs at line~\ref{help-upcas} in an invocation of \help$(scxPtr)$, where it operates on $scxPtr.fld$, using $scxPtr.old$ as its old value, and $scxPtr.new$ as its new value. The fields of $scxPtr$ do not change after $scxPtr$ is created at line~\ref{sct-create-op}. At this line, the arguments $fld$ and $new$ that were passed to the invocation of \sct$(V, R, fld, new)$ are stored in $scxPtr.fld$ and $scxPtr.new$, respectively.
\end{proof}


\begin{lem} \label{lem-first-upcas-while-frozen}
The first \upcas\ belonging to an \op\ $U$ on a \rec\ $r$ occurs while $r$ is frozen for $U$. 
\end{lem}
\begin{proof}
Let $upcas$ be the first \upcas \ belonging to $U$.
By line~\ref{help-upcas}, such an \upcas\ will modify $U.fld$ which, by Observation~\ref{obs-op-invariants}.\ref{inv-fld}, is a mutable field of a \rec\ $r$ in $U.V$.
Since $upcas$ is preceded by a \fstep\ in the pseudocode of \help, a \fstep\ belonging to $U$ must precede $upcas$.
Hence, Corollary~\ref{cor-records-frozen-from-fass-to-uass} applies, and each $r$ in $U.V$ is frozen for $U$ at all times between the first \fstep\ belonging to $U$ and the first \cstep\ $cstep$ belonging to $U$.
From the code of \help, if $cstep$ exists, then it must occur after $upcas$.
Thus, when $upcas$ occurs, $r$ is frozen for $U$.
\end{proof}

In Section~\ref{sec-impl} we described a constraint on the use of \sct\ that allows us to implement an optimized version of \sct\ (which avoids the creation of a new \rec\ to hold each value written to a mutable field), and noted that the correctness of the unoptimized version follows trivially from the correctness of the optimized version.
In order to prove the next few lemmas, we must invoke this constraint. 
In fact, we assume a weaker constraint, and are still able to prove what we would like to.
We now give this weaker constraint, and remark that it is automatically satisfied if the constraint in Section~\ref{sec-impl} is satisfied.

\begin{con} \label{con-use-of-sct}
Let $fld$ be a mutable field of a \rec \ $r$.
If an invocation $S$ of \sct$(V, R, fld, new)$ is linearized, then:
\begin{itemize}
    \item $new$ is not the initial value of $fld$, and
	\item no invocation of \sct$(V', R', fld, new)$ is linearized before the $\llt(r)$ linked to $S$ is linearized.
\end{itemize}
\end{con}

We prove the following six lemmas solely to prove that only the first \upcas\ belonging to an \op\ can succeed.
This result is eventually used to prove that exactly one successful \upcas\ belongs to any \op\ which is helped to \textit{successful} completion.

We need to know about the linearization of \sct s and linked \llt s to prove the next lemma, which uses Constraint~\ref{con-use-of-sct}. 
Let $S$ be an invocation of \sct, and $U$ be the \op \ that it creates.
As we shall see in Section~\ref{sec-corr-lin}, we linearize $S$ if and only if there is an \upcas\ 
belonging to $U$, and $S$ is linearized at its first \upcas.
Each invocation of \llt \ linked to an invocation of \sct \ is linearized at line~\ref{ll-reread}.
\after{Would it make more sense to put the formal discussion of linearization points here, since you use them in the next proof (instead of merely having a forward pointer to them).}

\begin{lem} \label{lem-two-scts-cannot-write-same-value}
No two \upcas s belonging to different \op s can attempt to change the same field to the same value.
\end{lem}
\begin{proof}
Suppose, to derive a contradiction, that \upcas s belonging to two different \op s $U$ and $U'$ attempt to change the same (mutable) field of some \rec\ $r$ to the same value.
Let $upcas$ and $upcas'$ be the first \upcas \ belonging to $U$ and $U'$, respectively.  Let $S$ and $S'$ be the invocation of \sct\ that created $U$ and $U'$, respectively.
From Observation~\ref{obs-all-upcas-write-new} and the fact that $upcas$ and $upcas'$ attempt to change the same field to the same value, we know that $S$ and $S'$ must have been passed the same $fld$ and $new$ arguments.
Note that $S$ and $S'$ are linearized at $upcas$ and $upcas'$, respectively.
Without loss of generality, suppose $S$ is linearized after $S'$.
By Constraint~\ref{con-use-of-sct}, $S'$ is linearized after the invocation $I$ of $\llt(r)$ linked to $S$ is linearized.

By Lemma~\ref{cor-r-not-frozen-in-linked-llt}, $r$ is not frozen when $I$ is linearized.
By Lemma~\ref{lem-first-upcas-while-frozen}, $r$ is frozen for $U'$ when $upcas'$ occurs (which is after $I$ is linearized).
By Lemma~\ref{lem-become-frozen-only-by-info-change}, $r$ can become frozen for $U'$ only by a successful \fcas \ belonging to $U'$ on $r$.
Therefore, a successful \fcas\ {\it fcas$'$} belonging to $U'$ on $r$ occurs after $I$ is linearized, and before $upcas'$.
By Lemma~\ref{lem-first-upcas-while-frozen}, $r$ is frozen for $U$ when $upcas$ occurs (which is after $upcas'$), which implies that a successful \fcas \ {\it fcas} belonging to $U$ on $r$ occurs after $upcas'$, and before $upcas$.
To recap, $I$ is linearized before {\it fcas$'$}, which is before $upcas'$, which is before {\it fcas}, which is before $upcas$.
By line~\ref{help-rinfo} and Observation~\ref{obs-op-invariants}.\ref{inv-llresults}, the old value $old$ for {\it fcas} is read from $r.\info$ and stored in $r\info$ at line~\ref{ll-read} by $I$.
Since $I$ performs line~\ref{ll-read} before it is linearized, $old$ is read from $r.\info$ before {\it fcas$'$}.
Since {\it fcas$'$} changes $r.\info$ to point to $U'$, Lemma~\ref{lem-no-aba-info} implies that $r.\info$ does not point to $U'$ at any time before {\it fcas$'$}.
Therefore, $old$ is not $U'$.
Since {\it fcas} is successful, $r.\info$ must be changed to $old$ at some point after {\it fcas$'$}, and before $upcas$.
However, this contradicts Lemma~\ref{lem-no-aba-info}, since $r.\info$ had already contained $old$ before {\it fcas$'$}.
\end{proof}

\begin{lem} \label{lem-upcas-cannot-write-initial-value}
An \upcas\ never changes a field back to its initial value.
\end{lem}
\begin{proof}
By Observation~\ref{obs-all-upcas-write-new}, each \upcas\ belonging to an \op\ $U$ attempts to change a field to the value $new$ that was passed as an argument to the invocation of \sct\ that created $U$.  Since Constraint~\ref{con-use-of-sct} implies that $new$ cannot be the initial value of the field, we know that no \upcas\ can change the field to its initial value.
\end{proof}

\begin{lem} \label{lem-upcas-cannot-use-same-old-and-new-values}
No \upcas\ has equal $old$ and $new$ values.
\end{lem}
\begin{proof}
Let $upcas$ be an \upcas\ and let $U$ be the \op\ to which it belongs.
By Observation~\ref{obs-all-upcas-write-new}, the old value used by $upcas$ is $U.old$, and the new value used by $upcas$ is $U.new$. Let $f$ be the field of a \rec\ pointed to by $U.fld$; this is the field to which $upcas$ is applied. By Lemma~\ref{lem-upcas-cannot-write-initial-value}, $U.new$ cannot be the initial value of $f$. If $U.old$ is the initial value of $f$, then we are done. So, suppose $U.old$ is not the initial value of $f$. Since a mutable field can only be changed by a successful \upcas, there exists a successful \upcas\ $upcas'$ which changed $f$ to $U.old$ prior to $upcas$. By Observation~\ref{obs-op-invariants}.\ref{inv-old}, $U.old$ was read from $f$ prior to the start of the invocation $S$ of \sct\ that created $U$ and, therefore, prior to $upcas$. Hence, when $upcas'$ occurs, $U$ has not yet been created. Note that $upcas'$ must occur in an invocation of \help$(ptr')$ where $ptr'$ points to some \op\ $U'$ different from $U$.
However, by Lemma~\ref{lem-two-scts-cannot-write-same-value}, $upcas$ and $upcas'$ use different new values, so $U.old$ (the new value for $upcas'$) must be different from $U.new$ (the new value for $upcas$).
\end{proof}

\begin{lem} \label{lem-only-one-succ-upcas}
At most one successful \upcas\ can belong to an \op.
\end{lem}
\begin{proof}
We prove this lemma by contradiction. Consider the earliest point in the execution when the lemma is violated.
Let $upcas'_U$ be the earliest occurring second successful \upcas\ belonging to any \op, and $U$ be the \op\ to which it belongs, and let $upcas_U$ be the preceding successful \upcas\ belonging to $U$.
Further, let $f$ be the field upon which $upcas'_U$ operates, and let $old$ and $new$ be the old and new values used by $upcas_U$, respectively.  
(By Observation~\ref{obs-all-upcas-write-new}, $upcas_U$ and $upcas'_U$ attempt to change the same field from the same old value to the same new value.)
By Lemma~\ref{lem-upcas-cannot-use-same-old-and-new-values}, we know that $old \neq new$.
Then, since $upcas'_U$ is successful, there must be a successful \upcas\ $upcas'_W$ belonging to some \op\ $W$ which changes $f$ to $old$ between $upcas_U$ and $upcas'_U$.
By Lemma~\ref{lem-upcas-cannot-write-initial-value}, $old$ is not the initial value of $f$. Hence, there must be another successful \upcas\ $upcas_W$ which changes $f$ to $old$ before $upcas_U$.
By Lemma~\ref{lem-two-scts-cannot-write-same-value}, $upcas_W$ must belong to $W$, so $upcas_W$ and $upcas'_W$ both precede $upcas'_U$.
This contradicts the definition of $upcas'_U$.
\end{proof}

\begin{lem} \label{lem-no-aba-on-mutable-fields}
An \upcas\ never changes a field to a value that has already appeared there.  (Hence, there is no ABA problem on mutable fields.)
\end{lem}
\begin{proof}
Suppose a successful \upcas\ $upcas$ belonging to an \op\ $U$ changes a field $f$ to have value $new$.
By Lemma~\ref{lem-upcas-cannot-write-initial-value}, $new$ is not the initial value of $f$. By Lemma~\ref{lem-only-one-succ-upcas}, all successful \upcas s that change $f$ must belong to different \op s. Hence, Lemma~\ref{lem-two-scts-cannot-write-same-value} implies that no \upcas\ other than $upcas$ can change $f$ to $new$. 
\end{proof}

Lemma \ref{lem-only-one-succ-upcas} proved that at most 
one \upcas\ of each \op\ can succeed.
Now we prove that such a successful \upcas\ must be the {\it first} one belonging to \op.

\begin{lem} \label{lem-only-first-upcas-can-succeed}
Only the first \upcas\ belonging to an \op\ $U$ can succeed.
\end{lem}
\begin{proof}
Let $upcas$ be the first \upcas\ belonging to $U$, and $f$ be the field that $upcas$ attempts to modify. 
If $upcas$ succeeds then, by Lemma~\ref{lem-only-one-succ-upcas}, there can be no other successful \upcas\ belonging to $U$.
So, suppose $upcas$ fails.
By Observation~\ref{obs-all-upcas-write-new}, each \upcas\ belonging to $U$ uses the same old value $U.old$. 
By Observation~\ref{obs-op-invariants}.\ref{inv-old}, $U.old$ was read from $f$ prior to the start of the invocation $S$ of \sct\ that created $U$ and, therefore, prior to $upcas$.
Then, since $upcas$ fails, $f$ must change between when $U.old$ is read from $f$ and when $upcas$ occurs.
By Observation~\ref{obs-only-upcas-modifies-records}, $f$ can only be changed by an \upcas. By Lemma~\ref{lem-no-aba-on-mutable-fields}, each \upcas\ applied to $f$ changes it to a value that it has not previously contained. Therefore, $f$ will never again be changed to $U.old$. Hence, every subsequent \upcas\ belonging to $U$ will fail.
\end{proof}

\subsection{Freezing works}

In addition to being used to prove the remaining lemmas of this section, the following two results are used to prove linearizability in Section~\ref{sec-corr-lin}.
Intuitively, they allow us to determine whether a \rec \ has changed simply by looking at its \info\ field, and whether it is frozen.

\begin{cor} \label{cor-upcas-only-modifies-frozen}
An \upcas\ belonging to an \op\ $U$ on a \rec\ $r$ can succeed only while $r$ is frozen for $U$. 
\end{cor}
\begin{proof}
Suppose a successful \upcas\ $upcas$ belongs to an \op\ $U$.
By Lemma~\ref{lem-only-first-upcas-can-succeed}, it is the first \upcas\ belonging to $U$. 
The claim follows from Lemma~\ref{lem-first-upcas-while-frozen}.
\end{proof}
By Observation~\ref{obs-only-upcas-modifies-records}, a mutable field of $r$ can only change while $r$ is frozen.

\begin{lem} \label{lem-read-unfrozen-info-twice-no-change}
If a \rec\ $r$ is not frozen at time $t_0$, $r.\info$ points to an \op\ $U$ at or before time $t_0$, and $r.\info$ points to $U$ at time $t_1 > t_0$, then no field of $r$ is changed during $[t_0,t_1]$.
\end{lem}
\begin{proof}
Since $r.\info$ points to $U$ at or before time $t_0$, and again at time $t_1$, Lemma~\ref{lem-no-aba-info} implies that $r.\info$ must point to $U$ at all times in $[t_0, t_1]$.
Further, from Lemma~\ref{lem-become-frozen-only-by-info-change}, $r$ can only be changed from unfrozen to frozen by a change to $r.\info$. 
Therefore, at all times in $[t_0, t_1]$, $r$ is not frozen.
By Corollary~\ref{cor-upcas-only-modifies-frozen}, and Observation~\ref{obs-only-upcas-modifies-records}, each mutable field of $r$ can change only while $r$ is frozen.
By Corollary~\ref{cor-marksteps}, $r.marked$ can change only while $r$ is frozen.
Finally, by Observation~\ref{obs-immutable-fields-do-not-change}, immutable fields do not ever change.
Hence, no field of $r$ changes during $[t_0, t_1]$.
\end{proof}

The remaining results of this section describe intervals over which certain fields of a \rec\ do not change.
Suppose $U$ is an \op \ created by an invocation $S$ of \sct, $r$ is a \rec\ in $U.V$, and $I$ is the invocation of $\llt(r)$ linked to $S$.
Intuitively, we use the preceding lemma to prove, over the next two lemmas, that no field of $r$ changes between when $I$ last reads $r.\info$, and when $r$ is frozen for $U$.
We then use this result in Section~\ref{sec-corr-help} to prove that $S$ succeeds if and only if this holds for each $r$ in $V$.
The remaining results of this section are used primarily to prove that exactly one successful \upcas\ belongs to $U$ if a \fstep\ belongs to $U$ (and $S$ does not crash, or some process helps it complete).

\begin{cor} \label{cor-read-unfrozen-info-then-fcas-no-change}
Let $U$ be an \op, and $S$ be the invocation of \sct\ that created $U$.
If there is a successful \fcas\ belonging to $U$ on $r$, then no field of $r$ changes after the \llt$(r)$ linked to $S$ reads $r\info.state$ at line~\ref{ll-read-state}, and before this \fcas\ occurs.
\end{cor}
\begin{proof}
Let {\it fcas} be a successful \fcas\ belonging to $U$ on $r$.
Note that the \llt$(r)$ linked to $S$ terminates before $S$ begins.
Since $S$ creates $U$ and {\it fcas} changes $r.\info$ to point to $U$, $S$ begins before {\it fcas}.
We now check that Lemma~\ref{lem-read-unfrozen-info-twice-no-change} applies. By Corollary~\ref{cor-r-not-frozen-in-linked-llt}, $r$ is not frozen when the \llt$(r)$ linked to $S$ executes line~\ref{ll-read-state}.
From line~\ref{help-rinfo} of \help\ and Observation~\ref{obs-op-invariants}.\ref{inv-llresults}, we know the old value $r\info$ for {\it fcas} is read from $r.\info$ at line~\ref{ll-read} by the \llt$(r)$ linked to $S$.
Further, since {\it fcas} succeeds, $r.\info$ contains $r\info$ just prior to {\it fcas}.
Thus, Lemma~\ref{lem-read-unfrozen-info-twice-no-change} applies, and proves the claim.
\end{proof}


\begin{lem} \label{lem-fcas-then-upcas-no-change}
If an \upcas \ belongs to an \op\ $U$ then, for each $r$ in $U.V$, there is a successful \fcas\ belonging to $U$ on $r$, and no mutable field of $r$ changes during $[t_0(r), t_1)$, where $t_0(r)$ is when the first such \fcas\ occurs, and $t_1$ is when the first \upcas \ belonging to $U$ occurs.
\end{lem}
\begin{proof}
Suppose an \upcas\ belongs to an \op\ $U$.
Let $upcas$ be the first such \upcas.
Since each \upcas\ is preceded in the code by a \fstep, a \fstep\ also belongs to $U$.
Fix any $r$ in $U.V$.
By Lemma~\ref{lem-fass-then-no-bcas}, there is a successful \fcas\ belonging to $U$ on $r$.
By Lemma~\ref{lem-records-frozen}, $r$ is frozen for $U$ at all times in $[t_0(r), t_2)$, where $t_2$ is when the first \cstep\ belonging to $U$ occurs.
Since an \upcas\ belonging to an \op\ $W$ can modify $r$ only while $r$ is frozen for $W$ (by Corollary~\ref{cor-upcas-only-modifies-frozen}), any \upcas\ that modifies $r$ during $[t_0(r), t_2)$ must belong to $U$.
From the code of \help, $t_0(r) < t_1 < t_2$.
However, since the first \upcas\ belonging to $U$ occurs at $t_1$, no \upcas\ belonging to $U$ can occur during $[t_0(r), t_1)$.
\end{proof}

\after{After Feb 10, may want to reword next lemma to cover the case where $S$ crashes
before any \cstep\ belonging to $U$ occurs}

\begin{lem} \label{lem-marked-changes}
Let $U$ be an \op \ created by an invocation $S$ of \sct, and $r$ be a \rec\ in $U.V$.  Let $t_0$ be when the \llt$(r)$ linked to $S$ reads $U.state$ at line~\ref{ll-read-state}, and $t_2$ be when the first \cstep \ belonging to $U$ occurs.
If an \upcas\ belongs to $U$ then, for each $r$ in $U.V$, between $t_0$ and $t_2$, $r.marked$ can be changed (from \false\ to \true, by the first \markstep\ belonging to $U$ on $r$) only if $r$ is in $U.R$.
\end{lem}
\begin{proof}
Fix any $r$ in $U.V$.
The fact that $r.marked$ can be changed only from \false\ to \true\ follows immediately from the fact that $r.marked$ is initially \false, and is only changed at line~\ref{help-markstep}.
It also follows that a successful \markstep\ on $r$ must be the first \markstep\ on $r$.
The rest of the claim is more subtle.
Suppose a successful \markstep\ $mstep$ belonging to an \op\ $W$ on $r$ occurs during $(t_0,t_2)$.
By Lemma~\ref{lem-frozen-forever-after-markstep}, $r$ is frozen for $W$ when $mstep$ occurs.
From the code of \help, a \fstep \ $fstep$ belonging to $U$ precedes the first \upcas \ belonging to $U$, and a \fcas \ belonging to $U$ on $r$ precedes $fstep$.
By Lemma~\ref{lem-if-fass-then-all-succ-fcas}, there is a successful \fcas\ belonging to $U$ on $r$.
By Lemma~\ref{lem-only-first-fcas-can-succeed}, it must be the first \fcas \ {\it fcas} belonging to $U$ on $r$.
Let $t_1$ be when {\it fcas} occurs.
Note that $t_0 < t_1 < t_2$.
By Corollary~\ref{cor-read-unfrozen-info-then-fcas-no-change}, $r$ does not change during $[t_0, t_1)$.
Thus, $mstep$ cannot occur in $[t_0, t_1)$.
By Lemma~\ref{lem-records-frozen}, $r$ is frozen for $U$ at all times during $[t_1, t_2]$.
This implies $U = W$, so $mstep$ is the first \markstep\ belonging to $U$ on $r$.
Finally, since there is a \markstep\ belonging to $U$ on $r$, we obtain from line~\ref{help-markstep} that $r$ is in $U.R$.
\end{proof}

\subsection{Correctness of \help} \label{sec-corr-help}

The following lemma shows that a helper of an \op\ cannot return \true\ until after the \op\ is \done.
We shall use this to ensure that the \sct\ does not return \true\ until after the \sct\ has taken effect.

\begin{lem} \label{lem-no-return-true-in-loop-until-uass}
An invocation of \help$(scxPtr)$ where $scxPtr$ points to an \op\ $U$
cannot return from line~\ref{help-return-true-loop} before the first \cstep\ belonging to $U$.
\end{lem}
\begin{proof}
Suppose an invocation $H$ of \help$(scxPtr)$ returns
at line~\ref{help-return-true-loop}.
Before returning, $H$ sees $r.\info \neq scxPtr$ at line~\ref{help-check-frozen}, which implies that $r.\info$ does not point to $U$.
Prior to this, $H$ performs a \fcas\ belonging to $U$ at line~\ref{help-fcas}.
By line~\ref{help-fcstep}, a \fstep\ belongs to $U$.
Then, since Lemma~\ref{lem-records-frozen} states that $r.\info$ points to $U$ at all times between the first \fcas\ belonging to $U$ and the first \cstep\ belonging to $U$, a \cstep\ belonging to $U$ must occur before $H$ returns.
\end{proof}

Next, we obtain an exact characterization of the \upcas\ steps that succeed.

\begin{lem} \label{lem-first-upcas-succ}
If there is an \upcas \ belonging to an \op \ $U$, then the first \upcas \ belonging to $U$ is successful and changes the mutable field pointed to by $U.fld$ from $U.old$ to $U.new$.
No other \upcas \ belonging to $U$ is successful.
\end{lem}
\begin{proof}
Let $t_0(r)$ be when the \llt$(r)$ linked to $S$ reads $r\info.state$ at line~\ref{ll-read-state}, and $t_1$ be when the first \upcas\ belonging to $U$ occurs.
Since an \upcas\ belongs to $U$, Lemma~\ref{lem-fcas-then-upcas-no-change} implies that, for each $r$ in $U.V$, no mutable field of $r$ changes between $t_0(r)$ and $t_1$.
From Observation~\ref{obs-op-invariants}.\ref{inv-old} and the code of \llt, we see that the value stored in $U.old$ is read from the field pointed to by $U.fld$ after time $t_0(r)$.
Further, since the \llt$(r)$ linked to $S$ terminates before $S$ begins (by the definition of an \llt\ linked to an \sct) and, in turn, before $U$ is created, we know the value stored in $U.old$ was read before any \upcas\ belonging to $U$ occurred.
Then, since $U.fld$ is a mutable field of a \rec\ in $U.V$, this field does not change between $t_0(r)$ and $t_1$.
By Observation~\ref{obs-all-upcas-write-new}, any \upcas\ belonging to $U$ will attempt to change $U.fld$ from $U.old$ to $U.new$, so the first \upcas\ belonging to $U$ will succeed.
Lemma~\ref{lem-only-first-upcas-can-succeed} completes the proof.
\end{proof}

Our next lemma shows that the \sct s that are linearized have the desired effect.

\begin{lem} \label{lem-if-fass-belongs-to-op-then}
Let $U$ be an \op \ created by an invocation $S$ of \sct, and $ptr$ point to $U$.
If either 
\begin{compactitem}
\item
a \fstep\ belongs to $U$ and some invocation of \help$(ptr)$ terminates, or 
\item
$S$ or any invocation of \help$(ptr)$ returns \true,
\end{compactitem}
 then the following claims hold.
\begin{enumerate}
\item{
	Every invocation of \help$(ptr)$ that terminates returns \true.
} \label{claim-help-true-then-all-helpers-true}
\item{
	Exactly one successful \upcas\ belongs to $U$,
	and it is the first \upcas\ belonging to $U$.
	It changes the mutable field pointed to by $U.fld$
	from $U.old$ to $U.new$. 
} \label{claim-help-true-then-fass}
\item{
	A \fstep\ of $U$ and a \cstep\ of $U$
	occur before any invocation of \help$(ptr)$ returns.%
}\label{claim-help-true-then-returns-after-uass}
\item{
	At all times after the first \cstep\ for $U$,
	each $r$ in $U.R$ is \permafrozen\ for $U$.
} \label{claim-help-true-then-finalized}
\end{enumerate}
\end{lem}
\begin{proof}
We first simplify the lemma's hypothesis.
If $S$ returns \true, then $S$'s invocation of \help$(ptr)$ has returned \true.
If some invocation of \help$(ptr)$ returns \true,
a \cstep\ belongs to $U$ by Lemma~\ref{lem-no-return-true-in-loop-until-uass},
and that \cstep\ is preceded by a \fstep\ of $U$.
So, for the remainder of the proof, we can assume that a \fstep\ belongs to $U$ and some
invocation of \help($ptr$) terminates.


\textbf{Proof of Claim~\ref{claim-help-true-then-all-helpers-true}:}
Since a \fstep \ belongs to $U$, Lemma~\ref{lem-fass-then-no-bcas} implies that no \astep \ belongs to $U$.
Thus, $H$ cannot return at line~\ref{help-return-false}, which implies that $H$ must return \true.

\textbf{Proof of Claim~\ref{claim-help-true-then-fass} and Claim~\ref{claim-help-true-then-returns-after-uass}:}
By Claim~\ref{claim-help-true-then-all-helpers-true}, $H$ must return \true.
If $H$ returns at line~\ref{help-return-true},
then it does so 
after performing an \upcas\ and a \cstep, each belonging to $U$.
Otherwise, $H$ returns at line~\ref{help-return-true-loop}.
However, by Lemma~\ref{lem-no-return-true-in-loop-until-uass}, no invocation of \help$(ptr)$ can return at line~\ref{help-return-true-loop} until the first \cstep\ belonging to $U$, which is necessarily preceded by an \upcas\ for $U$ (by inspection of \help).
Thus, Claim~\ref{claim-help-true-then-returns-after-uass} is proved.
Lemma~\ref{lem-first-upcas-succ} proves Claim~\ref{claim-help-true-then-fass}.

\textbf{Proof of Claim~\ref{claim-help-true-then-finalized}:}
By Corollary~\ref{cor-records-frozen-from-fass-to-uass}, every $r$ in $U.V$ is frozen for $U$ from the first \fstep\ belonging to $U$ until the first \cstep\ belonging to $U$.
From the code of \help, each $r$ in $U.R$ is marked before the first \cstep \ belonging to $U$, and Lemma~\ref{lem-marked-changes} implies that they are still marked when the first \cstep \ belonging to $U$ occurs.
Further, immediately after the first \cstep\ belonging to $U$ (which must exist by Claim~\ref{claim-help-true-then-returns-after-uass}), $U.state$ will be \done, so each $r$ that is in both $U.V$ and $U.R$ will be \permafrozen\ for $U$.
Since $R$ is a subsequence of $V$, and $U.R$ and $U.V$ do not change after they are obtained at line~\ref{sct-create-op} from $R$ and $V$, respectively, 
it follows from Lemma~\ref{lem-finalized-forever-frozen} that each $r$ in $U.R$ remains \permafrozen\ for $U$ forever.
\end{proof}

Now we show that \sct s that are not linearized do not modify any mutable fields, and do not return \true.

\begin{lem} \label{lem-help-false}
Let $U$ be an \op\ created by an invocation $S$ of \sct, and $ptr$ be a pointer to $U$.
If $S$ or any invocation $H$ of \help$(ptr)$ returns \false, then the following claims hold.
\begin{enumerate}
\item{Every invocation of \help$(ptr)$ that terminates returns \false.}
\label{claim-help-false}
\item{An \astep \ belonging to $U$ occurs before any invocation of \help$(ptr)$ returns.}
\label{claim-help-false-abort}
\item{No \upcas\ belongs to $U$.}
\label{claim-help-false-no-upcas}
\end{enumerate}
\end{lem}
\begin{proof}
Note that, if $S$ returns \false, then its invocation of \help$(ptr)$ returns \false, so an invocation $H$ exists and returns \false.
By Lemma~\ref{lem-if-fass-belongs-to-op-then}.\ref{claim-help-true-then-all-helpers-true}, if any invocation of \help$(ptr)$ returned \true, then $H$ would have to return \true.
Since $H$ returns \false, every terminating invocation of \help$(ptr)$ must return \false, which proves Claim~\ref{claim-help-false}.
We now prove Claim~\ref{claim-help-false-abort} and Claim~\ref{claim-help-false-no-upcas}.
Consider the invocation $H'$ of \help$(ptr)$ that returns earliest.
By Claim~\ref{claim-help-false}, $H'$ returns \false.
Before $H'$ returns \false, an \astep\ belonging to $U$ is performed at line~\ref{help-astep}.
Thus, Lemma~\ref{lem-fass-then-no-bcas} implies that no \fstep\ belongs to $U$.
By the code of \help, each \upcas\ belonging to $U$ follows a \fstep\ belonging to $U$.
\end{proof}

Now that we have proved each invocation of \help\ that returns \true\ or \false\ has its expected effect, we must prove that the return value is always correct.  (Otherwise, for example, two invocations of \sct \ with overlapping $V$ sequences could interfere with one another, but still return \true.)

\begin{lem} \label{lem-help-true-iff-records-unchanged}
Let $U$ be an \op\ created by an invocation $S$ of \sct, and $ptr$ be a pointer to $U$.
Any invocation $H$ of \help$(ptr)$ that terminates returns \true\ if no $r$ in $U.V$ changes from when the \llt$(r)$ linked to $S$ reads $r.\info$ at line~\ref{ll-reread} at time $t_0(r)$ to when the first \fcas\ belonging to $U$ on $r$ at time $t_1(r)$.  Otherwise, $H$ returns \false.
\end{lem}
\begin{proof}
Since $H$ terminates, it must return \true\ or \false.
Hence, it suffices to prove $H$ returns \true\ if and only if no $r$ in $U.V$ changes between $t_0(r)$ and $t_1(r)$.

\textbf{Case I: } Suppose $H$ returns \true.
By Lemma~\ref{lem-if-fass-belongs-to-op-then}.\ref{claim-help-true-then-returns-after-uass}, a \fstep\ belongs to $U$.  Hence, by Lemma~\ref{lem-if-fass-then-all-succ-fcas}, there is a successful \fcas\ belonging to $U$ for each $r$ in $U.V$.
Then, by Corollary~\ref{cor-read-unfrozen-info-then-fcas-no-change}, no field of $r$ changes between $t_0(r)$ and $t_1(r)$.

\textbf{Case II: }
Suppose $H$ returns \false.  We show that some $r$ in $U.V$ changes between $t_0(r)$ and $t_1(r)$.
Since $H$ can only return \false\ at line~\ref{help-return-false}, immediately before $H$ returns, it performs an \astep\ belonging to $U$.
Thus, Lemma~\ref{lem-abort-fcas-flow}.\ref{abort-fcas-flow-claim-rk-changes} applies and the claim is proved.
\end{proof}

%

\subsection{Linearizability of \llt/\sct/\vlt} \label{sec-corr-lin}

As described in Section~\ref{sec-operations}, we linearize all reads, all invocations of \llt\ that do not return \fail, all invocations of \vlt\ that return \true, and all invocations of \sct\ that modify the sequential data structure (all that return \true, and some that do not terminate).
In our implementation, subtle interactions with a concurrent invocation of \sct\ can cause an invocation of \llt \ to return \fail.
Since this cannot occur in a linearized execution,  we do not linearize  any such invocation of \llt.
Similarly, we allow some invocations of \sct \ and \vlt \ to return \false \ because of interactions with concurrent invocations of \sct. 
Rather than distinguishing between the invocations of \sct \ or \vlt \ that return \false \ because of an earlier linearized invocation of \sct\ (which is allowed by the sequential specification), and those that return \false \ because of contention, we simply opt not to linearize any invocation of \sct \ or \vlt \ that returns \false.
Alternatively, we could have accounted for the invocations that returns \false \ because of contention by allowing spurious failures in the sequential specification of the operations.
However, this would unnecessarily complicate the sequential specification.
Intuitively, an algorithm designer using \llt, \sct \ and \vlt \ is most likely to be interested in invocations of \sct \ and \vlt \ that return \true\ since these, respectively, change the sequential data structure, and indicate that a set of \rec\ have not changed since they were last passed to successful invocations of \llt\ by this process.
Knowing whether an invocation of \sct \ or \vlt \ was unsuccessful because of a change to the sequential data structure, or because of contention, is less likely to be useful.

Before we give the linearization points, we state precisely which invocations of \sct \ we shall linearize.
Let $S$ be an invocation of \sct, and $U$ be the \op\ it creates.
We linearize $S$ if and only if there is an \upcas \ belonging to $U$.
By Lemma~\ref{lem-if-fass-belongs-to-op-then}, every successful invocation of \sct\ will be linearized.
(By Lemma~\ref{lem-help-false}, no unsuccessful invocation of \sct\ will be linearized.)

We first give the linearization points of the operations, then we prove our \llt/\sct/\vlt\ implementation respects the correctness specification given in Section~\ref{sec-operations}.
\\

\noindent\textbf{Linearization points:}
\begin{itemize}
	\item{
		An \llt$(r)$ that returns values at line~\ref{ll-return}
		is linearized at line~\ref{ll-reread}.
		We linearize an \llt$(r)$ that returns \finalized\
		at line~\ref{ll-return-finalized}.
	}
	\item{
		Let $U$ be an \op \ created by an invocation $S$ of \sct. 
		Suppose there is an \upcas \ belonging to $U$.
		We linearize $S$ at the first such \upcas\
		(which is the unique successful \upcas \ belonging to $U$,
		 by Lemma~\ref{lem-first-upcas-succ}).
	}
	\item{
		An invocation $I$ of \vlt\ that returns \true\ is linearized
		at the first execution of line~\ref{vlt-reread}.
	}
	\item{
		We assume reads are atomic.
		Hence, a read is simply linearized when it occurs.
	}
\end{itemize}

\begin{lem} \label{lem-lin-points-during-ops}
The linearization point of each linearized operation occurs during the operation.
\end{lem}
\begin{proof}
This is trivial to see for reads, and invocations of \llt\ and \vlt.
Let $U$ be an \op \ created by a linearized invocation $S$ of \sct, and $ptr$ be a pointer to $U$.
This claim is not immediately obvious for $S$, since the first \upcas \ belonging to $U$ may be performed by a process helping $U$ to complete (not the process performing $S$).
Since $S$ is linearized, there is an \upcas \ $upcas$ belonging to $U$, which can only occur in an invocation of  \help$(ptr)$.
Since $ptr$ points to $U$, which is created by $S$, $upcas$ must occur after the start of $S$.
If $S$ does not terminate, then we are done.
Otherwise, Lemma~\ref{lem-if-fass-belongs-to-op-then}.\ref{claim-help-true-then-returns-after-uass} implies that a \cstep \ belongs to $U$, and the first such \cstep \ occurs before any invocation of \help$(ptr)$ returns.
From the code of \help, $upcas$ must occur before before the first \cstep \ belonging to $U$, so $upcas$ must occur before any invocation of \help$(ptr)$ returns.
Finally, since $S$ invokes \help$(ptr)$, $upcas$ must occur before $S$ terminates.
%
\end{proof}

We first show that each read returns the correct result according to its linearization point.

\after{Shorten the next proof.  It's too pedantic.}

\begin{lem} \label{lem-lin-read}
If a read $R_f$ of a field $f$ is linearized after a successful invocation of \sct$(V, R, fld, new)$, where $fld$ points to $f$, then $R_f$ returns the parameter $new$ of the last such \sct.
Otherwise, $R_f$ returns the initial value of $f$.
\end{lem}
\begin{proof}
We proceed by cases.

\textbf{Case I:} $f$ is an immutable field.
In this case, a pointer to $f$ cannot be the $fld$ parameter of an invocation of \sct.
Further, by Observation~\ref{obs-immutable-fields-do-not-change}, $f$ cannot be modified after its initialization, so $R_f$ returns the initial value of $f$.

\textbf{Case II:} $f$ is a mutable field.
Suppose there is no successful invocation of \sct$(V, R, fld, new)$, where $fld$ points to $f$, linearized before $R_f$.
Since an invocation of \sct \ is linearized at its first \upcas, there can be no \upcas \ on $f$ prior to $ R_f $.
Since $f $ can only be modified by successful \upcas, $ R_f $ must return the initial value of $f $.

Now, suppose there is a successful invocation of \sct$(V, R, fld, new)$, where $fld$ points to $f$, linearized before $R_f$.
Let $S$ be the last such invocation of \sct\ linearized before $R_f$, and $U$ be the \op \ it creates.
By Lemma~\ref{lem-if-fass-belongs-to-op-then}.\ref{claim-help-true-then-fass}, there is exactly one successful \upcas \ $upcas$ belonging to $U$, occurring at the linearization point of $S$.
Since each successful \upcas \ is the linearization point of some invocation of \sct, and no invocation of \sct$ (V', R', fld, new') $ is linearized between $S$ and $ R_f $, no successful \upcas \ occurs between $S$ and $ R_f $.
Since a mutable field can only be changed by \upcas, $ R_f $ returns the value stored by the successful \upcas \ $upcas$ belonging to $U$.
By Lemma~\ref{lem-if-fass-belongs-to-op-then}.\ref{claim-help-true-then-fass}, $upcas$ changes $f $ to $U.new$.
Finally, since $U.new$ does not change after it is obtained from $new$ at line~\ref{sct-create-op}, the claim is proved.
\end{proof}

Next, we prove that an \llt\ that returns a snapshot does return the correct result according to its linearization point.

\begin{cor} \label{cor-lin-llt-success}
Let $r$ be a \rec \ with mutable fields $f_1, ..., f_y$, and $I$ be an invocation of $\llt(r)$ that returns a tuple of values $\langle m_1, ..., m_y \rangle$ at line~\ref{ll-return}. 
For each mutable field $f_i$ of $r$, if $I$ is linearized after a successful invocation of \sct$(V, R, fld, new)$, where $fld$ points to $f_i$, then $m_i$ is the parameter $new$ of the last such invocation of \sct.
Otherwise, $m_i$ is the initial value of $f_i$.
\end{cor}
\begin{proof}
Since $I$ returns at line~\ref{ll-return}, 
the same value is read from $r.info$ on line~\ref{ll-read} and line~\ref{ll-reread} at times $t_0$ and $t_1$, respectively. 
By Lemma~\ref{lem-r-not-frozen-in-good-llt}, $r$ is unfrozen at line~\ref{ll-check-frozen}, which is between $t_0$ and $t_1$.
Thus, Lemma~\ref{lem-read-unfrozen-info-twice-no-change} implies that $r$ does not change during $[t_0,t_1]$.
Since the values returned by the \llt\ are read from the fields of $r$ between $t_0$ and $t_1$ (at line~\ref{ll-collect}), each read returns the same result as it would if it were executed atomically at the linearization point of the \llt\ (line~\ref{ll-reread}).
Finally, Lemma~\ref{lem-lin-read} completes the proof.
\end{proof}

Next, we show that an \llt$(r)$ returns \finalized\ only if $r$ really has been finalized.

\begin{lem} \label{lem-lin-llt-finalized}
Let $I$ be an invocation of \llt$(r)$, and $U$ be the \op \ to which $I$ reads a pointer at line~\ref{ll-read}.
If $I$ returns \finalized, then an invocation $S$ of \sct$ (V, R, fld, new) $ that created $U$ is linearized before $I$, and $r$ is in $R$.
\end{lem}
\begin{proof}
Suppose $I$ returns \finalized.
Then, $I$ is linearized at line~\ref{ll-return-finalized}.
When $I$ performs line~\ref{ll-check-finalized}, either $r\info.state$ is \done \ or $I$'s invocation of \help$(r\info)$ returns \true.
We show that, in either case, a \cstep \ belonging to $U$ must have occurred before $I$ performs line~\ref{ll-return-finalized}.
If $r\info.state$ is \done, then this follows immediately from the fact that no \op \ has \done\ as its initial state.
Otherwise, Lemma~\ref{lem-if-fass-belongs-to-op-then}.\ref{claim-help-true-then-returns-after-uass} implies that a \cstep \ belonging to $U$ occurs before $I$'s invocation of \help$(r\info)$ returns.
Since a \cstep\ belongs to $U$, Lemma~\ref{lem-no-steps-belong-to-dummy-op} implies that $U$ is not the dummy \op, so there must be an invocation $S$ of \sct$(V, R, fld, new)$ that created $U$.
From the code of \help, an \upcas\ belonging to $U$ occurs before the first \cstep\ belonging to $U$.
Since $S$ is linearized at its first \upcas, $S$ is linearized before $I$.

It remains to show that $r$ is in $R$.
Since $U.R$ does not change after it is obtained from $R$ at line~\ref{sct-create-op}, it suffices to show $r$ is in $U.R$.
By line~\ref{ll-check-finalized}, $marked_1 = \true$, which means that $r$ is marked when $I$ reads a pointer to $U$ from $r.\info$ at line~\ref{ll-read}.
Let $t_0$ be when $I$ performs line~\ref{ll-read}, and $t_1$ be when the first \cstep \ belonging to $U$ occurs.
We consider two cases.
Suppose $t_1 < t_0$.
Then, when $I$ performs line~\ref{ll-read}, $r$ is marked, $r.\info$ points to $U$, and $U.state = \done$, which means that $r$ is \permafrozen\ for $U$.
By Lemma~\ref{lem-finalized-forever-frozen}, $r$ is frozen for $U$ at all times after $t_0$.
Now, suppose $t_0 < t_1$.
In this case, Lemma~\ref{lem-state-transitions2} implies $U.state = \freezing$ when $I$ performs line~\ref{ll-read}, and that $U.state$ will never be \retry.
By Lemma~\ref{lem-no-info-change-while-freezing}, $r.\info$ must point to $U$ at all times in $[t_0, t_1]$.
Thus, at $t_1$, $r$ is marked and $r.\info$ points to $U$, which means that $r$ is \permafrozen\ for $U$.
%
%
By Lemma~\ref{lem-finalized-forever-frozen}, $r$ is frozen for $U$ at all times after $t_1$.
In each case, $r$ is frozen for $U$ at all times in some (non-empty) suffix of the execution.
Since $r$ is marked, there must be a successful \markstep\ belonging to some \op\ $W$ on $r$.
By Lemma~\ref{lem-frozen-forever-after-markstep}, $r$ is frozen for $W$ at all times after this \markstep.
Therefore, $U = W$, which means there is a \markstep\ belonging to $U$ on $r$.
Finally, line~\ref{help-markstep} implies that $r$ is in $U.R$.
\end{proof}

\begin{lem} \label{lem-lin-llt-finalized-if-after-sct}
Let $r$ be a \rec, $I$ be an invocation of \llt$(r)$ that terminates, and $S$ be a linearized invocation of \sct$(V, R, fld, new)$ with $r$ in $R$.
$I$  returns \finalized\ if it is linearized after $S$, or begins after $S$ is linearized.
(This implies $I$ will be linearized in both cases.)
\end{lem}
\begin{proof}
Let $U$ be the \op\ created by $S$.

\textbf{Case I:} $I$ begins after $S$ is linearized.
In this case, $S$ is linearized at a successful \upcas\ $upcas$ belonging to $U$ that occurs before $I$ begins.
From the code of \help, a \markstep\ on $r$ belonging to $U$ must occur before $upcas$.
Consider the first \markstep\ $mstep$ on $r$.
Since $r.marked$ is initially \false, $mstep$ must be successful.
Let $W$ be the \op\ to which $mstep$ belongs.
By Lemma~\ref{lem-frozen-forever-after-markstep}, $r$ is frozen for $W$ at all times after $mstep$.
By Lemma~\ref{cor-upcas-only-modifies-frozen}, $r$ is frozen for $U$ when $upcas$ occurs (which is after $mstep$).
Since $r$ can be frozen for only one \op\ at a time, $W = U$.
Therefore, $r.\info$ points to $U$ throughout $I$.
So, when $I$ performs line~\ref{ll-check-frozen}, it will see $state = \done$ and $marked_2 = \true$.
Moreover, when it subsequently performs line~\ref{ll-check-finalized}, it will see $r\info.state = \done$ and $marked_1 = \true$, so it will return \finalized.

\textbf{Case II:} $I$ is linearized after $S$.
$I$ can either be linearized at line~\ref{ll-reread} or at line~\ref{ll-return-finalized}.
If it is linearized at line~\ref{ll-return-finalized}, then it returns \finalized, and we are done.
Suppose, in order to derive a contradiction, that $I$ is linearized at line~\ref{ll-reread}.
Then, $I$ returns at line~\ref{ll-return} and, by Corollary~\ref{cor-r-not-frozen-in-linked-llt}, $r$ is unfrozen  at all times during $[t_0, t_1]$, where $t_0$ is when $I$ performs line~\ref{ll-read-state}, and $t_1$ is when $I$ performs line~\ref{ll-reread}.
Since $I$ is linearized at time $t_1$, $S$ must be linearized at an \upcas \ $upcas$ belonging to $U$ that occurs before time $t_1$.
By Corollary~\ref{cor-upcas-only-modifies-frozen}, $upcas$ can only occur while $r$ is frozen for $U$.
Since $r$ is unfrozen at all times during $[t_0,t_1]$, $upcas$ must occur at some point before $t_0$.
From the code of \help, a \markstep\ belonging to $U$ must occur before $upcas$.
Consider the first \markstep\ $mstep$ belonging to any \op\ $W$ on $r$.
As we argued in the previous case, $r$ is frozen for $W$ at all times after $mstep$.
However, this contradicts our argument that $r$ is unfrozen at all times during $[t_0,t_1]$ (since $t_0$ is after $mstep$).
Thus, $I$ cannot be linearized at line~\ref{ll-reread}, so $I$ must return \finalized.
\end{proof}

The following lemma proves that an \sct\ succeeds only when it is supposed to, according to the
correctness specification.

\begin{lem} \label{lem-lin-sct-vlt}
If an invocation $I$ of \vlt$(V)$ or \sct$(V, R, fld, new)$ is linearized then, for each $r$ in $V$, no \sct$(V', R', fld', new')$ with $r$ in $V'$ is linearized between the \llt$(r)$ linked to $I$ and $I$.
\end{lem}
\begin{proof}
Fix any $r$ in $V$. 
By the preconditions 
of \sct\ and \vlt, there must be an \llt$(r)$ linked to $I$.
If $I$ is an invocation of \sct, then let $U$ be the \op \ that it creates.
Let $L$ be the $\llt(r)$ linked to $I$, $t_0$ be when $L$ performs line~\ref{ll-read}, $t_1$ be when $L$ performs line~\ref{ll-read-state} and $t_2$ be when $L$ is linearized (at line~\ref{ll-reread}).
Since $L$ is linked to $I$, $t_2$ exists.
Let $t_3$ be when $I$ is linearized.
For \sct, $t_3$ is when the first \upcas\ belonging to $U$ occurs.
For \vlt, $t_3$ is when $I$ first performs line~\ref{vlt-reread}.
Since $I$ is linearized, $t_3$ exists.
Finally, we define time $t_4$.
For \sct, $t_4$ is when the first \cstep\ belonging to $U$ occurs, or the end of the execution if there is no \cstep\ belonging to $U$ (or $\infty$ if the execution is infinite).
For \vlt, $t_4$ is when $I$ sees $r\info = r.\info$ at line~\ref{vlt-reread} (in the iteration for $r$).
If $I$ is an invocation of \vlt\ then, since $I$ is linearized, it returns \true.
This implies that $I$ must see $r\info = r.\info$ at line~\ref{vlt-reread} in the iteration for $r$, so $t_4$ exists.
Clearly, $t_4$ exists if $I$ is an invocation of \sct.

We now prove $ t_0 <t_1 <t_2 <t_3 <t_4$.
From the code of \llt, $t_0 <t_1 <t_2$.
Suppose $I$ is an invocation of \vlt.
Since $L$ terminates before $I$ begins, $t_2 < t_3$.
Trivially, $t_3 < t_4$.
Now, suppose $I$ is an invocation of \sct.
Since each \upcas \ belonging to $U$ occurs in an invocation of \help$(ptr)$ where $ptr$ points to $U$, and $U$ is created during $I$, $t_3 $ must occur after the start of $I$.
Since $L$ terminates before $I$ begins, $t_2 < t_3$.
From the code of \help, the first \upcas \ belonging to $U$ precedes the first \cstep \ belonging to $U$ (as well as the other options for $t_4$), so $t_3 < t_4$.

Next, we prove that, at all times in $[t_1, t_4)$, $r$ is either frozen for $U$, or not frozen (i.e., $r$ is not frozen for any \op \ different from $U$ at any point during $[t_1, t_4)$).
We consider two cases.

\textbf{Case I:}
Suppose $I$ is an invocation of \vlt.
Then, $r.\info$ contains $r\info$ at $t_0$, and again at $t_4$.
By Lemma~\ref{lem-no-aba-info}, $r.\info$ must contain $r\info$ at all times in $[t_0, t_4]$.
By Corollary~\ref{cor-r-not-frozen-in-linked-llt}, $r$ is unfrozen at time $t_1$.
By Lemma~\ref{lem-become-frozen-only-by-info-change}, $r$ can only be changed from unfrozen the frozen by a change to $r.\info$.
Since $r$ does not change during $[t_0, t_4]$, $r$ is unfrozen at all times in $[t_1, t_4]$.

\textbf{Case II:}
Suppose $I$ is an invocation of \sct.
From the code of \help, a \fstep\ belonging to $U$ precedes the first \upcas \ belonging to $U$.
By Lemma~\ref{lem-if-fass-then-all-succ-fcas}, a successful \fcas \ belonging to $U$ on $r$ precedes the first \fstep \ belonging to $U$.
Let $t_2'$ be when the first successful \fcas \ {\it fcas} belonging to $U$ on $r$ occurs.
It follows that $t_2' < t_3$.
Since each  \fcas \ belonging to $U$ occurs in an invocation of \help$(ptr)$ where $ptr$ points to $U$, and $U$ is created during $I$, each \fcas \ belonging to $U$ must occur after the start of $I$.
Recall that $L$ terminates before the start of $I$.
Hence, $t_2' > t_2$.
By Observation~\ref{obs-op-invariants}.\ref{inv-llresults} and line~\ref{help-rinfo}, the old value for {\it fcas} is the value $r\info$ that was read from $r.\info$ at line~\ref{ll-read} of $L$ (at $t_0$).
Since {\it fcas} is successful, $r.\info$ must contain $r\info$ just before {\it fcas}.
Therefore, $r.\info$ contains $r\info$ at $t_0$, and again at $t_2'$.
By the same argument we made in Case I (but with $t_2'$ instead of $t_4$), Lemma~\ref{lem-no-aba-info}, Corollary~\ref{cor-r-not-frozen-in-linked-llt}, and Lemma~\ref{lem-become-frozen-only-by-info-change} imply that $r$ is unfrozen at all times in $[t_1, t_2']$.
By Lemma~\ref{lem-records-frozen}, $r$ is frozen for $U$ at all times in $(t_2',t_4)$, which proves this case.

At last, we have assembled the results needed to obtain a contradiction.
Suppose, to derive a contradiction, that an invocation $S$ of \sct$(V', R', fld', new')$ with $r$ in $V'$ is linearized between $L$ and $I$ (i.e., in $(t_2, t_3)$).
Let $W$ be the \op \ created by $S$.
$S$ is linearized at the first \upcas \ $upcas$ belonging to $W$.
From the code of \help, a \fstep \ belonging to $W$ precedes $upcas$, and $upcas$ precedes any \cstep  \ belonging to $W$.
By Lemma~\ref{lem-records-frozen}, a successful \fcas \ {\it fcas} belonging to $W$ on $r$ precedes $upcas$, and $r$ is frozen for $W$ at all times after {\it fcas}, and before the first \cstep\ belonging to $W$.
Therefore, $r$ is frozen for $W$ when $upcas$ occurs in $(t_2, t_3)$.
Since, at all times in $[t_1, t_4)$, $r$ is either frozen for $U$, or not frozen, we must have $W = U$.
This is a contradiction, since $upcas$ occurs before the \textit{first} \upcas \ belonging to $U$ occurs (at $t_3$).
Thus, $S$ cannot exist.
\end{proof}

\begin{thm}
\label{thm-llt-sct-vlt-correct} %
Our implementation of \llt/\sct/\vlt\ satisfies the correctness specification discussed in Section~\ref{sec-operations}.
That is, we linearize all successful \llt s, all successful \sct s, a subset of the \sct s that never terminate, all successful \vlt s, and all reads, such that:
\begin{enumerate}
	\item
		Each read of a field $f$ of a \rec\ $r$ returns
		the last value stored in $f$ by a linearized \sct\
		(or $f$'s initial value, if no linearized \sct\ has modified $f$).
	    \label{claim-llt-sct-vlt-correctness-read}
	\item{%
  		Each linearized \llt($r$) that does not return \finalized\
		returns	the last value stored in each mutable field $f$ of $r$
		by a linearized \sct\ (or $f$'s initial value, if no linearized \sct\ has modified~$f$).}%
        \label{claim-llt-sct-vlt-correctness-llt-values}
	\item{
		Each linearized \llt$(r)$ returns \finalized\ if and only if it is linearized
		after an \sct($V, R, fld,$ $new)$ with $r$ in $R$.}
        \label{claim-llt-sct-vlt-correctness-llt-finalized-if-lin-after-sct}
	\item
	    If an invocation $I$ of \sct$(V, R, fld, new)$ or \vlt$(V)$ returns \true\ then,
	    for all $r$ in $V$, there has been no \sct$(V', R', fld', new')$
	    with $r$ in $V'$ linearized since the \llt$(r)$ linked to $I$.
	    \label{claim-llt-sct-vlt-correctness-sct-vlt}
\end{enumerate}
\end{thm}
\begin{proof}
By Lemma~\ref{lem-lin-points-during-ops}, the linearization point of each operation occurs during that operation.
Claim~\ref{claim-llt-sct-vlt-correctness-read} follows immediately from Lemma~\ref{lem-lin-read}.
Claim~\ref{claim-llt-sct-vlt-correctness-llt-values} is immediate from Lemma~\ref{cor-lin-llt-success}.
The only-if direction of Claim~\ref{claim-llt-sct-vlt-correctness-llt-finalized-if-lin-after-sct} follows from Lemma~\ref{lem-lin-llt-finalized}, and the if direction follows from Lemma~\ref{lem-lin-llt-finalized-if-after-sct}.
Claim~\ref{claim-llt-sct-vlt-correctness-sct-vlt} is immediate from Lemma~\ref{lem-lin-sct-vlt}.
\end{proof}

\subsection{Progress Guarantees}

\begin{lem} \label{lem-wait-free}
\llt, \sct \ and \vlt \ are wait-free
\end{lem}
\begin{proof}
The loop in $H$ iterates over the elements of the finite sequence $V$ and performs a constant amount of work during each iteration.  If $H$ does not return from the loop, then it performs a constant amount of work after the loop and returns. The claim then follows from the code.
\end{proof}

\begin{lem}\label{finalize-progress} Our implementation satisfies the first progress property in Section \ref{sec-operations}:  Each terminating \llt$(r)$ returns \finalized\ if it begins after the end of a successful \sct$(V,$ $R,$ $fld,$ $new)$ with $r$ in $R$ or after another \llt$(r)$ has returned \finalized. 
\end{lem}
\begin{proof}
Consider a terminating invocation $I'$ of \llt($r$).  
If $I'$ begins after the end of a successful \sct$(V,$ $R,$ $fld,$ $new)$ with $r$ in $R$, the claim follows from Lemma~\ref{lem-lin-llt-finalized-if-after-sct}.
If $I'$ begins after another invocation $I$ of \llt$(r)$ has returned \finalized, then
$I$ is linearized after an invocation $S$ of \sct$(V, R, fld, new)$ with $r$ in $R$.
Since $I$ precedes $I'$, $I'$ starts after $S$ is linearized.
By Lemma~\ref{lem-lin-llt-finalized-if-after-sct}, $I'$ returns \finalized.
\end{proof}

We now begin to prove the non-blocking progress properties.  First, we design a way to assign
blame to an \sct\ for each failed invocation of \llt, \vlt\ or \sct.

\begin{defn} \label{defn-blame-llx}
Let $I$ be an invocation of \llt \ that returns \fail.
If $I$ enters the if-block at line~\ref{ll-check-frozen}, then let $U$ be the \op \ pointed to by $r.\info$ when $I$ performs line~\ref{ll-reread}.
Otherwise, let $U$ be the \op \ pointed to by $r.\info$ when $I$ performs line~\ref{ll-read}.
We say $I$ \textbf{blames} the invocation $S$ of \sct \ that created $U$.
(We prove below that $S$ exists.)
\end{defn}

\begin{lem} \label{lem-if-llt-fail-then-blame-sct}
If an invocation $I$ of \llt\ returns \fail, then it blames some invocation of \sct.
\end{lem}
\begin{proof}
Suppose $I$ enters the if-block at line~\ref{ll-check-frozen}.
Then, $I$ must see $r.\info \neq r\info $ at line~\ref{ll-reread}.
Let $U$ be the \op \ pointed to by $r.\info$ when $I$ performs line~\ref{ll-reread}.
Since $I$ reads $r\info$ from $r.\info$ at line~\ref{ll-read}, $r.\info$ must change to point to $U$ between when $I$ performs line~\ref{ll-read} and line~\ref{ll-reread}. 
Thus, there must be a successful \fcas \ belonging to $U$ on $r$ between these two times.
Since a successful \fcas \ belongs to $U$, Lemma~\ref{lem-no-steps-belong-to-dummy-op} implies that $U$ cannot be the dummy \op.
Therefore, $U$ must be created by an invocation of \sct.

Now, suppose $I$ does not enter the if-block at line~\ref{ll-check-frozen}.
Then, from the code of \llt, $r\info.state$ cannot be \retry \ when $I$ performs line~\ref{ll-read-state}, so the \op \ pointed to by $r\info$ is not the dummy \op.
Since $I$ reads the value stored in $r\info$ from $r.\info$ at  line~\ref{ll-read}, the \op \ pointed to by $r.\info$ when $I$ performs line~\ref{ll-read} must have been created by an invocation of \sct.
\end{proof}

\begin{defn} \label{defn-blame-vlt}
Let $I$ be an invocation of \vlt$(V)$ that returns \false, and $r$ be the \rec\ in $V$ for which $I$ sees $ r.\info \neq r\info $ at line~\ref{vlt-reread}.
Consider the first successful \fcas \ on $r$ between when the \llt$(r)$ linked to $I$ reads $r.\info$ at line~\ref{ll-read}, and when $I$ sees $ r.\info \neq r\info $ at line~\ref{vlt-reread}.
Let $U$ be the \op \ to which this \fcas \ belongs, and $S$ be the invocation of \sct \ that created $U$.
(We prove below that $S$ exists.)
We say $I$ \textbf{blames} $S$ \textbf{for} $r$.
\end{defn}

\begin{lem} \label{lem-if-vlt-false-then-blame-sct}
If an invocation $I$ of \vlt \ returns \false, then it blames some invocation of \sct.
\end{lem}
\begin{proof}
Since $I$ returns \false, it sees $ r\info \neq r.\info $ at line~\ref{vlt-reread}, for some $r$ in $V$.
Let $p$ be the process that performs $I$.
By line~\ref{vlt-info}, $r\info$ is a copy of $r$'s \info\ value in $p$'s local table of \llt \ results.
By the precondition of \vlt \ and the definition of an \llt$(r)$ linked to $I$, this value is read from $r.\info$ at line~\ref{ll-read} by the \llt$(r)$ linked to $I$.
Therefore, $r.\info$ must change between when the \llt$(r)$ linked to $I$ performs line~\ref{ll-read} and when $I$ sees $ r\info \neq r.\info $ at line~\ref{vlt-reread}.
Thus, there must be a successful \fcas \ belonging to some \op \ $U$ on $r$ between these two times.
Since a \fcas \ belongs to $U$, Lemma~\ref{lem-no-steps-belong-to-dummy-op} implies that $U$ is not the dummy \op, so $U$ must have been created by an invocation of \sct.
\end{proof}

\begin{defn} \label{defn-blame-scx}

Let $U$ be an \op \ created by an invocation $S$ of \sct \ that returns \false, and $U'$ be an \op \ created by an invocation $S'$ of \sct.
Consider the \rec s $r$ that are in both $U.V$ and $U'.V $, and for which there is no successful \fcas\ belonging to $U$ on $r$.
Let $r'$ be the \rec\ among these which occurs earliest in $U.V$. 
We say $S$ \textbf{blames} $S'$ \textbf{for} $r'$ if and only if there is a successful \fcas \ on $r'$ belonging to $U'$, and this \fcas \ is the earliest successful \fcas \ on $r'$ to occur between when the \llt$(r')$ linked to $S$ reads $r'.\info$ at line~\ref{ll-read} and the first \fcas \ belonging to $U$ on $r'$.
\end{defn}

\begin{lem} \label{lem-if-abort-then-blame-different-scx}
Let $U$ be an \op\ created by an invocation $S$ of \sct.
If $S$ returns \false, then it blames some other invocation of \sct.
\end{lem}
\begin{proof}
Since $S$ returns \false, Lemma~\ref{lem-help-false}.\ref{claim-help-false-abort} implies that an \astep \ belongs to $U$.
By Lemma~\ref{lem-abort-fcas-flow}, 
there is a \rec \ $r_k$ in $U.V$ such that there is a \fcas\ belonging to
$U$ on $r_k$, but no successful one.
Moreover,
$r_k.\info$ changes after time $t_1$, when the \llt$(r_k)$ linked to $S$ reads $r_k.\info$ at line~\ref{ll-reread}, and before time $t_2$, when the first \fcas\ belonging to $U$ on $r_k$ occurs.
Since the \llt$(r_k)$ linked to $S$ terminates before $U$ is created (at line~\ref{sct-create-op} of $S$), and a \fcas\ belonging to $U$ can only occur after $U$ is created, we know $t_1 < t_2$.
Let $t_0$ be the time when the \llt$(r_k)$ performs line~\ref{ll-read}.
Note that $t_0 < t_1 < t_2$.
Since $r_k.\info$ can only be changed by a successful \fcas, there must be a successful \fcas \ on $r_k$  during $(t_1, t_2)$.
Let {\it fcas} the the earliest successful \fcas\ on $r_k$  during $(t_0, t_2)$, and let $ U' $ be the \op\ to which it belongs.
Since 
{\it fcas} occurs before the first \fcas \ belonging to $U$ on $r_k$, we know that $U \neq U' $.
Let \[\rho = \{ r \mid r \mbox{ is in } U.V \mbox{ and } r \mbox{ is in } U'.V \mbox{ and } \nexists \mbox{ successful \fcas \ belonging to } U \mbox{ on } r \}.\]
By the code of \help, a \fcas \ belonging to $ U' $ can only modify a \rec\ in $ U'.V $. 
Thus, $r_k \in \rho$.

We now show $r_k$ is the element of $\rho$ that occurs earliest in $U.V$. 
Suppose, to derive a contradiction, that some $r_i \in \rho$ comes before $r_k$ in $U.V$. 
By Lemma~\ref{lem-abort-fcas-flow}.\ref{abort-fcas-flow-claim-fcas-on-rk} and Lemma~\ref{lem-abort-fcas-flow}.\ref{abort-fcas-flow-claim-no-succ-fcas-on-rk}, there is an unsuccessful \fcas \ belonging to $U$ on $r_k$.
By Lemma~\ref{lem-fcas-on-ri-only-after-succ-fcas-on-previous}, before this unsuccessful \fcas, there must be a successful \fcas\ belonging to $U$ on $r_i$.
However, this implies $r_i \notin \rho$, which is a contradiction.

Let $S'$ be the invocation of \sct \ that creates $U'$.
Thus far, we have shown that $S$ blames $S'$.
It remains to show that $S \neq S'$.
By Lemma~\ref{lem-no-steps-belong-to-dummy-op}, a \fcas \ or \astep \ cannot belong to the dummy \op.
Therefore, neither $U$ nor $ U' $ can be the dummy \op.
Since $U \neq U'$, $U$ and $U'$ must be created by different invocations of \sct.
\end{proof}

We now prove that an invocation of \llt$(r)$ can return \fail \ only under certain circumstances.

\begin{defn} \label{defn-threatening-section}
Let $U$ be an \op \ created by an invocation $S$ of \sct.
The \textbf{threatening section} of $S$ begins with the first \fcas \ belonging to $U$, and ends with the first \cstep \ or \astep \ belonging to $U$.
\end{defn}

\begin{lem} \label{lem-succ-fcas-or-upcas-only-during-sct}
Let $U$ be an \op \ created by an invocation $S$ of \sct.
The threatening section of $S$ lies within $S$, and every successful \fcas\ or \upcas\ belonging to $U$ occurs during $S$'s threatening section.
\end{lem}
\begin{proof}
Let $ptr$ be a pointer to $U$, and $t_0$ and $t_1$ be the times when $S$'s threatening section begins and ends, respectively.
Since each \fcas, \upcas, \astep\ or \cstep\ belonging to $U$ occurs in an invocation of \help $ (ptr) $, and $S$ creates $U$, we know that these steps can only occur after $S$ begins.
Hence, $t_0$ is after $S$ begins.
Clearly every \fcas \ belonging to $U$ occurs after $t_0$.
From the code of \help, the first \upcas \ belonging to $U$ occurs between $t_0$ and $t_1$.
By Lemma~\ref{lem-only-first-upcas-can-succeed}, this is the only \upcas \ belonging to $U$ that can succeed.
By Lemma~\ref{lem-only-first-fcas-can-succeed}, no \fcas \ belonging to $U$ can succeed after the first \fstep \ or \astep \ belonging to $U$.
From the code of help, the first \fstep \ belonging to $U$ must occur before the first \cstep \ belonging to $U$.
Thus, every successful \fcas \ belonging to $U$ occurs between $t_0$ and $t_1$.
From the code of \sct, $S$ performs an invocation $H$ of \help $ (ptr) $ before it returns, and $H$ will perform either a \cstep \ or \astep \ belonging to $U$, so long as it does not return from line~\ref{help-return-true-loop}.
By Lemma~\ref{lem-no-return-true-in-loop-until-uass}, $H$ cannot return from line~\ref{help-return-true-loop} until after the first \cstep \ belonging to $U$.
Therefore, a \cstep \ or \astep \ belonging to $U$ must occur before $S$ terminates, so $t_1$ is before $S$ terminates.
\end{proof}

\begin{obs} \label{obs-if-fcas-then-r-in-V}
Let $S$ be an invocation of \sct$(V, R, fld, new)$, and $U$ be the \op \ it creates.
If there is a \fcas \ belonging to $U$ on $r$, then $r$ is in $V$.
\end{obs}
\begin{proof}
From the code of \help, there will only be a \fcas \ belonging to $U$ on $r$ if $r$ is in $U.V$, and line~\ref{sct-create-op} implies that $r$ in $V$.
\end{proof}

\begin{lem} \label{lem-llx-can-only-fail-if-concurrent-scx}
An invocation $I$ of \llt $ (r)$ can return \fail \ only if it overlaps the threatening section of some invocation of \sct $ (V, R, fld, new)$ with $r$ in $V$.
\end{lem}
\begin{proof}
By Lemma~\ref{lem-if-llt-fail-then-blame-sct}, $I$ blames an invocation $S$ of \sct.
Let $U$ be the \op \ created by $S$, and $ ptr $ be a pointer to $U$.
By Definition~\ref{defn-blame-llx}, $I$ reads a pointer to $U$ from $r.\info$.
Since $U$ is not the dummy \op, $r.\info$ can only point to $U$ after a successful \fcas \ belonging to $U$ on $r$.
By Observation~\ref{obs-if-fcas-then-r-in-V}, $r$ is in $V$.
We now show that $I$ overlaps the threatening section of $S$.
Consider the two cases of Definition~\ref{defn-blame-llx}.

\textbf{Case I:} $I$ enters the if-block at line~\ref{ll-check-frozen}, and reads a pointer to $U$ from $r.\info$ at line~\ref{ll-reread}.
In this case, from the code of \llt, we know that $r.\info$ changes between when $I$ performs line~\ref{ll-read} and when $I$ performs line~\ref{ll-reread}.
Since $r.\info$ can only be changed to point to $U$ by a successful \fcas\ belonging to $U$, there must be a successful \fcas \ belonging to $U$ during $I$.
By Lemma~\ref{lem-succ-fcas-or-upcas-only-during-sct}, $I$ must overlap the threatening section of $S$.

\textbf{Case II:} $I$ does not enter the if-block at line~\ref{ll-check-frozen}.
Since $I$ reads a pointer to $U$ from $r.\info$ at line~\ref{ll-read}, we know that $r\info$ is a pointer to $U$.
By the test at line~\ref{ll-check-frozen}, either $state= \done $ and $marked_2=\true $, or $state = \freezing $.

Suppose $state= \freezing$.
Then, $U.state = \freezing $ when $I$ performs line~\ref{ll-read-state}.
Since $U$ is not the dummy \op, a pointer to $U$ can appear in $r.\info$ only after a successful \fcas \ belonging to $U$.
By Corollary~\ref{cor-state-transitions-respect-figure}, $U.state $ can only be \freezing \ before the first \cstep \ or \astep \ belonging to $U$.
Therefore, Definition~\ref{defn-threatening-section} implies that $I$ performs line~\ref{ll-read-state} during the threatening section of $S$.

Now, suppose $state = \done $ and $marked_2=\true$.
By Corollary~\ref{cor-state-transitions-respect-figure}, $U.state =$ \done \ at all times after $I$ performs line~\ref{ll-read-state}.
We consider two sub-cases.
If $marked_1= \true $, then $I$ will return \finalized \ if it reaches line~\ref{ll-check-finalized}.
Since we have assumed that $I$ returns \fail, this case is impossible.
Otherwise, a \markstep\ $mstep$ belonging to some \op \ $W$ changes $r.marked $ to \true\ between line~\ref{ll-read-marked1} and line~\ref{ll-read-marked2}.
It remains only to show that $mstep$ occurs during the threatening section of the invocation of \sct \ that created $W$.
Since $r.marked $ is initially \false, and is never changed from \true \ to \false, $mstep$ must be the first \markstep\ belonging to $W$ on $r$.
From the code of \help, a \fstep \ belonging to $W$ must precede $mstep$.
Therefore, Lemma~\ref{lem-fass-then-no-bcas} implies that no \astep\ belonging to $W$ ever occurs.
From the code of \help, $mstep$ must occur after the first \fcas \ belonging to $W$, and before the first \cstep \ belonging to $W$.
By Definition~\ref{defn-threatening-section}, $mstep$ occurs during the threatening section of the invocation of \sct \ that created $W$.
\end{proof}

We now prove that an invocation of \sct \ or \vlt \ can return \false \ only under certain circumstances.

\begin{defn} \label{defn-vulnerable-section}
The \textbf{vulnerable interval} of an invocation $I$ of \sct \ or \vlt \ begins at the earliest starting time of an \llt$(r)$ linked to $I$, and ends when $I$ ends.
\end{defn}

\begin{lem} \label{lem-if-sct-or-vlt-blames-then-succ-fcas}
Let $I$ be an invocation of \sct \ or \vlt, and $U$ be an \op \ created by an invocation $S$ of \sct.
If $I$ blames $S$ for a \rec\ $r$, then a successful \fcas \ belonging to $U$ on $r$ occurs during $I$'s vulnerable interval.
\end{lem}
\begin{proof}
Suppose $I$ is an invocation of \sct.
Let $U_I$ be the \op \ created by $I$.
By Definition~\ref{defn-blame-scx}, a successful \fcas \ belonging to $U$ on $r$ occurs between when the \llt$(r)$ linked to $I$ performs line~\ref{ll-read}, and the first \fcas \ {\it fcas} belonging to $U_I$ on $r$.
By Lemma~\ref{lem-succ-fcas-or-upcas-only-during-sct}, {\it fcas} occurs during $I$.
Now, suppose $I$ is an invocation of \vlt.
Then, by Definition~\ref{defn-blame-vlt}, a successful \fcas \ belonging to $U$ on $r$ occurs between when the \llt$(r)$ linked to $I$ performs line~\ref{ll-read}, and when $I$ sees $ r.\info \neq r\info $ at line~\ref{ll-reread}.
\end{proof}

\begin{obs} \label{obs-if-sct-or-vlt-blames-for-r-then-r-in-blaming-V}
If an invocation $I$ of \sct$ (V, R, fld, new)$ or \vlt$ (V)$ blames an invocation of \sct \ for a \rec \ $r$, then $r$ is in $V$.
\end{obs}
\begin{proof}
Suppose $I$ is an invocation of \sct.
Let $U$ be the \op \ created by $I$.
By Definition~\ref{defn-blame-scx}, $r$ is in $U.V $.
Since $ U.V $ does not change after $U$ is created at line~\ref{sct-create-op} of $I$, $r$ is in $V$.
Now, suppose $I$ is an invocation of \vlt.
In this case, the claim is immediate from Definition~\ref{defn-blame-vlt}.
\end{proof}

\begin{obs} \label{obs-if-sct-or-vlt-blames-for-r-then-r-in-blamed-V}
If an invocation $I$ of \sct\ or \vlt\ blames an invocation $S$ of \sct$(V, R, fld, new)$ for a \rec \ $r$, then $ r$ is in $V$.
\end{obs}
\begin{proof}
Let $U$ be the \op \ created by $S$.
By Lemma~\ref{lem-if-sct-or-vlt-blames-then-succ-fcas}, there is a successful \fcas \ belonging to $U$ on $r$.
The claim then follows from Observation~\ref{obs-if-fcas-then-r-in-V}.
\end{proof}

\begin{lem} \label{lem-sct-or-vlt-false-only-if}
An invocation $I$ of \sct$(V, R, fld, new)$ or \vlt$(V)$ ending at time $t$ can return \false \ only if its vulnerable interval overlaps the threatening section of some other \sct$(V', R', fld', new')$,
where some \rec\ appears in both $V$ and $V'$.
\end{lem}
\begin{proof}
Suppose $I$ returns \false.
By Lemma~\ref{lem-if-vlt-false-then-blame-sct} and Lemma~\ref{lem-if-abort-then-blame-different-scx}, $I$ blames an invocation $S$ of \sct$(V',$ $R',$ $fld',$ $new')$, where $ I \neq S $, for a \rec \ $r$.
Let $U$ be the \op \ created by $S$, and $U_I$ be the \op \ created by $I$.
By Lemma~\ref{lem-if-sct-or-vlt-blames-then-succ-fcas}, a successful \fcas \ {\it fcas} belonging to $U$ on $r$ occurs during $I$'s vulnerable interval.
By Lemma~\ref{lem-succ-fcas-or-upcas-only-during-sct}, {\it fcas} occurs during the threatening section of $S$.
Therefore, $I$'s vulnerable interval overlaps the threatening section of $S$.
By Observation~\ref{obs-if-sct-or-vlt-blames-for-r-then-r-in-blaming-V} and Observation~\ref{obs-if-sct-or-vlt-blames-for-r-then-r-in-blamed-V}, $r$ is in both $V$ and $V'$.
\end{proof}

We now prove bounds on the number of invocations of \llt, \sct \ and \vlt \ that can blame an invocation of \sct.

\begin{lem} \label{lem-uass-or-bcas-before-llx-returns-fail}
Let $I$ be an invocation of \llt$(r)$ that returns \fail, and $U$ be the \op\ created by the invocation of \sct \ that is blamed by $I$.
A \cstep \ or \astep \ belonging to $U$ occurs before $I$ returns.
\end{lem}
\begin{proof}
By Definition~\ref{defn-blame-llx}, we know that $I$ reads a pointer to $U$ from $r.\info$.
By Lemma~\ref{lem-if-llt-fail-then-blame-sct}, $U$ is not the dummy \op.
This implies that a pointer to $U$ can only appear in $r.\info$ after a successful \fcas\ belonging to $U$ on $r$.
By Corollary~\ref{cor-if-succ-fcas-then-point-u-until-bcas-or-uass}, $r.\info$ points to $U$ at all times after the first \fcas\ belonging to $U$ on $r$, and before the first \cstep\ or \astep\ belonging to $U$.
Thus, if $r.\info$ changes after a pointer to $U$ is read from $r.\info$ at line~\ref{ll-read} or line~\ref{ll-reread}, and before either of the two times that $r.\info$ is read at line~\ref{ll-help-fail}, then we know that a \cstep\ or \astep\ belonging to $U$ has already occurred.
Otherwise, any read of $r.\info$ at line~\ref{ll-help-fail} returns a pointer to $U$.
Hence, $I$ checks whether $U.state = $ \freezing \ at line~\ref{ll-help-fail} and, if so, $I$ helps $U$.
We consider two cases.

\textbf{Case I:}
$I$ sees $U.state =$ \freezing \ at line~\ref{ll-help-fail}.
In this case, $I$ helps $U$ before returning.
From the code of \help, if $I$'s invocation of \help  \ returns \false, then $I$ performs an \astep \ belonging to $U$ during its invocation of \help.
Otherwise, by Lemma~\ref{lem-if-fass-belongs-to-op-then}.\ref{claim-help-true-then-returns-after-uass}, a \cstep \ belonging to $U$ occurs before $I$'s invocation of \help\ returns.

\textbf{Case II:}
$I$ sees $U.state \neq$ \freezing \ at line~\ref{ll-help-fail}.
In this case, $U.state $ must be \done \ or \retry.
Since $U$ is not the dummy \op, we know that $U.state $ is initially \freezing.
Therefore, a \cstep \ or \astep \ belonging to $U$ must occur before line~\ref{ll-help-fail}.
\end{proof}

\after{Is it actually true that two llts by a proc can blame the same sct, or can we tighten the next lemma to say $at most one$ instead of $at most two$?}

\begin{lem} \label{lem-sct-only-blamed-by-one-llt-per-process}
Each invocation of \sct \ can be blamed by at most two invocations of \llt\ per process.
\end{lem}
\begin{proof}
Let $S$ be an invocation of \sct.
To derive a contradiction, suppose there is some process $p$ that blames $S$ for three failed
invocations of \llt, $I'', I'$ and $I$ (which are performed by $p$ in this order).
By Definition~\ref{defn-blame-llx}, $I$, $I'$ and $I''$ each read a pointer to $U$ from $r.\info$, either at line~\ref{ll-read} or at line~\ref{ll-reread}.
Since $r.\info$ points to $U$ at some point during $I''$, and again at or after the time $I'$ performs line~\ref{ll-read}, we know from Lemma~\ref{lem-no-aba-info} that $r.\info$ points to $U$ when $I'$ performs line~\ref{ll-read}.
By the same argument,  $r.\info$ points to $U$ when $I$ performs line~\ref{ll-read}.
Thus, in both $I'$ and $I$, the local variable $r\info$ is a pointer to $U$.

By Lemma~\ref{lem-uass-or-bcas-before-llx-returns-fail}, a \cstep \ or \astep \ belonging to $U$ occurs prior to the termination of $I''$, which is before the start of $I'$.
By Corollary~\ref{cor-state-transitions-respect-figure}, $r\info.state $ does not change after it is set to \done \ or \retry \ by this \cstep \ or \astep.
Since $I'$ returns \fail, when $I'$ performs line~\ref{ll-check-finalized}, either $r\info.state = \done$ and $marked_1=\false$, or $r\info.state = $ \retry.

Suppose $r\info.state = \retry$ when $I'$ performs line~\ref{ll-check-finalized}.
Then, $r\info.state$ was \retry\ when $I'$ performed line~\ref{ll-read-state}, and $I'$ passed the test at line~\ref{ll-check-frozen}, and entered the if-block.
Since we have assumed that $I'$ blames $S$, $I'$ must read a pointer to $U$ from $r.\info$ when it performs line~\ref{ll-reread}.
However, since $I'$ returns \fail, 
the value $I'$ reads from $r.\info$ at line~\ref{ll-reread} is different from the value read at line~\ref{ll-read}, which is a contradiction.
Hence, this case is impossible.

\after{I got lost in the following paragraph (Eric)}
Now, suppose $r\info.state = \done$ and $marked_1=\false$ when $I'$ performs line~\ref{ll-check-finalized}.
Then, $r\info.state$ was \done\ at all times since $I'$ began, so $state = \done$.
If $marked_2$ was also \false\ when $I'$ performed line~\ref{ll-check-frozen}, then $I'$ passed the test at line~\ref{ll-check-frozen}, and entered the if-block, so we obtain the same contradiction as above.
Hence, this case, too, is impossible.
Thus, $marked_2$ must have been \true\ when $I'$ performed line~\ref{ll-check-frozen}, which means $r$ was marked when $I'$ performed line~\ref{ll-read-marked2}.
Since a \cstep\ belonging to $U$ had already occurred when $I'$ performed line~\ref{ll-read-marked1}, the code of \help \ implies that, if $r$ were in $U.R$, then the first \markstep\ belonging to $U$ on $r$ would already have occurred before $I'$ performed line~\ref{ll-read-marked1}.
Since a $marked$ bit is never changed from \true \ to \false, $r$ would be marked when $I'$ performed line~\ref{ll-read-marked1}, which is incompatible with our assumption that $marked_1=\false$.
\after{This last sentence doesn't make sense.  You assumed $marked_1=\false$ inside $I'$ not $I$! (Eric)}
Therefore, $r$ is not in $U.R$, so no \markstep\ belonging to $U$ on $r$ can ever occur.
Since $r$ was marked when $I'$ performed line~\ref{ll-read-marked2}, there must have been a successful \markstep\ $mstep$ belonging to some other \op \ $W$ on $r$ after $I'$ performed line~\ref{ll-read-marked1}, and before $I'$ performed line~\ref{ll-read-marked2}.
By Lemma~\ref{lem-frozen-forever-after-markstep}, $r.\info$ points to $W$ when $mstep$ occurs.
However, since $r.\info$ points to $U$ during $I''$, which is before $mstep$, and again during $I$, which is after $mstep$, at some point after $mstep$, $r.\info$ must be changed to a value ($U$) that has previously appeared there, which contradicts Lemma~\ref{lem-no-aba-info}.
Thus, it is impossible for three or more invocations of \llt\ by $p$ to blame the same invocation of \sct.
\end{proof}


\begin{lem} \label{lem-sct-only-blamed-by-v-scts-or-vlts-per-process}
Each invocation of \sct$(V, R, fld, new)$ can be blamed by at most $|V|$ invocations of \sct \ or \vlt \ per process.
\end{lem}
\begin{proof}
By Observation~\ref{obs-if-sct-or-vlt-blames-for-r-then-r-in-blamed-V}, if an invocation of \sct \ or \vlt \ blames an invocation of \sct $ (V, R, fld, new) $ for $r$, then $r$ is in $V$.
Thus, it suffices to prove that an invocation $S$ of \sct$ (V, R, fld, new) $
cannot be blamed for any $r$ in $V$ by more than one invocation of \sct \ or \vlt\ performed by process $p$. 

Let $I$ and $I'$ be invocations of \sct \ or \vlt \ performed by process $p$, and $U$, $U'$ and $U_S$ be the \op s created by $I$, $I'$ and $S$, respectively.
Without loss of generality, let $I'$ occur after $I$.
Suppose, in order to derive a contradiction, that $I$ and $I'$ both blame 
$S$ 
for the same \rec \ $r$.
Let $t_0$ ($t_0'$) be the time when the \llt$(r)$ linked to $I$ ($I'$) performs line~\ref{ll-read}, and $t_1$ ($t_1'$) be the time when $I$ ($I'$) finishes.
By Lemma~\ref{lem-if-sct-or-vlt-blames-then-succ-fcas}, a successful \fcas \ belonging to $U_S$ occurs between $t_0$ and $t_1$, and a successful \fcas \ belonging to $U_S$ occurs between $t_0'$ and $t_1'$.
If we can show $ t_0 < t_1 < t_0' < t_1' $, then we shall have demonstrated that there must be two such \fcas s, which contradicts Lemma~\ref{lem-only-first-fcas-can-succeed}.

Since the \llt$(r)$ linked to $I$ ($I'$) terminates before $I$ ($I'$), we know $t_0 < t_1$ ($t_0' < t_1'$).
By Observation~\ref{obs-if-sct-or-vlt-blames-for-r-then-r-in-blaming-V}, $r$ is in the $V$ sequences of invocations $I$ and $I'$.
Hence, Definition~\ref{defn-llt-linked-to-sct}.\ref{prop-no-sct-or-vlt-between-linked-llt-and-sct-or-vlt} implies that $t_1 \notin [t_0', t_1']$, and $t_1' \notin [t_0, t_1]$.
Since $I'$ occurs after $I$, $ t_1 < t_1'$.
Therefore, $ t_0 < t_1 < t_0' < t_1'$.
\end{proof}

We now define the blame graph and prove a number of its properties.

\begin{defn} \label{defn-blame-graph}
We define the \textbf{blame graph} for an execution to be a directed graph whose nodes are the invocations of \llt, \vlt\ and \sct, with an edge from an invocation $I$ to another invocation $I'$ if and only if $I$ blames $I'$.  (Note that only the nodes corresponding to invocations of \sct\ can have incoming edges.)
\end{defn}

The next property we prove is that, for each execution, there is a bound on the length of the longest path in the blame graph.
%
As mentioned in Section \ref{sec-operations}, we require the following constraint in order to prove this bound exists.

\begin{con} \label{con-partial-order}
If there is a configuration $C$ after which the value of no
field of any \rec\ changes,
then
 there is a total order $\prec$
on all \rec s created during the execution such that,
if \rec\ $r_1$ appears before data
\rec\ $r_2$ in the sequence $V$ passed to an invocation 
of \sct\ whose linked \llt s begin after $C$,
then $r_1 \prec r_2$.
\end{con}

\begin{lem} \label{lem-after-fcas-previous-recs-frozen}
Let $U$ be an \op\ created by an invocation of \sct\ whose linked \llt s begin after the configuration $C$ that is specified in Constraint~\ref{con-partial-order}.
Immediately after a successful \fcas\ belonging to $U$ on $r$, $r.\info$ points to $U$ and, for each $ r'$ in $U.V $, where $r' \prec r$, $r'.\info$ points to $U$ and a successful \fcas \ belonging to $U$ on $r'$ has occurred.
\end{lem}
\begin{proof}
Let {\it fcas} be a successful \fcas \ belonging to $U$ on $r$, and let $r'$ be any \rec\ in $U.V$ that satisfies $ r' \prec r $.
By Constraint~\ref{con-partial-order}, $r'$ must occur before $r$ in the sequence $U.V$.
By Lemma~\ref{lem-fcas-on-ri-only-after-succ-fcas-on-previous}, a successful \fcas\ {\it fcas$'$} belonging to $U$ on $r'$ occurs prior to {\it fcas}.
Thus, Corollary~\ref{cor-if-succ-fcas-then-point-u-until-bcas-or-uass} implies that $r'.\info$ points to $U$ at all times after {\it fcas$'$} and before the first \cstep \ or \astep \ belonging to $U$.
Similarly, $r.\info$ points to $U$ at all times after {\it fcas} and before the first \cstep \ or \astep \ belonging to $U$.
By Lemma~\ref{lem-no-succ-fcas-after-fass-or-bcas}, {\it fcas} must precede the first \fstep\ or \astep \ belonging to $U$.
From the code of \help, the first \fstep\ belonging to $U$ must precede the first \cstep\ belonging to $U$.
Hence, {\it fcas} and {\it fcas$'$} both precede the first \cstep \ or \astep \ belonging to $U$.
Therefore, immediately after {\it fcas}, the \info\ fields of $r$ and $r'$ both point to $U$.
\end{proof}

\begin{lem} \label{lem-if-blame-chain-then-recs-ordered}
Let $U_1$, $U_2$ and $U_3$ be \op s respectively created by invocations $S_1$, $S_2$ and $S_3$ of \sct\ whose linked \llt s begin after the configuration $C$ that is specified in Constraint~\ref{con-partial-order}, and $r$ and $r'$ be \rec s.
If $S_1$ blames $S_2$ for $r$, and $S_2$ blames $S_3$ for $r'$, then $r \prec r'$.
\end{lem}
\begin{proof}
Since $S_1$ blames $S_2$ for $r$, we know from Definition~\ref{defn-blame-scx} that $r$ is in $U_2.V$.
Similarly, since $S_2$ blames $S_3$ for $r'$, we know $r'$ is in $U_2.V$.
Furthermore, a successful \fcas\ belonging to $U_2$ on $r$ occurs, and no successful \fcas\ belonging to $U_2$ on $r'$ occurs.
By Lemma~\ref{lem-fcas-on-ri-only-after-succ-fcas-on-previous}, $r$ must occur before $r'$ in the sequence $U_2.V$.
Thus, Constraint~\ref{con-partial-order} implies $r \prec r'$.
\end{proof}

\begin{lem} \label{lem-one-to-one-correspondence-upcas-and-sct}
There can be only as many successful \upcas s as there are invocations of \sct \ that either return \true, or do not terminate.
\end{lem}
\begin{proof}
From the code, an \upcas \ can only occur in an invocation of \help$ (ptr) $, where $ptr$ points to an \op \ $U$.
Further, from the code of \help, there is at least one \fcas \ belonging to $U$ or \fstep \ belonging to $U$.
Hence, Lemma~\ref{lem-no-steps-belong-to-dummy-op} implies that $U$ is not the dummy \op.
Thus, $U$ is created by an invocation of \sct \ at line~\ref{sct-create-op}.
By Lemma~\ref{lem-only-first-upcas-can-succeed}, only the first \upcas \ belonging to an \op \ can succeed.
By Lemma~\ref{lem-help-false}.\ref{claim-help-false-no-upcas}, no \upcas \ belongs to an \op\ created by an unsuccessful invocation of \sct.
\end{proof}

\after{I'm still a bit uncomfortable about what entry points really mean, formally (Eric).}

We think of processes as accessing such a data structure via
a fixed number of special \rec s called \textit{entry points}, each of which has a single mutable pointer to a \rec.
We assume there is always some \rec\ reachable by following pointers from an entry point that is not finalized. 
(This assumption that entry points cannot be finalized is not crucial, but it 
simplifies the statement of some progress guarantees.)

\begin{defn} \label{defn-rec-initiated}
A \rec \ is \textbf{initiated} at all times after it first becomes reachable by following \rec \ pointers from an entry point.
\end{defn}

\begin{obs} \label{obs-only-upcas-can-initiate}
The only step in an execution that can cause a \rec\ to become initiated is a successful \upcas.
\end{obs}
\begin{proof}
Follows immediately from Observation~\ref{obs-only-upcas-modifies-records}.
\end{proof}

\begin{lem} \label{lem-if-blame-for-r-then-r-already-initiated}
Let $S_1$ and $S_2$ be invocations of \sct, and let $r$ be a \rec.
If $S_1$ blames $S_2$ for $r$, then $r$ was initiated before the start of $S_1$, and before the start of $S_2$.
\end{lem}
\begin{proof}
Let $U_1$ and $U_2$ be the \op s created by $S_1$ and $S_2$, respectively.
By Definition~\ref{defn-blame-scx}, $r$ is in both $U.V$ and $U'.V$.
By Observation~\ref{obs-op-invariants}.\ref{inv-llresults}, there are invocations of \llt$(r)$ linked to $S_1$ and $S_2$, respectively.
By the precondition of \llt, $r$ must be initiated before the \llt$(r)$ linked to $S_1$, and before the \llt$(r)$ linked to $S_2$.
Finally, the \llt$(r)$ linked to $S_1$ must terminate before $S_1$ begins, and the \llt$(r)$ linked to $S_2$ must terminate before $S_2$ begins.
\end{proof}

\after{Instead of defining $\sigma$, define its complement, since we mostly
talk about things {\it not} in $\sigma$}

\begin{lem} \label{lem-if-no-succ-sct-after-time-t-then}
If no \sct\ is linearized after some time $t$, then the following hold.
\begin{enumerate}
\item A finite number $N$ 
of \rec s  are ever initiated in the execution.%
\label{claim-if-no-succ-sct-after-time-t-then-finite-number-of-initiated-rec}
\item Let $\sigma$ be the set of invocations of \sct \ in the execution whose vulnerable intervals start at or before $t$.  The longest path in the blame graph consisting entirely of invocations of \llt, \sct, and \vlt \ that are not in $\sigma$ has length at most $N+2$
\label{claim-if-no-succ-sct-after-time-t-then-no-path-longer-than}
\end{enumerate}
\end{lem}
\begin{proof}
Claim~\ref{claim-if-no-succ-sct-after-time-t-then-finite-number-of-initiated-rec} follows immediately from Observation~\ref{obs-only-upcas-can-initiate} and Lemma~\ref{lem-one-to-one-correspondence-upcas-and-sct}.

We now prove claim~\ref{claim-if-no-succ-sct-after-time-t-then-no-path-longer-than}.
Suppose, in order to derive a contradiction, that there is a path of length at least $N +3$ in the blame graph consisting entirely of invocations of \llt, \sct, and \vlt \ that are not in $\sigma$.
Since only invocations of \sct \ can be blamed, at least $N+2 $ of the nodes on this path must correspond to invocations of \sct.
Let $ S_1, S_2,..., S_{N+2}$ be invocations of \sct \ corresponding to any $N+2 $ consecutive nodes on this path, and let $U_1, U_2, ..., U_{N+2}$ be the \op s they created, respectively.
For each $ i \in\{ 1, 2,...,N+1\} $, let $r_i$ be the \rec \ for which $S_i$ blames $S_{i+1}$.
Since no invocation of \sct \ is linearized after $t$, and the vulnerable sections of $ S_1, S_2,..., S_{N+2}$ all start after $t$, no invocation of \sct\ is linearized after the first $\llt(r)$ linked to any of these invocations of \sct.
Therefore, from Lemma~\ref{lem-if-blame-chain-then-recs-ordered} and the fact that, for each $i \in \{1, 2, ..., N\}$, $S_i$ blames $S_{i+1}$ for $r_i$ and $S_{i+1}$ blames $S_{i+2}$ for $r_{i+1}$, we obtain $r_i \prec r_{i+1}$.
%
By Lemma~\ref{lem-if-blame-for-r-then-r-already-initiated}, before any invocation of \sct \ in \{$ S_1, S_2,..., S_{N+2} $\} begins, $r_1, r_2, ..., r_{N+1}$ have all been initiated.
Therefore, some \rec \ $r$ appears twice in $\{r_1, r_2, ..., r_{N+1}\}$.
Since the $\prec$ relation is transitive, we obtain $r \prec r$, which is a contradiction.
\end{proof}

We now prove the main progress property for \sct.

\begin{lem} \label{lem-sct-progress}
If invocations of \sct \ complete infinitely often, then invocations of \sct \ succeed infinitely often.
\end{lem}
\begin{proof}
Suppose, to derive a contradiction, that after some time $t'$, invocations of \sct \ are performed infinitely often, but no invocation of \sct \ is successful.
Then, since we only linearize successful \sct s, and a subset of the non-terminating \sct s, there is a time $t \ge t'$ after which no invocation of \sct\ is linearized.
Let $\sigma$ be the set of invocations of \sct \ in the execution whose vulnerable intervals start at or before $t$. 
By Lemma~\ref{lem-sct-only-blamed-by-one-llt-per-process} and Lemma~\ref{lem-sct-only-blamed-by-v-scts-or-vlts-per-process},  the in-degree of each node in the blame graph is bounded.
Since $\sigma$ is finite, and the in-degree of each node in $\sigma$ is bounded, only a finite number of invocations of \sct \ can blame invocations in $\sigma$.
Now, consider any maximal path $\pi$ consisting entirely of invocations of \llt, \sct, and \vlt \ that are \textit{not} in $\sigma$.
By Lemma~\ref{lem-if-no-succ-sct-after-time-t-then}.\ref{claim-if-no-succ-sct-after-time-t-then-no-path-longer-than}, $\pi$ has length at most $N+3$.
Since no invocation of \sct \ is successful after $t$, the invocation $S$ of \sct \ corresponding to the last node on path $\pi$ must be unsuccessful.
By Lemma~\ref{lem-if-abort-then-blame-different-scx}, $S$ must blame some other invocation of \sct.
Since $\pi$ is maximal, $S$ must blame an invocation of \sct \ in $\sigma$.
Thus, there can be only finitely many of these paths (of bounded length).
However, this contradicts our assumption that invocations of \sct \ occur infinitely often.
\end{proof}

Unfortunately, since \sct\ cannot be invoked unless a sequence of invocations of \llt\ (linked to this \sct) return values different from \fail\ or \finalized, the previous result is not strong enough unless we can guarantee that processes can invoke \sct\ infinitely often.
To address this, we define \textit{setting up} an invocation of \sct.

\begin{defn} \label{defn-set-up-sct}
A process $p$ \textbf{sets up} an invocation of \sct$(V, R, fld, new)$ by invoking \llt$(r)$ for each $r$ in $V$, and then invoking \sct$(V, R, fld, new)$ if none of these \llt s return \fail\ or \finalized.
\end{defn}

\begin{thm} \label{thm-llt-sct-vlt-progress}
Our implementation of \llt/\sct/\vlt\ satisfies the following progress properties.
\begin{enumerate}
\item If operations (\llt, \sct, \vlt) are performed infinitely often, then operations succeed infinitely often.
\label{claim-progress-if-operations-io-then-succ-io}
\item If invocations of \sct \ are set up infinitely often, then invocations of \sct \ succeed infinitely often.
\label{claim-progress-if-set-up-sct-io-then-succ-io}
\item If there is always some \rec\ reachable by following pointers from an entry point that is not finalized, then invocations of \sct \ can be set up infinitely often.
\label{claim-progress-if-not-finalized-then-can-set-up-sct-io}
\end{enumerate}
\end{thm}
\begin{proof}
Claim~\ref{claim-progress-if-not-finalized-then-can-set-up-sct-io} is obvious.
The first two claims have similar proofs, by cases.

\textbf{Proof of Claim~\ref{claim-progress-if-operations-io-then-succ-io}.}
Suppose operations are performed infinitely often.

\textbf{Case I:} invocations of \sct \ are performed infinitely often.
In this case, Lemma~\ref{lem-sct-progress} implies that invocations of \sct \ will succeed infinitely often, and the claim is proved.

\textbf{Case II:} after some time $t$, no invocation of \sct \ is performed.
In this case, the blame graph contains a finite number of invocations of \sct.
By Lemma~\ref{lem-sct-only-blamed-by-one-llt-per-process} and Lemma~\ref{lem-sct-only-blamed-by-v-scts-or-vlts-per-process},  the in-degree of each node in the blame graph is bounded.
By Lemma~\ref{lem-if-llt-fail-then-blame-sct} and Lemma~\ref{lem-if-vlt-false-then-blame-sct}, each unsuccessful invocation of \llt \ or \vlt \ blames an invocation of \sct.
Therefore, only finitely many invocations of \llt \ and \vlt \ can be unsuccessful.
Thus, eventually, every invocation of \llt\ or \vlt\ succeeds.

\textbf{Proof of Claim~\ref{claim-progress-if-set-up-sct-io-then-succ-io}.}
Suppose invocations of \sct \ are set up infinitely often.

\textbf{Case I:} invocations of \sct \ are performed infinitely often.
In this case, Lemma~\ref{lem-sct-progress} implies that invocations of \sct \ will succeed infinitely often, and the claim is proved.

\textbf{Case II:} after some time $t$, no invocation of \sct \ is performed.
Suppose, to derive a contradiction, that after some time $t$, \sct \ is never invoked.
Then, as we argued in Case II, above, only finitely many invocations of \llt \ and \vlt \ can be unsuccessful.
This implies that, after some time $t'$, every invocation of \llt \ is successful.
If a process $p$ begins setting up an invocation of \sct \ after $t'$, then all of its invocations of \llt \ will be successful, and the only way that $p$ will not invoke \sct \ is if some invocation of \llt$(r)$ by $p$ returns \finalized.
By Definition~\ref{defn-set-up-sct}, $p$ will not invoke \llt$(r)$ next time it sets up an invocation of \sct.
By Lemma~\ref{lem-if-no-succ-sct-after-time-t-then}, there are a finite number $N$ of \rec s that are ever initiated in the execution.
Therefore, eventually, $p$ will have performed an invocation of \llt$(r')$ that returned \finalized\ for every \rec\ $r'$ that is ever initiated in the execution.
After this, $p$ will no longer be able to set up invocations of \sct.
This contradicts our assumption that invocations of \sct \ are set up infinitely often.
\end{proof}

\newpage
\section{Additional properties of \llt/\sct/\vlt} \label{sec-properties}

In this section we prove some additional properties of \llt/\sct/\vlt\ that are intended to simplify the design of certain data structures.
At this level, a \textit{configuration} consists of the state of each process, and a collection of \rec s (which have only mutable and immutable fields).
A \textit{step} is either a \func{Read}, or a linearized invocation of \llt, \sct\ or \vlt.

\begin{defn} \label{defn-rec-in-added-removed}
A \rec \ $r$ is \textbf{in the data structure} in some configuration $C$ if and only if $r$ is reachable by following pointers from an entry point.
We say a \rec\ $r$ is \textbf{removed (from the data structure) by} some step $s$ if and only if $r$ is in the data structure immediately before $s$, and $r$ is not in the data structure immediately after $s$.
We say a \rec \ $r$ is \textbf{added (to the data structure) by} some step $s$ if and only if $r$ is not in the data structure immediately before $s$, and $r$ is in the data structure immediately after $s$.
\end{defn}

Note that a \rec \ can be removed from or added to the data structure only by a linearized invocation of \sct.

If the following constraint is satisfied, then the results of this section apply.

\begin{con} \label{con-mark-all-removed-recs}
For each linearized invocation $S$ of \sct$(V, R, fld, new)$, $R$ contains precisely the \rec s that are removed from the data structure by $S$.
\end{con}

\begin{lem} \label{lem-if-initiated-rec-not-in-data-structure-then-does-not-change}
If a \rec \ $r$ is removed from the data structure for the first time by step $s$, then no linearized invocation of \sct$(V, R, fld, new)$, where $fld$ is a mutable field of $r$, occurs at or after $s$.
(Hence, $r$ does not change at or after $s$.)
\end{lem}
\begin{proof}
The only step that can change $r$ is a linearized invocation of \sct.
The invocation $S'$ of \sct$(V',$ $R',$ $fld',$ $new')$ that removes $r$ modifies a mutable field of some \rec\ different from $r$.
Thus, $fld'$ is not a field of $r$.
Since this is the only change to the data structure when $s$ occurs, $r$ does not change when $s$ occurs.
Suppose, to derive a contradiction, that an invocation $S$ of \sct$(V, R, fld, new)$, where $fld$ is a mutable field of $r$, occurs after $s$.
Then, since $r$ is in $V$, the precondition of \sct\ implies that an invocation $I$ of $\llt(r)$ linked to $S$ must occur before $S$.
By Constraint~\ref{con-mark-all-removed-recs}, $r$ is in $R'$.
Thus, if $I$ occurs after $S'$, then it returns \finalized, which contradicts Definition~\ref{defn-llt-linked-to-sct}.
Otherwise, $S'$ occurs between $I$ and $S$, so $S$ cannot be linearized, which contradicts our assumption.
\end{proof}

\begin{lem} \label{lem-rec-in-data-structure-just-before-llt}
If an invocation $I$ of $\llt(r)$ returns a value different from \fail\ or \finalized, then $r$ is in the data structure just before $I$ is linearized.
\end{lem}
\begin{proof}
By the precondition of \llt, $r$ is initiated and, hence, in the data structure, at some point before $I$.
Suppose, to derive a contradiction, that $r$ is not in the data structure just before $I$ is linearized.
Then, some linearized invocation of \sct$(V, R, fld, new)$ must remove $r$ before $I$ is linearized.
By Constraint~\ref{con-mark-all-removed-recs}, $r$ is in $R$.
However, this implies that $I$ must return \finalized, which is a contradiction.
\end{proof}

\begin{lem} \label{lem-rec-in-data-structure-after-linearized-sct}
If $S$ is a linearized invocation of \sct$(V, R, fld, new)$, where $new$ is a \rec, then $new$ is in the data structure just after $S$.
\end{lem}
\begin{proof}
Note that $fld$ is a mutable field of a \rec\ $r$ in $V$.
We first show that $r$ is in the data structure at some point before $S$.
By the precondition of \sct, before $S$, there is an $\llt(r)$ linked to $S$.
By the precondition of \llt, $r$ must be initiated when this linked \llt\ occurs.
Thus, Definition~\ref{defn-rec-initiated} and Definition~\ref{defn-rec-in-added-removed} imply that $r$ is in the data structure at some point before $S$.
Suppose, to derive a contradiction, that $r$ is not in the data structure just after $S$.
Then, $r$ must either be removed by $S$, or by some previous step.
However, this directly contradicts Lemma~\ref{lem-if-initiated-rec-not-in-data-structure-then-does-not-change}.
\end{proof}

Let $C_1$ and $C_2$ be configurations in the execution.
We use $C_1 < C_2$ to mean that $C_1$ precedes $C_2$ in the execution.
We say $C_1 \le C_2$ precisely when $C_1 = C_2$ or $C_1 < C_2$.
We denote by $[C_1, C_2]$ the set of configurations $\{C \mid C_1 \le C \le C_2\}$.

\after{I think the proof of the following lemma could be simplified by first proving
a version of the second part of the claim (for any field f), then taking the special case of a pointer
field to prove the first part.  May require some rewording of the statement, though.}

\begin{lem} \label{lem-if-rec-traversed-then-rec-in-data-structure}
Let $r_1,r_2,...,r_l$ be a sequence of \rec s, where $r_1$ is an entry point, and $C_1,C_2,...,C_{l-1}$ be a sequence of configurations satisfying $C_1 < C_2 < ... < C_{l-1}$.
If, for each $i \in \{1, 2, ..., l-1\}$, a field of $r_i$ points to $r_{i+1}$ in configuration $C_i$, then $r_{i+1}$ is in the data structure in some configuration in $[C_1, C_i]$.
Additionally, if a mutable field $f$ of $r_l$ contains a value $v$ in some configuration $C_l$ after $C_{l-1}$ then, in some configuration in $[C_1, C_l]$, $r_l$ is in the data structure and $f$ contains $v$.
\end{lem}
\begin{proof}
We prove the first part of this result by induction on $i$.

Since each entry point is always in the data structure, and $r_1$ points to $r_2$ in configuration $C_1$, $r_2$ is in the data structure in $C_1$.
Thus, the claim holds for $i=1$.

Suppose the claim holds for $i$, $1 \le i \le l-2$.
We prove it holds for $i+1$.
If $r_i$ is in the data structure when it points to $r_{i+1}$ in $C_i$, then $r_{i+1}$ is in the data structure in $C_i$, and we are done.
Suppose $r_i$ is \textit{not} in the data structure in $C_i$.
By the inductive hypothesis, $r_i$ is in the data structure in some configuration in $[C_1, C_{i-1}]$.
Let $s$, $C_1 < s < C_i$, be the first step such that $r_i$ is removed from the data structure by $s$.
In the configuration $C$ just before $s$, $r_i$ is in the data structure.
By Lemma~\ref{lem-if-initiated-rec-not-in-data-structure-then-does-not-change}, $r_i$ does not change at or after $s$. 
Thus, $r_i$ does not change after $C$.
Since $C$ occurs before $C_i$, and $r_i$ points to $r_{i+1}$ in $C_i$, $r_i$ must point to $r_{i+1}$ in $C$.
Therefore, in $C$ (which satisfies $C_1 \le C < C_i$), $r_i$ is in the data structure and points to $r_{i+1}$.

The second part of the proof is quite similar to the inductive step we just finished.
Suppose $f$ contains $v$ in $C_l$.
If $r_l$ is in the data structure in $C_l$, then we are done.
Suppose $r_l$ is not in the data structure in $C_l$.
We have shown above that $r_l$ is in the data structure in some configuration in $[C_1, C_l]$.
Let $s'$, $C_1 < s' < C_l$, be the first step such that $r_l$ is removed from the data structure by $s'$.
In the configuration $C'$ just before $s'$, $r_l$ is in the data structure.
By Lemma~\ref{lem-if-initiated-rec-not-in-data-structure-then-does-not-change}, $r_l$ does not change at or after $s'$. 
Thus, $r_l$ does not change after $C'$.
Since $C'$ occurs before $C_l$, and $f$ contains $v$ in $C_l$, $f$ must contain $v$ in $C'$.
Therefore, in $C'$ (which satisfies $C_1 \le C' < C_l$), $r_l$ is in the data structure and $f$ contains $v$.
\end{proof}

\newpage
\section{Complete Proof of Multiset Implementation}
\label{app-multiset}

The full pseudocode for the multiset algorithm appears in Fig.~\ref{code-list}.
Initially, the data structure contains a $head$ entry point, containing a single mutable field $next$ that points to a sentinel \listrec\ with a special key $\infty$ that is larger than any key that can appear in the multiset.

In the following, we define the \textbf{response} of a \search\ to be a step at which a value is returned.
Note that we specify Lemma~\ref{lem-multiset-constraints-invariants}.\ref{claim-multiset-finalized-before-removed}, instead of directly proving the considerably simpler statement in Constraint~\ref{con-mark-all-removed-recs}, so that we can reuse the intermediate results when proving linearizability.

\begin{lem} \label{lem-multiset-constraints-invariants}
The multiset algorithm satisfies the following properties.
\begin{enumerate}
\item    Every invocation of \llt \ or \sct \ has valid arguments, and satisfies its preconditions.
\label{claim-multiset-llt-sct-preconditions}
\item    Every invocation of \search\ satisfies its postconditions.
\label{claim-multiset-search-postconditions}
%
\item    Let $S$ be an invocation of \sct$(V, R, fld, new)$ performed by an invocation $I$ of \ins\ or \del, and $p$, $r$ and $rnext$ refer to the local variables of $I$.
         If $I$ performs $S$ at line~\ref{multiset-insert-sct1}, then no \rec\ is added or removed by $S$, and $R = \emptyset$.
         If $I$ performs $S$ at line~\ref{multiset-insert-sct2}, then only $new$ is added by $S$, no \rec\ is removed by $S$, and $R = \emptyset$.
         If $I$ performs $S$ at line~\ref{multiset-delete-sct1}, then only $new$ is added by $S$, only $r$ is removed by $S$, and $R = \{r\}$.
         If $I$ performs $S$ at line~\ref{multiset-delete-sct2}, then only $new$ is added by $S$, only $r$ and $rnext$ are removed by $S$, and $R = \{r, rnext\}$.
\label{claim-multiset-finalized-before-removed}
\item    The $head$ entry point always points to a \listrec, the $next$ pointer of each \listrec\ with $key \neq \infty$ points to some \listrec\ with a strictly larger key, and the $next$ pointer of each \listrec\ with $key = \infty$ is \nil. 
\label{claim-multiset-sorted-list}
\end{enumerate}
\end{lem}
\begin{proof}
We prove these claims by induction on the sequence of steps taken in the execution.
The only steps that can affect these claims are invocations of \llt \ and \sct, and responses of \search es.
\textbf{Base case.}
Clearly, 
Claim~\ref{claim-multiset-llt-sct-preconditions}, Claim~\ref{claim-multiset-search-postconditions} 
and Claim~\ref{claim-multiset-finalized-before-removed} hold before any such step occurs.
Before the first \sct, the data structure is in its initial configuration.
Thus, 
Claim~\ref{claim-multiset-sorted-list} holds before any step occurs.
\textbf{Inductive step.}
Suppose these claims hold before some step $s$.
We prove they hold after $s$.

\textbf{Proof of Claim~\ref{claim-multiset-llt-sct-preconditions}.}
The only steps that can affect this claim are invocations of \llt \ and \sct.

Suppose $s$ is an invocation of \llt.
By inductive Claim~\ref{claim-multiset-finalized-before-removed} and Observation~\ref{obs-multiset-satisfies-con-mark-all-removed-recs}, Constraint~\ref{con-mark-all-removed-recs} is satisfied at all times before $s$ occurs.
The only places in the code where $s$ can occur are at lines~\ref{multiset-insert-llt-r}, \ref{multiset-insert-llt-p}, \ref{multiset-delete-llt-p}, \ref{multiset-delete-llt-r} and \ref{multiset-delete-llt-rnext}.
Suppose $s$ occurs at line~\ref{multiset-insert-llt-r}, \ref{multiset-insert-llt-p}, \ref{multiset-delete-llt-p} or \ref{multiset-delete-llt-r}.
Then, by inductive Claim~\ref{claim-multiset-search-postconditions}, argument to $s$ is non-\nil.
We can apply Lemma~\ref{lem-if-rec-traversed-then-rec-in-data-structure} to show that the argument to $s$ is in the data structure and, hence, initiated, at some point during the last \search\ before $s$.
Now, suppose $s$ occurs at line~\ref{multiset-delete-llt-rnext} (so $localr.next$ is the argument to $s$).
Then, $key = r.key$ when line~\ref{multiset-delete-return-false} is performed so, by the precondition of \del, $r.key \neq \infty$.
By inductive Claim~\ref{claim-multiset-sorted-list}, $r.next \neq \nil$ when \llt$(r)$ is performed at line~\ref{multiset-delete-llt-r}, so $localr.next \neq \nil$.
We can apply Lemma~\ref{lem-if-rec-traversed-then-rec-in-data-structure} to show that $localr.next$ is in the data structure and, hence, initiated, at some point between the start of the last \search\ before $s$ and the last \llt$(r)$ before $s$ (which reads $localr.next$ from $r.next$).

Suppose $s$ is a step that performs an invocation $S$ of \sct$(V, R, fld, new)$.
Then, the only places in the code where $s$ can occur are at lines~\ref{multiset-insert-sct1}, \ref{multiset-insert-sct2}, \ref{multiset-delete-sct1} and \ref{multiset-delete-sct2}.
It is a trivial exercise to inspect the code of \ins\ and \del, and argue that the process that performs $s$ has done an $\llt(r)$ linked to $S$ for each $r \in V$, that $R \subseteq V$, and that $fld$ points to a mutable field of a \rec \ in $V$.
It remains to prove that Precondition~\presctinfo\ and Precondition~\presctfld\ of \sct\ are satisfied.
Let $I$ be the invocation of $\llt(r)$ linked to $S$.
Suppose $s$ occurs at line~\ref{multiset-insert-sct1}.
The only step that can affect the claim is a linearized invocation $s$ of \sct$(V, R, fld, new)$ performed at line~\ref{multiset-insert-sct1}.
From the code, $fld$ is $r.count$.
Since $s$ is linearized, no invocation of \sct$(V'', R'', fld'', new'')$ with $r \in V''$ is linearized between $I$ and $s$.
Thus, $r.count$ does not change between when $I$ and $s$ are linearized.
From the code, $new$ is $count$ plus the value read from $r.count$ by $I$.
Therefore, $new$ is strictly larger than $r.count$ was when $I$ was linearized, which implies that $new$ is strictly larger than $r.count$ when $s$ is linearized.
This immediately implies Precondition~\presctinfo\ and Precondition~\presctfld\ of \sct. 
Now, suppose $s$ occurs at line~\ref{multiset-insert-sct2}, \ref{multiset-delete-sct1} or \ref{multiset-delete-sct2}.
Then, $new$ is a pointer to a \listrec\ that was created after $I$.
Thus, no invocation of \sct$(V', R', fld, new)$ can even \textit{begin} before $I$.
We now prove that $new$ is not the initial value of the field pointed to by $fld$.
From the code, $fld$ is $p.next$.
If $p.next$ is initially \nil, then we are done.
Otherwise, $p.next$ initially points to some \listrec\ $r'$.
Clearly, $r'$ must be created before $p$.
Hence, $r'$ must be created before the invocation of \search\ followed a pointer to $p$.
Since $new$ is a pointer to a \listrec\ that is created after this invocation of \search, $new \neq r'$.

\textbf{Proof of Claim~\ref{claim-multiset-search-postconditions}.}
To affect this claim, $s$ must be the response of an invocation of \search$(key)$.
We prove a loop invariant that states $r$ is a \listrec, and either $p$ is a \listrec\ and $p.key < key$ or $p = head$.
Before the loop, $p = head$ and $r = head.next$.
By inductive Claim~\ref{claim-multiset-sorted-list} $head.next$ is always a \listrec, so the claim holds before the loop.
Suppose the claim holds at the beginning of an iteration.
Let $r$ and $p$ be the respective values of local variables $r$ and $p$ at the beginning of the iteration, and $r'$ and $p'$ be their values at the end of the iteration.
From the code, $p' = r$ and $r'$ is the value read from $r.next$ at line~\ref{multiset-search-advance-r}.
By the inductive hypothesis, $p'$ is a \listrec.
Since the loop did not exit before this iteration, $key > p'.key$.
Further, since \search$(key)$ is invoked only when $key < \infty$ (by inspection of the code and preconditions), $p'.key < \infty$.
By inductive Claim~\ref{claim-multiset-sorted-list}, $p'.next = r.next$ always points to a \listrec, so $r'$ is a \listrec, and the inductive claim holds at the end of the iteration.
Finally, the exit condition of the loop implies $key \le r'.key$, so \search\ satisfies its postcondition.

\textbf{Proof of Claim~\ref{claim-multiset-finalized-before-removed}.}
Since a \rec \ can be removed from the data structure only by a change to a mutable field of some other \rec, this claim can be affected only by linearized invocations of \sct.
Suppose $s$ is a linearized invocation of \sct$(V, R, fld, new)$.
Then, $s$ can occur only at line~\ref{multiset-insert-sct1}, \ref{multiset-insert-sct2}, \ref{multiset-delete-sct1} or \ref{multiset-delete-sct2}.
Let $I$ be the invocation of \ins\ or \del\ in which $s$ occurs.
We proceed by cases.

Suppose $s$ occurs at line~\ref{multiset-insert-sct1}.
Then, $fld$ is a pointer to $r.count$. 
Thus, $s$ changes a $count$ field, \textit{not} a $next$ pointer.
Since this is the only change that is made by $s$, no \rec\ is removed by $s$, and no \rec\ is added by $s$.
Since $R = \emptyset$, the claim holds.

Suppose $s$ occurs at line~\ref{multiset-insert-sct2}.
Then, $fld$ is a pointer to $p.next$, and $new$ is a pointer to a new \listrec.
Before performing $s$, $I$ performs an invocation $L_1$ of \llt$(p)$, which returns a value different from \fail, or \finalized, at line~\ref{multiset-delete-llt-p}.
Just after performing $L_1$, $I$ sees that $localp.next = r$.
Note that $L_1$ is linked to $s$.
Since $s$ is linearized, and $p \in V$, $p.next$ does not change in between when $L_1$ and $s$ are linearized.
Therefore, $s$ changes $p.next$ from $r$ to point to a new \listrec\ whose $next$ pointer points to $r$.
Since $s$ is linearized, Lemma~\ref{lem-if-initiated-rec-not-in-data-structure-then-does-not-change} implies that $p$ must be in the data structure just before $s$ (and when its change occurs).
Since this is the only change that is made by $s$, no \rec\ is removed by $s$, and $new$ points to the only \rec\ that is added by $s$.
Since $R = \emptyset$, the claim holds.

Suppose $s$ occurs at line~\ref{multiset-delete-sct1} or line~\ref{multiset-delete-sct2}.
Then, $fld$ is a pointer to $p.next$, and $new$ is a pointer to a new \listrec.
Before performing $s$, $I$ performs invocations $L_1$ and $L_2$ of \llt$(p)$ and \llt$(r)$, respectively, which each return a value different from \fail, or \finalized.
Note that $L_1$ and $L_2$ are linked to $s$.
Just after performing $L_1$, $I$ sees that $localp.next = r$.
Since $s$ is linearized, and $p \in V$, $p.next$ does not change in between when $L_1$ and $s$ are linearized.
Similarly, since $r \in V$, $r.next$ does not change between when $L_1$ and $s$ are linearized.
Before $s$, $I$ sees $key = r.key$ at line~\ref{multiset-delete-return-false}.
By the precondition of \del, $r.key \neq \infty$.
Thus, inductive Claim~\ref{claim-multiset-sorted-list} (and the fact that keys do not change) implies that $r.next$ points to some \listrec\ $rnext = localr.next$ at all times between when $L_2$ and $s$ are linearized.
We consider two sub-cases.

\textit{Case I:}
$s$ occurs at line~\ref{multiset-delete-sct1}.
Therefore, $s$ changes $p.next$ from $r$ to point to a new \listrec\ whose $next$ pointer points to $rnext$ and, when this change occurs, $r.next$ points to $rnext$.
Since $s$ is linearized, Lemma~\ref{lem-if-initiated-rec-not-in-data-structure-then-does-not-change} implies that $p$ must be in the data structure just before $s$ (and when its change occurs).
Since this is the only change that is made by $s$, $r$ points to the only \rec\ that is removed by $s$, and $new$ points to the only \rec\ that is added by $s$.
Since $R = \{r\}$, the claim holds.

\textit{Case II:}
$s$ occurs at line~\ref{multiset-delete-sct2}.
Since $rnext \in V$, $rnext.next$ does not change between when the \llt\ at line~\ref{multiset-delete-llt-rnext} and $s$ are linearized.
Thus, $rnext.next$ contains the same value $v$ throughout this time.
Therefore, $s$ changes $p.next$ from $r$ to point to a new \listrec\ whose $next$ pointer contains $v$ and, when this change occurs, $p.next$ points to $r$, $r.next$ points to $rnext$, and $rnext.next$ contains $v$.
Since $s$ is linearized, Lemma~\ref{lem-if-initiated-rec-not-in-data-structure-then-does-not-change} implies that $p$ must be in the data structure just before $s$ (and when its change occurs).
Since this is the only change that is made by $s$, $r$ and $rnext$ point to the only \rec s that are removed by $s$, and $new$ points to the only \rec \ that is added by $s$.
Since $R = \{r, next\}$, the claim holds.

\textbf{Proof of Claim~\ref{claim-multiset-sorted-list}.}
This claim can be affected only by a linearized invocation of \sct\ that changes a $next$ pointer.
Suppose $s$ is a linearized invocation of \sct$(V, R, fld, new)$.
Then, $s$ can occur only at line~\ref{multiset-insert-sct2}, \ref{multiset-delete-sct1} or \ref{multiset-delete-sct2}.
We argued in the proof of Claim~\ref{claim-multiset-finalized-before-removed} that, in each of these cases, $s$ changes $p.next$ from $r$ to point to a new \listrec, and that this is the only change that it makes.
Let $I$ be the invocation of \ins\ or \del\ in which $s$ occurs.

Suppose $s$ occurs at line~\ref{multiset-delete-sct1}.
We argued in the proof of Claim~\ref{claim-multiset-finalized-before-removed} that, at all times between when the \llt$(r)$ at line~\ref{multiset-delete-llt-r} and $s$ are linearized, $p.next$ points to $r$ and $r.next$ points to some \listrec\ $rnext$.
Therefore, $new.key = r.key$ and $new.next$ points to $rnext$.
We show $r$ is a \listrec\ (and not the $head$ entry point), and $r.key \neq \infty$.
Since $r.next$ points to a \listrec\ $rnext \neq \nil$, $r \neq \nil$ and $r.key \neq \infty$ (by the inductive hypothesis).
Similarly, since $p.next$ points to $r$, either $r = \nil$ or $r$ is a \listrec, so we are done.
Since $r.next$ points to $rnext$ just before $s$ is linearized, setting $new.next$ to point to $rnext$ does not violate the inductive hypothesis.
Since $p.next$ points to $r$, the inductive hypothesis implies that either $p$ is the $head$ entry point or $p.key < r.key$.
Clearly, setting $p.next$ to point to $new$ does not violate the inductive hypothesis in either case.

Suppose $s$ occurs at line~\ref{multiset-insert-sct2}.
Then, $new.key = key$ and $new.next$ points to $r$.
Before $s$, $I$ invokes \search$(key)$ at line~\ref{multiset-insert-search}, and then sees $key \neq r.key$ at line~\ref{multiset-insert-check-key}.
By inductive Claim~\ref{claim-multiset-search-postconditions}, this invocation of \search\ satisfies its postconditions, which implies that $r$ points to a \listrec\ which satisfies $key < r.key$.
Since \search$(key)$ is invoked only when $key < \infty$ (by inspection of the code and preconditions), $key < \infty$.
Thus, setting $new.next$ to point to $r$ does not violate the inductive hypothesis.
The post conditions of \search\ also imply that either $p$ is a \listrec\ and $p.key < key$ or $p = head$.
Therefore, setting $p.next$ to point to $new$ does not violate the inductive hypothesis.

Suppose $s$ occurs at line~\ref{multiset-delete-sct2}.
We argued in the proof of Claim~\ref{claim-multiset-finalized-before-removed} that, at all times between when the \llt$(r)$ at line~\ref{multiset-delete-llt-r} and $s$ are linearized, $p.next$ points to $r$, $r.next$ points to some \listrec\ $rnext$ (pointed to by $localr.next$) and $rnext.next$ points to some \listrec\ $rnext'$.
Thus, $new.key = rnext.key$ and $new.next$ points to $rnext'$.
Since $rnext.next$ points to a \listrec\ $rnext'$, $rnext.key < \infty$ (by the inductive hypothesis).
Therefore, since $rnext.next$ points to $rnext'$ just before $s$, setting $new.next$ to point to $rnext'$ does not violate the inductive hypothesis.
By the inductive hypothesis, either $p$ is the $head$ entry point, or $p.key < r.key < rnext.key = new.key < \infty$.
Clearly, setting $p.next$ to point to $new$ does not violate the inductive hypothesis in either case.
\end{proof}

\begin{cor} \label{cor-multiset-data-structure-always-sorted-list}
The $head$ entry point always points to a sorted list with strictly increasing keys. 
\end{cor}
\begin{proof}
Immediate from Lemma~\ref{lem-multiset-constraints-invariants}.\ref{claim-multiset-sorted-list}.
\end{proof}

\begin{obs} \label{obs-multiset-satisfies-con-mark-all-removed-recs}
Lemma~\ref{lem-multiset-constraints-invariants}.\ref{claim-multiset-finalized-before-removed} implies Constraint~\ref{con-mark-all-removed-recs}.
\end{obs}

We now argue that the multiset algorithm satisfies a constraint placed on the use of \llt\ and \sct.
This constraint is used to guarantee progress for \sct. 

\begin{obs} \label{obs}
Consider any execution that contains a configuration $C$ after which no field of any \rec\ changes.
There is a total order on all \rec s created during this execution such that, if \rec\ $r_1$ appears before \rec\ $r_2$ in the sequence $V$ passed to an invocation $S$ of \sct\ whose linked \llt s begin after $C$, then $r_1 < r_2$.
\end{obs}
\begin{proof}
Since the \llt s linked to $S$ begin after $C$, it follows immediately from the multiset code that $V$ is a subsequence of nodes in the list.
By Corollary~\ref{cor-multiset-data-structure-always-sorted-list}, they occur in order of strictly increasing keys, so $r_1$ before $r_2$ in $V$ implies $r_1.key < r_2.key$.
Thus, we take the total order on keys to be our total order.
\end{proof}

\begin{defn} \label{defn-multiset-in-data-structure}
The number of occurrences of $key \neq \infty$ \textbf{in the data structure} at time $t$ is $count$ if there is a \rec\ $r$ in the data structure at time $t$ such that $r.key = key$ and $r.count = count$, and zero, otherwise.
\end{defn}

We call an invocation of \ins\ or \del\ \textbf{effective} if it performs a linearized invocation of \sct\ (which either returns \true, or does not terminate).
From the code of \ins\ and \del, each \textbf{effective} invocation of \ins\ or \del\ performs exactly one linearized invocation of \sct, each invocation of \ins\ that returns is effective, and each invocation of \del\ that returns \true\ is effective.
We linearize each effective invocation of \ins\ or \del\ at its linearized invocation of \sct.
The linearization point for an invocation $I$ of \del$(key, count)$ that returns \false\ is subtle.
Suppose $I$ returns \false\ after seeing $r.key \neq key$.
Then, we must linearize it at a time when the nodes $p$ and $r$ returned by its invocation $I'$ of \search\ are both in the data structure and $p.next$ points to $r$.
By Observation~\ref{obs-multiset-satisfies-con-mark-all-removed-recs}, Constraint~\ref{con-mark-all-removed-recs} is satisfied.
This means we can apply Lemma~\ref{lem-if-rec-traversed-then-rec-in-data-structure} to show that there is a time during $I'$ when $p$ is in the data structure and $p.next = r$ (so $r$ is also in the data structure).
We linearize $I$ at the last such time.
Now, suppose $I$ returns \false\ after seeing $r.count < count$.
Then, we must linearize it at a time when the node $r$ returned by its invocation $I'$ of \search\ is both in the data structure, and satisfies $r.count < count$.
As in the previous case, we can apply Lemma~\ref{lem-if-rec-traversed-then-rec-in-data-structure} to show that there is a time after the start of $I'$, and at or before when $I$ reads a value $v$ from $r.count$ at line~\ref{multiset-delete-return-false}, such that $r$ is in the data structure and $r.count = v$.
We linearize $I$ at the last such time.
Similarly, we linearize each \func{Get} at the last time after the start of the \func{Search} in \func{Get}, and at or before when the \func{Get} reads a value $v$ from $r.count$, such that $r$ is in the data structure and $r.count = v$.
Clearly, each operation is linearized during that operation.

\begin{lem} \label{lem}
At all times $t$, the multiset $\sigma$ of keys in the data structure is equal to the multiset $\sigma_L$ of keys that would result from the atomic execution of the sequence of operations linearized up to time $t$.
\end{lem}
\begin{proof}
We prove this claim by induction on the sequence of steps taken in the execution.
Since $next$ pointers and $count$ fields can be changed only by linearized invocations of \sct \ (and $key$ fields do not change), we need only consider linearized invocations of \sct \ when reasoning about $\sigma$.
Thus, invocations of \ins\ and \del\ that are not effective cannot change the data structure.
Since invocations of \func{Get} do not invoke \sct, they cannot change the data structure.
Therefore, we need only consider effective invocations of \ins\ and \del\ when reasoning about $\sigma_L$.
Since each effective invocation of \ins\ or \del\ is linearized at its linearized invocation of \sct, the steps that can affect $\sigma$ and $\sigma_L$ are exactly the same.
\textbf{Base case.}
Before any linearized \sct\ has occurred, no $next$ pointer has been changed.
Thus, the data structure is in its initial configuration, which implies $\sigma = \emptyset$.
Since no effective invocation of \ins\ or \del\ has been linearized, $\sigma_L = \emptyset$.
\textbf{Inductive step.}
Let $s$ be a linearized invocation $S$ of \sct$(V, R, fld, new)$, $I$ be the (effective) invocation of \ins\ or \del\ that performs $S$, and $p$, $r$ and $rnext$ refer to the local variables of $I$.
Suppose $\sigma = \sigma_L$ before $s$.
Let $\sigma'$ denote $\sigma$ after $s$, and $\sigma_L'$ denote $\sigma_L$ after $s$.
We prove $\sigma' = \sigma_L'$.

Suppose $S$ is performed at line~\ref{multiset-insert-sct1}.
Then, $I$ is an invocation of \ins$(key, count)$, and $\sigma_L' = \sigma_L + \{count$ copies of $key\}$.
By Lemma~\ref{lem-multiset-constraints-invariants}.\ref{claim-multiset-finalized-before-removed}, no \rec \ is added or removed by $S$.
Before $I$ performs $S$, $I$ performs an invocation $L$ of $\llt(r)$ linked to $S$ at line~\ref{multiset-insert-llt-r}.
Since $S$ is linearized, no mutable field of $r$ changes between when $L$ and $S$ are linearized.
Therefore, the value $localr.count$ that $L$ reads from $r.count$ is equal to the value of $r.count$ at all times between when $L$ and $S$ are linearized, and line~\ref{multiset-insert-sct1} implies that $S$ changes $r.count$ from $localr.count$ to $localr.count + count$.
Since $S$ is linearized, Lemma~\ref{lem-if-initiated-rec-not-in-data-structure-then-does-not-change} implies that $r$ must be in the data structure just before $S$ is linearized.
By Lemma~\ref{lem-multiset-constraints-invariants}.\ref{claim-multiset-sorted-list}, $r$ is the only \listrec\ in the data structure with key $key$, so $\sigma$ contains exactly $v$ copies of $key$ just before $S$ is linearized.
Since this is the only change made by $S$, $\sigma' = \sigma + \{count$ copies of $key\}$, and the inductive hypothesis implies $\sigma' = \sigma_L'$.

Suppose $S$ is performed at line~\ref{multiset-insert-sct2}.
Then, $I$ is an invocation of \ins$(key, count)$, and $\sigma_L' = \sigma_L + \{count$ copies of $key\}$.
By Lemma~\ref{lem-multiset-constraints-invariants}.\ref{claim-multiset-finalized-before-removed}, no \rec \ is removed by $S$, and only $new$ is added by $S$.
From the code of \ins, $new.key = key$ and $new.count = count$.
Therefore, $\sigma' = \sigma + \{count$ copies of $key\}$, and the inductive hypothesis implies $\sigma' = \sigma_L'$.

Suppose $S$ is performed at line~\ref{multiset-delete-sct1}.
Then, $I$ is an invocation of \del$(key, count)$.
Before $I$ performs $S$, $I$ performs an invocation $L$ of $\llt(r)$ linked to $S$ at line~\ref{multiset-delete-llt-r}.
Since $S$ is linearized, no mutable field of $r$ changes between when $L$ and $S$ are linearized.
Thus, the value $localr.count$ that $L$ reads from $r.count$ is equal to the value of $r.count$ at all times between when $L$ and $S$ are linearized.
This implies that $I$ sees $r.key = key$ and $r.count \ge count$ at line~\ref{multiset-delete-return-false}.
By Lemma~\ref{lem-multiset-constraints-invariants}.\ref{claim-multiset-finalized-before-removed}, $r$ is the only \rec\ removed by $S$, and $new$ is the only \rec\ added by $S$.
By Definition~\ref{defn-rec-in-added-removed}, $r$ must be in the data structure just before $S$ is linearized.
By Lemma~\ref{lem-multiset-constraints-invariants}.\ref{claim-multiset-sorted-list}, $r$ is the only \listrec\ in the data structure with key $key$.
Hence, $\sigma$ contains exactly $localr.count$ copies of $key$ just before $S$ is linearized.
From the code of \del, $new.key = r.key$ and $new.count = localr.count - count$.
Therefore, $\sigma' = \sigma - \{count$ copies of $key\}$.
By the inductive hypothesis, $\sigma = \sigma_L$.
Thus, there are $localr.count \ge count$ copies of $key$ in $\sigma_L$.
Therefore, if $I$ is performed atomically at its linearization point, it will enter the if-block at line~\ref{multiset-delete-check-count-less}, so $\sigma_L' = \sigma_L - \{count$ copies of $key\} = \sigma'$.

Suppose $S$ is performed at line~\ref{multiset-delete-sct2}.
Then, $I$ is an invocation of \del$(key, count)$.
Before $I$ performs $S$, $I$ performs an invocation $L$ of $\llt(r)$ linked to $S$ at line~\ref{multiset-delete-llt-r}.
Since $S$ is linearized, no mutable field of $r$ changes between when $L$ and $S$ are linearized.
Thus, the value $localr.count$ that $L$ reads from $r.count$ is equal to the value of $r.count$ at all times between when $L$ and $S$ are linearized.
This implies that $I$ sees $r.key = key$ and $r.count \ge count$ at line~\ref{multiset-delete-return-false}, and $count \ge r.count$ at line~\ref{multiset-delete-check-count-less}.
Hence, $r.count = count$ at all times between when $L$ and $S$ are linearized.
Let $rnext$ be the \listrec\ pointed to by $I$'s local variable $localr.next$.
(We know $rnext$ is a \listrec, and not \nil, from $r.key = key < \infty$ and Lemma~\ref{lem-multiset-constraints-invariants}.\ref{claim-multiset-sorted-list}.)
After $L$, $I$ performs an invocation $L'$ of $\llt(rnext)$ linked to $S$ at line~\ref{multiset-delete-llt-rnext}.
By the same argument as for $r.count$, the value $v$ that $L'$ reads from $rnext.count$ is equal to the value of $rnext.count$ at all times between when $L'$ and $S$ are linearized.
By Lemma~\ref{lem-multiset-constraints-invariants}.\ref{claim-multiset-finalized-before-removed}, $r$ and $rnext$ are the only \rec s removed by $S$, and $new$ is the only \rec \ added by $S$.
By Definition~\ref{defn-rec-in-added-removed}, $r$ and $rnext$ must be in the data structure just before $S$.
By Lemma~\ref{lem-multiset-constraints-invariants}.\ref{claim-multiset-sorted-list}, $r$ is the only \listrec\ in the data structure with key $key$, and $rnext$ is the only \listrec\ in the data structure with its key.
Hence, $\sigma$ contains exactly $r.count = count$ copies of $key$, and exactly $v$ copies of $rnext.key$.
From the code of \del, $new.key = rnext.key$ and $new.count = rnext.count = v$.
Therefore, $\sigma' = \sigma - \{count$ copies of $key\}$
By the inductive hypothesis, $\sigma = \sigma_L$.
Thus, there are exactly $count$ copies of $key$ just before $I$ in the linearized execution.
From the code of \del, in the linearized execution, $I$ will enter the else block at line~\ref{multiset-delete-check-count-else}, so $\sigma_L' = \sigma_L - \{count$ copies of $key\} = \sigma'$.
\end{proof}

\begin{lem} \label{lem}
Each invocation of \func{Get}$(key)$ that terminates returns the number of occurrences of $key$ in the data structure just before it is linearized.
\end{lem}
\begin{proof}
Consider any invocation $I$ of \func{Get}$(key)$.
Let $I'$ be the invocation of \search$(key)$ performed by \func{Get}$(key)$, and $p$ and $r$ refer to the local variables  of $I'$.
By Lemma~\ref{lem-multiset-constraints-invariants}.\ref{claim-multiset-search-postconditions}, $I'$ satisfies its postcondition, which means that $key \le r.key$, and either $p.key < key$ or $p = head$.
We proceed by cases.
Suppose $key = r.key$.
Then, after $I'$, $I$ reads a value $v$ from $r.count$ and returns $v$.
By Observation~\ref{obs-multiset-satisfies-con-mark-all-removed-recs}, Constraint~\ref{con-mark-all-removed-recs} is satisfied.
By Lemma~\ref{lem-if-rec-traversed-then-rec-in-data-structure}, there is a time after the start of $I'$, and at or before when $I$ reads $r.count$, such that $r$ is in the data structure and $r.count = v$.
$I$ is linearized at the last such time.
By Corollary~\ref{cor-multiset-data-structure-always-sorted-list}, $r$ is the only \rec\ in the list that contains key $key$.
Suppose that either $key < r.key$ and $p = head$, or $key < r.key$ and $p.key < key$.
Then, $I$ returns zero.
By Lemma~\ref{lem-if-rec-traversed-then-rec-in-data-structure}, at sometime during $I'$, $p$ was in the data structure and $p.next$ pointed to $r$.
$I$ is linearized at the last such time.
By Corollary~\ref{cor-multiset-data-structure-always-sorted-list}, the data structure contains no occurrences of $key$ when $I$ is linearized.
\end{proof}

\begin{lem} \label{lem}
Each invocation $I$ of \del$(key, count)$ that terminates returns \true \ if the data structure contains at least $count$ occurrences of $key$ just before $I$ is linearized, and \false \ otherwise.
\end{lem}
\begin{proof}
\textbf{Case I:} $I$ returns \false.
In this case, $I$ satisfies $key \neq r.key$ or $localr.count < count$ at line~\ref{multiset-delete-return-false}.
Suppose $key \neq r.key$.
Then, by the postcondition of \search, $key < r.key$, and either $p.key < key$ or $p = head$.
By Observation~\ref{obs-multiset-satisfies-con-mark-all-removed-recs}, Constraint~\ref{con-mark-all-removed-recs} is satisfied.
By Lemma~\ref{lem-if-rec-traversed-then-rec-in-data-structure}, there is a time during the preceding invocation $I'$ of \search, when $p$ was in the data structure and $p.next$ pointed to $r$.
$I$ is linearized at the last such time.
Corollary~\ref{cor-multiset-data-structure-always-sorted-list} implies that there are no occurrences of $key$ in the data structure when $I$ is linearized.
By the precondition of \del, $count > 0$, so the claim is satisfied. 

Now, suppose $localr.count < count$ at line~\ref{multiset-delete-return-false}.
By Lemma~\ref{lem-if-rec-traversed-then-rec-in-data-structure}, there is a time after the start of $I'$, and before $I$'s \llt$(r)$ reads $localr.count$ from $r.count$, such that $r$ is in the data structure and $r.count = localr.count$.
$I$ is linearized at the last such time.
By Corollary~\ref{cor-multiset-data-structure-always-sorted-list}, $r$ is the only \rec\ in the list that contains key $key$, so there are $r.count < count$ occurrences of $r.key = key$ in the data structure when $I$ is linearized.

\textbf{Case II:} $I$ returns \true.
In this case, $I$ satisfies $key = r.key$ and $localr.count \ge count$ at line~\ref{multiset-delete-return-false}, and $I$ is linearized at an invocation $S$ of \sct\ at line~\ref{multiset-delete-sct1} or \ref{multiset-delete-sct2}.
In each case, Lemma~\ref{lem-multiset-constraints-invariants}.\ref{claim-multiset-finalized-before-removed} implies that $r$ is removed by $S$, so $r$ is in the data structure just before $S$ is linearized.
Hence, $r$ is in the data structure just before $I$ is linearized.
Before $I$ performs $S$, $I$ performs an invocation $L$ of $\llt(r)$ linked to $S$ at line~\ref{multiset-delete-llt-r} that reads $localr.count$ from $r.count$.
Since $S$ is linearized, no mutable field of $r$ changes between when $L$ and $S$ are linearized.
Therefore, the value of $localr.count$ is equal to the value of $r.count$ at all times between when $L$ and $S$ are linearized.
Thus, just before $I$ is linearized, $r$ is in the data structure and $r.count \ge count$.
Finally, Corollary~\ref{cor-multiset-data-structure-always-sorted-list} implies that $r$ is the only \rec\ in the list that contains key $key$, so the claim holds.
\end{proof}

We now prove that our algorithm satisfies an assumption that we made in the paper.

\begin{lem} \label{lem-multiset-only-one-finalized-llt}
No process performs more than one invocation of $\llt(r)$ that returns \finalized, for any \rec\ $r$.
\end{lem}
\begin{proof}
Let $r$ be a \rec.
Suppose, to derive a contradiction, that a process $p$ performs two invocations $L$ and $L'$ of $\llt(r)$ that return \finalized.
Without loss of generality, let $L$ occur before $L'$.
From the code of \ins\ and \del, $p$ must perform an invocation of \search, $L$, another invocation $I$ of \search, and then $L'$.
Since $L$ returns \finalized, it is linearized after an invocation $S$ of \sct$(V, R, fld, new)$ with $r \in R$.
By Lemma~\ref{lem-multiset-constraints-invariants}.\ref{claim-multiset-finalized-before-removed}, $r$ is removed from the data structure by $S$.
We now show that $r$ cannot be added back into the data structure by any subsequent invocation of \sct.
From the code of \ins\ and \del, each invocation of \sct$(V', R', fld', new')$ that changes a $next$ pointer is passed a newly created \listrec, that is not known to any other process, as its $new'$ argument.
This implies that $new'$ is not initiated, and cannot have previously been removed from the data structure.
Therefore, $r$ is not in the data structure at any point during $I$.
By Observation~\ref{obs-multiset-satisfies-con-mark-all-removed-recs}, Constraint~\ref{con-mark-all-removed-recs} is satisfied.
By Lemma~\ref{lem-if-rec-traversed-then-rec-in-data-structure}, $r$ is in the data structure at some point during $I$, which is a contradiction.
\end{proof}

\begin{lem} \label{lem-multiset-progress}
If operations (\ins, \del\ and \func{Get}) are invoked infinitely often, then operations complete infinitely often.
\end{lem}
\begin{proof}
Suppose, to derive a contradiction, that operations are invoked infinitely often but, after some time $t$, no operation completes.
If \sct s are performed infinitely often, then they will succeed infinitely often and, hence, operations will succeed infinitely often.
Thus, there must be some time $t' \ge t$ after which no \sct\ is performed.
Then, after $t'$, the data structure does not change, and only a finite number of nodes with keys different from $\infty$ are ever added to the data structure.
Consider an invocation $I$ of \search$(key)$ that is executing after $t'$.
Each time $I$ performs line~\ref{multiset-search-advance-r}, it reads a \listrec\ $rnext$ from $r.next$, and $rnext.key > r.key$.
Therefore, by Corollary~\ref{cor-multiset-data-structure-always-sorted-list}, $I$ will eventually see $r.key = \infty$ at line~\ref{multiset-search-loop}.
This implies that every invocation of \func{Get} eventually completes.
Therefore, \ins\ and \del\ must be invoked infinitely often after $t'$.
From the code of \ins\ (\del), in each iteration of the while loop, a \search\ is performed, followed by a sequence of \llt s.
If these \llt s all return values different from \fail\ or \finalized, then an invocation of \sct\ is performed.
Since every invocation of \search\ eventually completes, Definition~\ref{defn-set-up-sct} implies that invocations of \sct\ are set up infinitely often.
Thus, invocations of \sct\ succeed infinitely often.
From the code of \ins\ and \del, after performing a successful invocation of \sct, an invocation of \ins\ or \del\ will immediately return.
\end{proof}

\end{document}